\documentclass[letterpaper,12pt]{article}
\usepackage[margin=1in]{geometry}
\usepackage{xcolor, graphicx, appendix}
\usepackage[bookmarks, plainpages = false, colorlinks = true, citecolor = teal, linkcolor = chicago-maroon, anchorcolor = red, urlcolor = chicago-maroon]{hyperref}
\definecolor{chicago-maroon}{RGB}{128,0,0}
\definecolor{northwestern-purple}{RGB}{82,0,99}
\usepackage{url}
\usepackage{amsmath, amsfonts, amsthm, amssymb, bm, bbm, verbatim, dsfont, mathtools, mathrsfs}
\usepackage{etoolbox}
\usepackage{almendra}
\usepackage{authblk}
\usepackage{array}
\usepackage{multirow}
\usepackage{ulem}
\usepackage{cancel}
\usepackage{esvect}
\providecommand{\keywords}[1]{\textbf{Keywords:} #1}

\definecolor{ForestGreen}{RGB}{34,139,34}

\usepackage{float}
\usepackage{pgf,tikz}
\usetikzlibrary{arrows,shapes.arrows,shapes.geometric,shapes.multipart,fit,automata,decorations.pathmorphing,positioning}
\tikzset{
	>=stealth',
	true/.style={
		rectangle,
		draw=black, very thick,
		text width=6.5em,
		minimum height=2em,
		text centered,
		fill=gray, opacity = 0.5},
	punkt/.style={
		rectangle,
		rounded corners,
		draw=black, very thick,
		text width=6.5em,
		minimum height=2em,
		text centered},
	est/.style={
		circle,
		draw=black, very thick,
		text centered},
	shade/.style={
		circle,
		draw=black, very thick, fill=gray!50,
		text centered},
	weight/.style={
		circle,
		draw=black, very thick,
		text width=6.5em,
		minimum height=2em,
		text centered},
	pil/.style={
		->,
		thick,
		shorten <=2pt,
		shorten >=2pt,},
	double/.style={
		<->,
		thick,
		shorten <=2pt,
		shorten >=2pt,},
	dash/.style={
		dashed,
		thick,
		shorten <=2pt,
		shorten >=2pt,},
	dashdouble/.style={
		<->,
		dashed,
		thick,
		shorten <=2pt,
		shorten >=2pt,}
}
\usepackage{tikz-cd}
\makeatletter 
\newcolumntype{C}[1]{>{\centering\arraybackslash}p{#1}}

\usepackage{epstopdf}
\usepackage{fullpage}
\usepackage{algorithm}
\usepackage{algorithmicx}
\usepackage{subcaption}

\usepackage[nameinlink]{cleveref}
\usepackage[numbers,sort&compress]{natbib}

\usepackage[utf8]{inputenc} 
\usepackage[T1]{fontenc}    
\usepackage{url}            
\usepackage{booktabs}     
\usepackage{nicefrac}       
\usepackage{pdfpages}
\usepackage{scalerel}
\usepackage{dashrule}
\usepackage{lmodern}
\usepackage{fancyhdr}
\usepackage{tabu}
\usepackage{enumitem}

\def\IIFF{\mathbb{IF}}

\def\var{\mathsf{var}}
\def\cov{\mathsf{cov}}

\def\bias{\mathsf{bias}}

\def\op{\mathsf{op}}

\renewcommand{\[}{\left[}

\renewcommand{\hat}{\widehat}
\renewcommand{\tilde}{\widetilde}

\usepackage{xr,xspace}
\usepackage{todonotes}

\usepackage{caption,subcaption,soul}
\usepackage{algpseudocode}
\makeatletter
\usepackage{romanbar}

\theoremstyle{plain}
\newtheorem{theorem}{Theorem}
\newtheorem{lemma}{Lemma}
\newtheorem{proposition}{Proposition}

\theoremstyle{definition}

\newtheorem{example}{Example}

\newtheorem{assumption}{Assumption}
\newtheorem{conjecture}{Conjecture}
\newtheorem{problem}{Problem}

\newtheorem{remark}{Remark}
\AtBeginEnvironment{remark}{%
  \pushQED{\qed}%
}
\AtEndEnvironment{remark}{\popQED\endexample}

\newtheorem*{remark*}{Remark}

\usepackage{algorithm}
\usepackage{algpseudocode}

\usepackage{xspace, prettyref}

\newcommand{\diff}{{\mathrm d}}

\newcommand\indep{\protect\mathpalette{\protect\independenT}{\perp}}
\def\independenT#1#2{\mathrel{\rlap{$#1#2$}\mkern2mu{#1#2}}}

\newrefformat{eq}{(\ref{#1})}
\newrefformat{chap}{Chapter~\ref{#1}}
\newrefformat{sec}{Section~\ref{#1}}
\newrefformat{alg}{Algorithm~\ref{#1}}
\newrefformat{fig}{Fig.~\ref{#1}}
\newrefformat{tab}{Table~\ref{#1}}
\newrefformat{rmk}{Remark~\ref{#1}}
\newrefformat{clm}{Claim~\ref{#1}}
\newrefformat{def}{Definition~\ref{#1}}
\newrefformat{cor}{Corollary~\ref{#1}}
\newrefformat{lmm}{Lemma~\ref{#1}}
\newrefformat{prop}{Proposition~\ref{#1}}
\newrefformat{prob}{Problem~\ref{#1}}
\newrefformat{app}{Appendix~\ref{#1}}
\newrefformat{hyp}{Hypothesis~\ref{#1}}
\newrefformat{thm}{Theorem~\ref{#1}}

\newcommand{\diag}{\mathsf{diag}}

\newcommand{\sff}{{\mathsf{f}}}

\newcommand{\calF}{{\mathcal{F}}}

\renewcommand{\tilde}{\widetilde}
\renewcommand{\hat}{\widehat}

\newcommand{\bmu}{{\bm{u}}}

\newcommand{\bbE}{{\mathbb{E}}}

\newcommand{\bbP}{{\mathbb{P}}}

\newcommand{\bbR}{{\mathbb{R}}}

\usepackage{tcolorbox}

\makeatletter
\def\ubar#1{\underline{\sbox\tw@{$#1$}\dp\tw@\z@\box\tw@}}
\makeatother

\def\leftarrowCirc{\hbox{$\leftarrow$}\kern-1.5pt\hbox{$\circ$}}
\def\Circrightarrow{\hbox{$\circ$}\kern-1.5pt\hbox{$\rightarrow$}}
\def\Circleftarrow{\hbox{$\circ$}\kern-1.5pt\hbox{$\leftarrow$}}
\def\rightarrowCirc{\hbox{$\rightarrow$}\kern-1.5pt\hbox{$\circ$}}

\def\aipw{\mathsf{aipw}}

\def\b{\mathrm{b}}

\def\bH{\mathbf{H}}
\def\bI{\mathbf{I}}

\def\bM{\mathbf{M}}

\def\bO{\mathbf{O}}

\def\bX{\mathbf{X}}

\def\bo{\mathbf{o}}

\def\bt{\mathbf{t}}

\def\bv{\mathbf{v}}

\def\bx{\mathbf{x}}
\def\by{\mathbf{y}}

\def\bbeta{\bm{\beta}}

\def\bmu{\bm{\mu}}

\def\bSigma{\bm{\Sigma}}

\usepackage{pifont}
\newcommand{\cmark}{\ding{51}}%
\newcommand{\xmark}{\ding{55}}%

\usepackage{setspace}
\def\bSig\mathbf{\Sigma}

\def\hat{\widehat}

\def\sff{\mathsf{d}}

\def\bias{\mathrm{bias}}
\def\var{\mathrm{var}}
\def\op{\mathrm{op}}
\def\cov{\mathrm{cov}}
\def\diag{\mathrm{diag}}
\def\diff{\mathrm{d}}

\def\bmu{\bm{\mu}}

\def\adj{\mathsf{adj}}
\def\db{\mathsf{db}}
\def\unadj{\mathsf{unadj}}

\def\aipw{\mathsf{aipw}}

\def\plim{\mathrm{plim}}

\def\bSigma{\bm{\Sigma}}

\def\tilde{\widetilde}

\def\trace{\mathsf{tr}}

\usepackage{bibunits}
\usepackage{threeparttable}
\usepackage{appendix} 

\usepackage{makecell}
\usepackage{multirow}

\usepackage{titlesec}
\titlespacing*{\section}{0pt}{0.5\baselineskip}{0.2\baselineskip}
\titlespacing*{\subsection}{0pt}{0.4\baselineskip}{0.15\baselineskip}
\usepackage{orcidlink}

\renewcommand{\b}{\mathrm{b}}

\def\mytitle{\Large Covariate Adjustment in Randomized Experiments Motivated by Higher-Order Influence Functions}

\begin{document}

\title{\mytitle}

\author[1]{\small Sihui Zhao\thanks{E-mail: \href{shzhao0115@sjtu.edu.cn}{shzhao0115@sjtu.edu.cn}}}
\author[2, 3]{\small Xinbo Wang\thanks{E-mail: \href{cinbo_w@sjtu.edu.cn}{cinbo\_w@sjtu.edu.cn}}}
\author[1, 3]{\small Lin Liu\orcidlink{0000-0002-9883-7962}\thanks{E-mail: \href{linliu@sjtu.edu.cn}{linliu@sjtu.edu.cn}}}
\author[4]{\small Xin Zhang\thanks{E-mail: \href{xin.zhang6@pfizer.com}{xin.zhang6@pfizer.com} The last two authors are alphabetically ordered.}}

\affil[1]{\small School of Mathematical Sciences, Institute of Natural Sciences, and MOE-LSC, Shanghai Jiao Tong University}
\affil[2]{\small Department of Bioinformatics and Biostatistics, School of Life Sciences, Shanghai Jiao Tong University}
\affil[3]{\small SJTU-Yale Joint Center for Biostatistics and Data Science, Shanghai Jiao Tong University}
\affil[4]{\small Data Sciences and Analytics, Pfizer Inc}
    
\date{\small\today}
    
\maketitle

\begin{abstract}
{\small Higher-Order Influence Functions (HOIF), developed in a series of papers over the past twenty years, is a fundamental theoretical device for constructing rate-optimal causal-effect estimators from observational studies. However, the value of HOIF for analyzing well-conducted randomized controlled trials (RCT) has not been explicitly explored. In the recent U.S. Food and Drug Administration and European Medicines Agency guidelines on the practice of covariate adjustment in analyzing RCT, in addition to the simple, unadjusted difference-in-mean estimator, it was also recommended to report the estimator adjusting for baseline covariates via a simple parametric working model, such as a linear model. However, when the number of baseline covariates $p$ is large, the recommendation is somewhat murky. In this paper, we show that HOIF-motivated estimators for the treatment-specific mean have significantly improved statistical properties compared to popular adjusted estimators in practice when $p$ is relatively large relative to the sample size $n$. We also characterize the conditions under which the HOIF-motivated estimator improves upon the unadjusted one. More importantly, we demonstrate that several state-of-the-art adjusted estimators proposed recently can be interpreted as particular HOIF-motivated estimators, thereby placing these estimators in a more unified framework. Numerical and empirical studies are conducted to corroborate our theoretical findings. An accompanying \texttt{R} package can be found on \href{https://cran.r-project.org/web/packages/HOIFCar/index.html}{CRAN}.}
\end{abstract}

\keywords{\small Covariate adjustment, Randomized clinical trials, Higher-order influence functions}

\doublespacing
\normalsize

\begin{bibunit}[plainnat]

\section{Introduction}
\label{sec:intro}

Evidence from randomized clinical trials (RCT) is widely regarded as the gold standard for evaluating treatment effects in comparative effectiveness research. Complete randomization, along with its extensions, such as covariate-adaptive randomization and rerandomization \citep{pocock1975sequential, morgan2012rerandomization, ma2024new}, relieves analysts' burden of justifying the unconfoundedness assumption. Furthermore, the true propensity score in RCT is known to data analysts \citep{aronow2025nonparametric}, offering another important advantage over observational studies.

The most straightforward RCT designs include Completely Randomized Experiments (CRE) and Bernoulli sampling, in which treatments are randomly assigned without leveraging any covariate information. Under CRE or Bernoulli sampling, the Difference-in-Mean estimator, referred to as \emph{the unadjusted estimator} henceforth, is unbiased, $\sqrt{n}$-consistent, centered and asymptotically normal (CAN) for the average treatment effect (ATE), where $n$ denotes the sample size. 
However, modern trials typically collect multiple baseline measurements.
It is now widely recognized that, both theoretically and empirically, adjusting for baseline covariates (especially those prognostic factors affecting potential outcomes) in either the design stage or the analysis stage can yield greater efficiency than the unadjusted estimator (see e.g. \citet{yang2001efficiency, zhang2008improving, lin2013agnostic, ma2020statistical, zhao2021covariate, ma2022regression, ye2023toward} and references therein). In this paper, we focus only on the adjusting for baseline covariates in the analysis stage and primarily on CRE. But we will mention parallel results under Bernoulli sampling when doing so facilitates understanding.

Covariate adjustment methods in RCT have a long history in statistics \citep{freedman2008regressiona, freedman2008regressionb}. Adjustment using linear or simple parametric working models \citep{ma2022regression, ye2023toward} has been endorsed in the latest statistical analysis guidelines issued by the U.S. Food and Drug Administration (FDA) \citep{FDA2023}.
The rationale for the improved efficiency through covariate adjustment, even with a likely misspecified linear working model, can be clearly articulated within the superpopulation framework. 
As noted in \citet{richardson2014causal}, the semiparametric theory developed in \citet{robins1994estimation} motivated key works on covariate adjustment in RCT \citep{tsiatis2008covariate, moore2009covariate}.
Specifically, the theory indicates that the Augmented Inverse Probability Weighting (AIPW) estimator, rather than the Inverse Probability Weighting (IPW) estimator, achieves the smallest possible asymptotic variance, more precisely the Semiparametric Variance Bound (SVB), when the outcome model is correctly specified. 
The SVB of the ATE is characterized by its (first-order) influence function, a central concept in semiparametric theory. 
The variance reduction property of the AIPW estimator, constructed based on the influence function, can be explained geometrically: the augmented term subtracted off the IPW estimator can be viewed as a \emph{projection} of the IPW estimator onto a specific subspace, thereby reducing variance through the contracting norm property of the projection. 
This intuition holds even with a misspecified linear model.

In CRE, the unadjusted estimator coincides with the IPW estimator, while the adjusted estimator (see \eqref{adjusted estimator} for its specific form) is algebraically equivalent to the AIPW estimator, using an estimated linear working model for the outcome regression. 
 Thus, the previous geometric reasoning shows variance reduction for the adjusted over unadjusted estimator in fixed $p$ (dimension of baseline covariates), large $n$ regimes. Section~\ref{sec:hoif} elaborates this intuition.

With  technical advancements, modern RCT  routinely collect multidimensional baseline covariates. However, per the recent FDA guideline, when $p$ is large relative to $n$, the best practice to adjust for covariates in the analysis stage becomes murky. The recent literature has therefore gradually turned to the development of methods adjusting for higher-dimensional covariates. Since we primarily focus on the case using linear working models, we highlight some of the most relevant papers. \citet{ma2022regression} and \citet{ye2023toward} demonstrated that the estimator that adjusts for baseline covariates $\bx$ and the interaction between $\bx$ and treatment $t$ in a linear working model using OLS is CAN and more efficient or as efficient as the unadjusted estimator in various designs when $p$ is fixed and $n \rightarrow \infty$, in the superpopulation framework. \citet{jiang2025adjustments} showed that the same results hold for this OLS-based estimator when $p = o (\sqrt{n})$; but when $\sqrt{n} \lesssim p \lesssim n$, a true linear outcome model needs to be further assumed, an assumption often considered unduly strong in RCTs.
\citet{lei2021regression} also considered covariate adjustment in a linear working model, under CRE and the design-based (or randomization-based) framework, which allows $p = O (n^{2 / 3})$ up to log-factors. Unlike \citet{jiang2025adjustments}, they do not assume that the linear working model is correctly specified when $p \gtrsim \sqrt{n}$. More recently,
\citet{lu2025debiased} devised a debiased estimator that improves on that of \citet{lei2021regression}: their estimator is never less efficient than, and can sometimes be more efficient than the unadjusted estimator when $p=o(n)$. \citet{lu2025debiased} note that \citet{chang2024exact} constructed a similar but exactly unbiased estimator earlier, but only with theoretical results for fixed $p$.

Over the past two decades, higher-order influence functions (HOIF) have been developed as a generalization of classical semiparametric theory to construct rate-optimal estimators for parameters like the ATE from observational data \citep{robins2008higher}. 
Given the significant role of classical semiparametric theory in covariate adjustment, a natural inquiry arises: \emph{can the HOIF of the ATE also inform the development of covariate adjustment methods in RCT?} In this paper, we provide an affirmative answer to this question. 
For ease of presentation, we focus solely on the treatment-specific mean in the treatment arm (and, by symmetry, the control arm) instead of the ATE. 
In our accompanying \texttt{R} package available on \href{https://cran.r-project.org/web/packages/HOIFCar/index.html}{CRAN}, the corresponding point and interval estimators of ATE are also available.

\subsection*{Main contributions and organization}

The main contributions of our paper are summarized below:
\begin{enumerate}[leftmargin=0.5cm,topsep=0.25pt]
\item 
We demonstrate that a HOIF-motivated adjusted estimator has improved statistical properties over standard adjusted or unadjusted estimators. 
The theoretical analyses 
are ``conceptually simple'', involving only elementary calculations. We also develop a variance estimator of this 
estimator, deferred to Appendix~\ref{app:variance estimators} due to space limitation. This new variance estimator is recommended because of the improved coverage when $n$ is small, based on empirical observations from our simulation studies.


\item Our impression is that the HOIF theory remains elusive even among statisticians, so, in Section~\ref{sec:hoif}, 
we review it in an accessible manner to garner more interest from statisticians. The more important point we hope to make is not to propose a new estimator. Instead,
we show that several state-of-the-art adjusted estimators mentioned above are specific HOIF-motivated estimators, placing them in a more unified framework; see Table~\ref{t:one}.

\item We develop an accompanying user-friendly \texttt{R} package that is available on \href{https://cran.r-project.org/web/packages/HOIFCar/index.html}{CRAN}, which delivers both point and interval estimators for ATE and treatment/control-specific means.
\end{enumerate}

The rest of the paper is organized as follows. Section \ref{sec:setup} introduces the basic setup. Section~\ref{sec:main results} presents a variety of HOIF-motivated adjusted estimators and their statistical properties (Sections~\ref{sec:hoif results} and \ref{sec:variety}), and also draws connections to other state-of-the-art estimators (Section~\ref{sec:understanding}). Simulation studies and real data analysis are carried out in Section~\ref{sec:simulations}. We conclude the paper in Section \ref{sec:conclusions}. The Appendix contains supplementary technical and empirical results.

\section{Notation and Basic Setup}
\label{sec:setup}

\subsection*{Notation}

In this paper, we denote sample size as $n$ and  baseline covariates dimension as $p$. 
Design-based quantities have superscript ``$\sff$'' to distinguish from superpopulation frameworks; e.g., $\bbE^{\sff}$, $\bias^{\sff}$, and $\var^{\sff}$ for expectation, bias, variance. Superscripts are omitted when unambiguous.
We also adopt the common asymptotic and stochastic asymptotic notation, including $O (\cdot)$, $o (\cdot)$, $O_{\bbP} (\cdot)$, $o_{\bbP} (\cdot)$, 
with $\bbP$ the true distribution.
For square matrix $\bM$, $\trace(\bM)$ is the trace, $\Vert \bM \Vert_{\op}$ the operator norm, and $\bM^{-}$ the inverse or pseudoinverse. Vector norm $\Vert \bv \Vert$ is the $\ell_{2}$-norm, and $i \in [n]$ denotes $i=1,\dots,n$.

For the $n \times p$ covariate matrix $\bX \coloneqq (\bx_{1}, \cdots, \bx_{n})^{\top}$ with $\bx_{i} \in \bbR^{p}$ for $i\in [n]$, let $\bar{\bx} \coloneqq n^{-1} \sum_{i = 1}^{n} \bx_{i} \in \bbR^{p}$ 
be the row-wise average vector. Following standard practice in covariate adjustment in randomized experiments, we center the covariate/design matrix by $\bar{\bx}$ 
to obtain $\bX_{c} \coloneqq \left( \bx_{1} - \bar{\bx}, \cdots, \bx_{n} - \bar{\bx} \right)^{\top}$. 
We then define the $n \times n$ ``hat'' projection matrix 
$\bH \coloneqq \bX_{c} \hat{\bSigma}^{-} \bX_{c}^{\top} \equiv \left( H_{i, j}, 1 \leq i, j \leq n \right), \text{ where } \hat{\bSigma} \coloneqq \bX_{c}^{\top} \bX_{c}.$

In the random design setting of the superpopulation framework, 
denote $\bmu \coloneqq \bbE \bx$ and $\bSigma \coloneqq n \bbE (\bx - \bmu) (\bx - \bmu)^{\top}$. 
For any vector $\bv = (v_{1}, \cdots, v_{n})^{\top}$ of length $n$, define 
\begin{align*}
V_{n} (v) \coloneqq \frac{1}{n} \sum_{i = 1}^{n} v_{i}^{2} - \frac{1}{n (n - 1)} \sum_{1 \leq i \neq j \leq n} v_{i} v_{j} \equiv \frac{1}{n - 1} \sum_{i = 1}^{n} \left( v_{i} - \bar{v} \right)^{2},
\end{align*}
as the sample variance of $\bv$, where $\bar{v} \coloneqq n^{-1} \sum_{i = 1}^{n} v_{i}$.

\subsection*{Basic Setup}

Throughout this paper, we  
observe the  data matrix: $\bO \in \bbR^{n \times (p + 2)} \coloneqq \left( \bo_{1}, \cdots, \bo_{n} \right)^{\top}$, where $ \bo_{i} \coloneqq \left( \bx_{i}^{\top}, t_{i}, y_{i} \right)^{\top} \in \bbR^{p + 2}$ for $i \in [n]$, with $t$ and $y$ denoting the treatment indicator and the outcome.
Let $\bt \coloneqq (t_{1}, \cdots, t_{n})^{\top}$ and $\by \coloneqq (y_{1}, \cdots, y_{n})^{\top}$. We always assume that $p < n$ and $\lim_{n \rightarrow \infty} p / n = \alpha \in [0, 1)$
without further mentioning this assumption. Our result thus covers the relatively more challenging ``proportional asymptotic'' regime at least for $p < n$. Without loss of generality, we take $t \in \{0, 1\}$ and $y \in \bbR$. In RCTs, $\bt$ is determined by exogenous randomization, hence under the investigator's control. In particular, we mainly consider the CRE, which assign $n_{1}$ out of $n$ subjects uniformly into the treatment group ($t = 1$), and the remaining $n_{0}$ subjects into the control group. Let $\pi_{1} \coloneqq n_{1} / n$ and $\pi_{0} \coloneqq 1 - \pi_{1}$ 
be the treatment and control proportions respectively.

Denoting the potential outcome vector under treatment as $\by (1) = (y_{1} (1), \cdots, y_{n} (1))^{\top}$ and under control as $\by (0) = (y_{1} (0), \cdots, y_{n} (0))^{\top}$. By the standard consistency assumption, the observed outcome and potential outcomes are connected by $y \equiv t y (1) + (1 - t) y (0)$. In RCT, randomization licenses the use of observables $\bO$ to identify certain causal quantities defined via potential outcomes. In this paper, we assume :
\begin{assumption}[Randomization]
\label{as:randomization}
$\bt$ is assigned via CRE or the Bernoulli sampling, so $\bt \indep \{\bx, \by (0), \by (1)\}$.
\end{assumption}

For example, if one is interested in the treatment-specific mean $\bar{\tau} \coloneqq \bbE^{\sff} y (1) \equiv \frac{1}{n} \sum_{i = 1}^{n} y_{i}(1)$,
one can identify $\bar{\tau}_{1}$ via
$\hat{\tau}_{\unadj} \coloneqq \frac{1}{n_{1}} \sum_{i = 1}^{n} t_{i} y_{i} \equiv \frac{1}{n} \sum_{i = 1}^{n} \frac{t_{i}}{\pi_{1}} y_{i}.$
$\hat{\tau}_{\unadj}$ is often called \textit{the unadjusted estimator} or \textit{the IPW estimator}. It is easy to see that $\hat{\tau}_{\unadj}$ is unbiased for $\bar{\tau}_{1}$. One can similarly define the design-based control specific mean and the ATE, together with their corresponding unadjusted or IPW estimators. Without loss of generality, we only consider the treatment specific mean $\bar{\tau}$ for two reasons. First, all the results hold for the control specific mean by symmetry. Second, treatment and control specific means are more primitive parameters than the ATE.

Occasionally, we also consider the superpopulation framework, under which the observed data is drawn i.i.d. from a common probability distribution $\bbP$:
$\bo_{i} \overset{\rm i.i.d.}{\sim} \bbP, \,\,\,\, i \in [n].$
In the superpopulation framework, we only consider the Bernoulli sampling of the treatment assignment vector, i.e. $\{t_{i}\}_{i = 1}^{n} \overset{\rm i.i.d.}{\sim} \mathrm{Bernoulli} (\pi_{1})$. The superpopulation treatment specific mean is denoted as $\tau \coloneqq \bbE y (1)$.

\section{HOIF-Motivated Covariate Adjustment: Statistical Intuition}
\label{sec:hoif}

Our main theoretical results (in Section \ref{sec:main results}) are stated under the design-based framework. However, in this section, we first explain the main intuition of using HOIF to construct adjusted estimator of $\tau$ under the superpopulation framework. In our own opinion, for most parts, the \emph{statistical intuition} gathered from the superpopulation framework can be carried over to the design-based framework.
While familiar to HOIF experts, this section aims to interest practitioners in RCT analysis.

Guarded by randomization, the unadjusted estimator $\hat{\tau}_{\unadj}$ has already fulfilled the following desiderata:
\begin{itemize}[leftmargin=0.5cm,topsep=0.25pt]
\item It is model-free, unbiased and has variance of order $1 / n$ under certain 
regularity conditions on $\bX$ and $\by$ (see Assumption \ref{as:regularity conditions} later);

\item It is CAN and the variance is easy to estimate.
\end{itemize}

However, $\hat{\tau}_{\unadj}$ 
fails to leverage the information of $\bx$.
One convincing argument for using $\bx$ comes from \textit{semiparametric theory}. It says that the variance of the following random variable, referred to as the (efficient) first-order influence function of $\tau$, characterizes the SVB of any Regular and Asymptotic Linear estimator of $\tau$:
\begin{equation}
\label{EIF}
\dot{\tau}_{1, \tau} \equiv \dot{\tau}_{1, \tau} (\bo) \coloneqq \frac{t}{\pi_{1}} y - \left( \frac{t}{\pi_{1}} - 1 \right) \b (\bx) - \tau,
\end{equation}
where $\b (\cdot) \coloneqq \bbE (y | \bx = \cdot, t = 1)$. 
This motivates the AIPW estimator:
\begin{equation}
\label{aipw}
\hat{\tau}_{\aipw} \coloneqq \frac{1}{n} \sum_{i = 1}^{n} \frac{t_{i}}{\pi_{1}} y_{i} - \left( \frac{t_{i}}{\pi_{1}} - 1 \right) \hat{\b} (\bx_{i}),
\end{equation}
where $\hat{\b}$ estimates $\b$ 
using parametric models or machine learning algorithms \citep{bannick2025general}. 
If $\hat{\b}$ is consistent for $\b$ and $\bbP$-Donsker (intuitively speaking, sufficiently stable), 
$\hat{\tau}_{\aipw}$ attains the SVB asymptotically.

However, since $\b$ is not under the investigator's control, $\hat{\tau}_{\aipw}$ may 
inflate the asymptotic variance when $\plim_{n \rightarrow \infty} \hat{\b} \neq \b$. 
Fortunately, covariate adjustment via a possibly misspecified linear working model with the least square estimator still guarantees efficiency improvement over $\hat{\tau}_{\unadj}$ when $p$ is fixed \citep{lin2013agnostic, ma2022regression, ye2023toward}. An oracle version of 
this adjusted estimator is:
\begin{equation}
\label{adj}
\tilde{\tau}_{\adj} \coloneqq \frac{1}{n} \sum_{i = 1}^{n} \frac{t_{i}}{\pi_{1}} y_{i} - \left( \frac{t_{i}}{\pi_{1}} - 1 \right) (\bx_{i} - \bmu)^{\top} \bbeta,
\end{equation}
where $\bbeta \coloneqq n \bSigma^{-} \cdot \bbE ((\bx - \bmu) \frac{t}{\pi_{1}} (y - \tau))$ is the population projection of 
$y (1) - \tau$ onto the linear span of 
$\bx - \bmu$. 
$\tilde{\tau}_{\adj}$ has the same form as the AIPW estimator with linear outcome regression. 
The following result is immediate and well known \citep{robins1994estimation}. 
\begin{lemma}
\label{lem:svb}
Under Assumption~\ref{as:randomization}, we have $\var (\hat{\tau}_{\unadj}) = \dfrac{1}{n} \var \left\{ \dfrac{t}{\pi_{1}} y \right\}$ and
\begin{align*}
\var (\tilde{\tau}_{\adj}) = \frac{1}{n} \left( \var \left\{ \frac{t}{\pi_{1}} y \right\} - \var \left\{ \left( \frac{t}{\pi_{1}} - 1 \right) (\bx - \bmu)^{\top} \bbeta \right\} \right) \leq \var (\hat{\tau}_{\unadj}).
\end{align*}
If $(\bx_{i} - \bmu)^{\top} \bbeta$ in $\tilde{\tau}_{\adj}$ is replaced by the true outcome regression function $\b (\cdot)$, the variance of $\tilde{\tau}_{\adj}$ attains the \normalfont{SVB}.
\end{lemma}

See Appendix \ref{app:lemma1} for the proof of Lemma \ref{lem:svb}.
When the dimension $p$ of the baseline covariates is small compared to the sample size $n$, one can compute a feasible estimator $\hat{\tau}_{\adj, 1}^{\dag}$ by estimating $\bbeta$ with ordinary least squares (OLS) between $t (y - \bar{\tau}) / \pi_{1}$ and $\bx - \bmu$ with $\bmu$ replaced by $\bar{\bx}$ and $\bar{\tau}$ replaced by $\hat{\tau}_{\unadj}$:
\begin{equation}
\label{adjusted estimator}
\begin{split}
\hat{\bbeta}_{c} & \coloneqq \left( \sum_{l = 1}^{n} (\bx_{l} - \bar{\bx}) (\bx_{l} - \bar{\bx})^{\top} \right)^{-1} \sum_{j = 1}^{n} (\bx_{j} - \bar{\bx}) \frac{t_{j} (y_{j} - \hat{\tau}_{\unadj})}{\pi_{1}}, \\
\hat{\tau}_{\adj, 1}^{\dag} & \coloneqq \hat{\tau}_{\unadj} - \frac{1}{n} \sum_{i = 1}^{n} \left( \frac{t_{i}}{\pi_{1}} - 1 \right) (\bx_{i} - \bar{\bx})^{\top} \hat{\bbeta}_{c} \equiv \hat{\tau}_{\unadj} - \frac{1}{n} \sum_{i = 1}^{n} \sum_{j = 1}^{n} \left( \frac{t_{i}}{\pi_{1}} - 1 \right) H_{i, j} \frac{t_{j} (y_{j} - \hat{\tau}_{\unadj})}{\pi_{1}}.
\end{split}
\end{equation}
We also define an alternative adjusted estimator not centering $y$ that will appear later:
\begin{equation}
\label{adjusted estimator 1}
\hat{\tau}_{\adj, 1} \coloneqq \hat{\tau}_{\unadj} - \frac{1}{n} \sum_{i = 1}^{n} \sum_{j = 1}^{n} \left( \frac{t_{i}}{\pi_{1}} - 1 \right) H_{i, j} \frac{t_{j} y_{j}}{\pi_{1}}.
\end{equation}
Written in the form of \eqref{adjusted estimator} or \eqref{adjusted estimator 1}, one can view the adjusted estimator by linear working models as augmenting the unadjusted estimator with a second-order $V$-statistic. The corresponding least square regression coefficients $\hat{\bbeta}$ is defined similarly to $\hat{\bbeta}_{c}$ except for not centering $y$. To directly see the potential negative impact of the augmented $V$-statistic, its mean is, under the design-based framework and CRE,
$- \frac{\pi_{0}}{\pi_{1}} \frac{n - 2}{n (n - 1)} \sum_{i = 1}^{n} H_{i, i} y_{i} (1) = O \left( p / n \right),$
under certain regularity conditions (e.g. Assumption \ref{as:regularity conditions} later) on $\bX$ and $\by$. The above derivation explains why the usual adjusted estimator $\hat{\tau}_{\adj, 1}^{\dag}$ or $\hat{\tau}_{\adj, 1}$ may hurt statistical inference when $p$ is close to $n$. The bias of $\hat{\tau}_{\adj, 1}^{\dag}$ can be similarly shown to also be of order $O (p / n)$.

\subsection*{The Role of HOIF and a Review}

To explain how the theory of HOIF directly leads to an improved estimator, we first notice that $\hat{\tau}_{\aipw}$ is unbiased and also assume that $\hat{\b}$ is ``sufficiently independent'' from the sample $\bO$. In the derivation below, we might write 0 redundantly as $\frac{1}{\pi_{1}}  - \frac{1}{\pi_{1}}$:
\begin{equation*}
\label{bias}
\bias (\hat{\tau}_{\aipw}) = \bbE \left[ t \left( \frac{1}{\pi_{1}} - \frac{1}{\pi_{1}} \right) \left( \b (\bx) - \hat{\b} (\bx) \right) \right] \equiv \bbE \left[ \pi_{1}^{\frac{1}{2}} \left( \frac{1}{\pi_{1}} - \frac{1}{\pi_{1}} \right) \pi_{1}^{\frac{1}{2}} \left( \b (\bx) - \hat{\b} (\bx) \right) \right].
\end{equation*}
The HOIF theory, in a nutshell, is to approximate the above bias of $\hat{\tau}_{\aipw}$ by first choosing a set of $k$-dimensional transformations of $\bx$, $\bar{\phi}_{k}(\bx) = (\phi_{1} (\bx), \cdots, \phi_{k} (\bx))^{\top}$. $\bar{\phi}_{k}$ is often chosen by some background knowledge on the space the residual $\b - \hat{\b}$ may lie in. Here we simply take $\phi$ as $\bar{\phi}_{k}(\bx) \equiv \bx - \bmu$. Next, we project the two residuals above, $\mathsf{res}_{1} \coloneqq \pi_{1}^{1 / 2} \left( \frac{1}{\pi_{1}} - \frac{1}{\pi_{1}} \right)$ and $\mathsf{res}_{2} \coloneqq \pi_{1}^{1 / 2} \left( \b (\bx) - \hat{\b} (\bx) \right)$ onto the linear space spanned by $\pi_{1}^{1 / 2} (\bx - \bmu)$:
\begin{align*}
\widetilde{\mathsf{res}}_{1} & \coloneqq \pi_{1}^{1 / 2} (\bx - \bmu)^{\top} \left\{ \bbE [\pi_{1} (\bx - \bmu) (\bx - \bmu)^{\top}] \right\}^{-1} \bbE \left[ \pi_{1} (\bx - \bmu) \left( \frac{1}{\pi_{1}} - \frac{1}{\pi_{1}} \right) \right], \\
\widetilde{\mathsf{res}}_{2} & \coloneqq \pi_{1}^{1 / 2} (\bx - \bmu)^{\top} \left\{ \bbE [\pi_{1} (\bx - \bmu) (\bx - \bmu)^{\top}] \right\}^{-1} \bbE \left[ \pi_{1} (\bx - \bmu) (\b (\bx) - \hat{\b} (\bx)) \right].
\end{align*}
The weight $\pi_{1}^{1 / 2}$ is chosen to ensure that the unknown $\b$ appeared in the above expectation can be replaced by the observed $y$. Armed with the above projections of the residuals, we can decompose $\bias (\hat{\tau}_{\aipw})$ into two components by Pythagorean theorem:
\begin{equation*}
\begin{split}
\bias (\hat{\tau}_{\aipw}) \equiv & \ \bbE \left[ \mathsf{res}_{1} \cdot \mathsf{res}_{2} \right] = \widetilde{\bias} (\hat{\tau}_{\aipw}) + \widetilde{\bias}^{\perp} (\hat{\tau}_{\aipw}) \equiv 0, \text{ where } \\
\widetilde{\bias} (\hat{\tau}_{\aipw}) \coloneqq & \ \bbE \left[ \widetilde{\mathsf{res}}_{1} \cdot \widetilde{\mathsf{res}}_{2} \right] \equiv 0 \\
\equiv & \ \bbE \left[ \pi_{1} \left( \frac{1}{\pi_{1}} - \frac{1}{\pi_{1}} \right) (\bx - \bmu)^{\top} \right] \left\{ \bbE [\pi_{1} (\bx - \bmu) (\bx - \bmu)^{\top}] \right\}^{-1} \bbE \left[ (\bx - \bmu) \pi_{1} (\b (\bx) - \hat{\b} (\bx)) \right] \\
= & \ \bbE \left[ \left( \frac{t}{\pi_{1}} - 1 \right) (\bx - \bmu)^{\top} \right] \left\{ \bbE \left[ (\bx - \bmu) (\bx - \bmu)^{\top} \right] \right\}^{-1} \bbE \left[ (\bx - \bmu) \frac{t (y - \hat{\b} (\bx))}{\pi_{1}} \right].
\end{split}
\end{equation*}

Setting $\hat{\b} \equiv 0$, so $\hat{\tau}_{\unadj} \equiv \hat{\tau}_{\aipw}$, the gist of the HOIF theory is to estimate the ``projected bias'' $\widetilde{\bias} (\hat{\tau}_{\aipw}) \equiv \widetilde{\bias} (\hat{\tau}_{\unadj})$ by a \emph{second-order $U$-statistic} as follows:
\begin{align}
\widehat{\IIFF}_{\unadj, 2, 2} & \coloneqq \frac{1}{n (n - 1)} \sum_{1 \leq i \neq j \leq n} \left( \frac{t_{i}}{\pi_{1}} - 1 \right) (\bx_{i} - \bar{\bx})^{\top} \left\{ \frac{1}{n - 1} \sum_{l = 1}^{n} (\bx_{l} - \bar{\bx}) (\bx_{l} - \bar{\bx})^{\top} \right\}^{-1} (\bx_{j} - \bar{\bx}) \frac{t_{j} y_{j}}{\pi_{1}} \nonumber \\
& \equiv \frac{1}{n} \sum_{1 \leq i \neq j \leq n} \left( \frac{t_{i}}{\pi_{1}} - 1 \right) H_{i, j} \frac{t_{j} y_{j}}{\pi_{1}}. \label{IF22}
\end{align}
Here we adopt the $\widehat{\IIFF}$ notation first introduced in \citet{robins2008higher} because $\widehat{\IIFF}_{\unadj, 2, 2}$ is an estimator of the \emph{second-order influence function} of $\tilde{\bias} (\hat{\tau}_{\unadj})$.

With $\widehat{\IIFF}_{\unadj, 2, 2}$, one can construct the following covariate adjusted estimator:
\begin{equation}
\label{our estimator}
\hat{\tau}_{\adj, 2} \coloneqq \hat{\tau}_{\unadj} - \widehat{\IIFF}_{\unadj, 2, 2},
\end{equation}
which is the main estimator that we study in Section \ref{sec:main results}. Though $\widehat{\IIFF}_{\unadj, 2, 2}$ and $\hat{\tau}_{\adj, 2}$ are motivated under the superpopulation framework, the way we tacitly estimate the precision matrix $\left\{ \bbE \left[ (\bx - \bmu) (\bx - \bmu)^{\top} \right] \right\}^{-1}$ in $\widehat{\IIFF}_{\unadj, 2, 2}$ happens to be the ``correct'' choice under the design-based framework, as we condition on $\bX$. It is also worth noting that $\hat{\tau}_{\adj, 2}$ can be viewed as a ``diagonal/trace-free'' version of $\hat{\tau}_{\adj, 1}$. 

As for the variance of $\hat{\tau}_{\adj, 2}$, following well-established statistical theory of HOIF estimators, we have $\var (\hat{\tau}_{\adj, 2}) = O \left( \frac{1}{n} + \frac{p}{n^{2}} \right)$, where the first factor $1 / n$ is attributed to $\hat{\tau}_{\unadj}$ and the second factor $p / n^{2}$ comes from the second-order $U$-statistic $\widehat{\IIFF}_{\unadj, 2, 2}$. When $p = o (n)$, covariate adjustment by $\hat{\tau}_{\adj, 2}$ never hurts the asymptotic variance (in fact, in Theorem~\ref{thm:HOIF-CRE, design-based}, $\hat{\tau}_{\adj, 2}$ may reduce the asymptotic variance if $p = o (n)$). This follows from the ``guiding principle'' \citep{liu2017semiparametric} that if either the propensity score or the outcome regression is consistently estimated, correcting bias by adding the second-order $U$-statistic does not inflate the asymptotic variance.

Even if $p = O (n)$, the variance 
maintains the parametric $1 / n$ rate. 
However, demonstrating that  $\hat{\tau}_{\adj, 2}$ actually improves efficiency for $p = O (n)$ requires more careful analysis, deferred to Section \ref{sec:main results}. 
Intuitively, since $\widehat{\IIFF}_{\unadj, 2, 2}$ estimates a $L_{2} (\bbP)$-projection of $\frac{t}{\pi_{1}} y$, efficiency gains are possible even when $p$ is close to $n$.
This fact has been known and explicitly noted in \citet{liu2020nearly, liu2023hoif}. We hope that this review stimulates more interest in the HOIF theory by those developing statistical methods on randomized experiments.

\begin{remark}
\label{rem:diff}
We briefly compare $\widehat{\IIFF}_{\unadj, 2, 2}$ with estimators from \cite{liu2020nearly, liu2020rejoinder, liu2023hoif}.
Their original proposal addresses observational studies with unknown propensity scores, using $(\bx_{i} - \bar{\bx})^{\top} \hat{\bSigma}_{1}^{-1} (\bx_{j} - \bar{\bx})$ instead of $H_{i, j}$ to protect against model misspecification, where
\begin{equation}
\label{tSigma}
\hat{\bSigma}_{1} \coloneqq \sum_{l = 1}^{n} \frac{t_{l}}{\pi_{1}} (\bx_{l} - \bar{\bx}) (\bx_{l} - \bar{\bx})^{\top}
\end{equation}
Since we consider CRE, the form in \eqref{our estimator} is preferable. 
In fact, it is much more difficult to analyze the statistical properties of $\hat{\tau}_{\adj, 2}$ under observational studies and the superpopulation framework. There one has to control the difference between the sample and population precision matrices. Without imposing strong structural assumptions on the propensity score and outcome regression, higher-order $U$-statistics are needed to remove the bias due to estimating the precision matrix of a very large dimension. Finally, we discuss the choice of the transformation $\bar{\phi}_{k}$, largely ignored above. 
Choosing $\bar{\phi}_{k}(\bx) = \bx - \bmu$ makes $\hat{\tau}_{\adj, 2}$ correct the \emph{own observation bias} of $\hat{\tau}_{\adj, 1}$.
If the analyst has better knowledge of the $\bx \rightarrow y (1)$ mechanism, selecting nonlinear transformations of  $\bx$ better reflect this mechanism may further improve efficiency.
\end{remark}

\section{HOIF-Motivated Covariate Adjustment: Theoretical Results}
\label{sec:main results}

To state our main theoretical results on $\hat{\tau}_{\adj, 2}$, we need to further impose certain regularity conditions on the observed data $\bO$. It is worth noting that to make our exposition more accessible, we choose the following easier-to-interpret regularity conditions on the data instead of the mathematically weaker conditions considered in the mainstream literature on design-based inference \citep{lei2021regression, lu2025debiased}.
\begin{assumption}[Regularity conditions on $\bO$]
\label{as:regularity conditions}
The following regularity condition is occasionally imposed on the observed data $\bO$: there exists an $n$-independent universal constant $B > 0$ such that
$\max \left\{ \Vert y (1) \Vert_{\infty}, \frac{n}{p} \Vert H \Vert_{\infty} \right\} \leq B,$
and $\hat{\bSigma} \equiv \sum_{i = 1}^{n} (\bx_{i} - \bar{\bx}) (\bx_{i} - \bar{\bx})^{\top}$ is invertible.
\end{assumption}

\begin{remark}
\label{rem:gadget}
As will be seen in Appendix \ref{app:gadgets}, under Assumption \ref{as:regularity conditions}, by directly looking at the decomposition of variance formula into a sum of various ``gadgets'', one can immediately tell the bias or the variance order (should be no greater than $O (1 / n)$) after covariate adjustment. In Appendix \ref{app:gadgets}, we will explain that the statistical orders of the ``gadgets'' can be easily deduced based on simple statistical intuition.
\end{remark}

\subsection{Statistical Properties of HOIF-Motivated Estimators in RCT}
\label{sec:hoif results}

We state our main theoretical results under CRE in the theorem below. The corresponding results for the Bernoulli sampling will be commented in Remark \ref{rem:Bernoulli} that follows.

\begin{theorem}
\label{thm:HOIF-CRE, design-based}
Under CRE (Assumption~\ref{as:randomization}) and the design-based framework, we have the following theoretical guarantees on the HOIF estimator $\hat{\tau}_{\adj, 2}$:
\begin{enumerate}[leftmargin=1cm,topsep=0.25pt]
\item The bias of $\hat{\tau}_{\adj, 2}$ has the following form:
\begin{equation}
\label{bias, HOIF-CRE, design-based}
\begin{split}
\bias^{\sff} \left( \hat{\tau}_{\adj, 2} \right) & = \frac{\pi_{0}}{\pi_{1}} \frac{1}{n (n - 1)} \sum_{1 \leq i \neq j \leq n} H_{i, j} y_{j} (1) = - \frac{\pi_{0}}{\pi_{1}} \frac{1}{n (n - 1)} \sum_{i = 1}^{n} H_{i, i} y_{i} (1).
\end{split}
\end{equation}
In addition, if Assumption \ref{as:regularity conditions} holds, then $\bias^{\sff} \left( \hat{\tau}_{\adj, 2} \right) = O \left( \frac{\pi_{0}}{\pi_{1}} \frac{\alpha}{n} \right).$

\item The variance of $\hat{\tau}_{\adj, 2}$ has the following form: Under Assumption \ref{as:regularity conditions},
\begin{align}
& \ \var^{\sff} \left( \hat{\tau}_{\adj, 2} \right) = \nu^{\sff} + \mathsf{Rem} = O \left( \frac{1}{n} \left\{ 1 + \frac{p}{n} \right\} \right), \label{var, HOIF-CRE, design-based}
\end{align}
where the exact form of $\mathsf{Rem} = o (1 / n)$ can be deduced from the proof of this claim in Appendix \ref{app:hoif results} and the main term $\nu^{\sff}$ reads as follows:
\begin{equation}
\label{the variance}
\begin{split}
\nu^{\sff} \coloneqq & \ \underbrace{\left( \frac{\pi_{0}}{\pi_{1}} \right) \frac{1}{n} V_{n} \left[ y_{i} (1) - \sum_{j \neq i} H_{j, i} y_{j} (1) \right]}_{\eqqcolon \, \nu^{\sff}_{1} \, = \, O \left( \frac{1}{n} \right)} \\
& + \underbrace{\left( \frac{\pi_{0}}{\pi_{1}} \right)^{2} \frac{1}{n} \left\{ \frac{1}{n} \sum_{i = 1}^{n} H_{i, i} (1 - H_{i, i}) y_{i} (1)^{2}  + \frac{1}{n} \sum_{1 \leq i \neq j \leq n} H_{i, j}^{2} y_{i} (1) y_{j} (1) \right\}}_{\eqqcolon \, \nu^{\sff}_{2} \, = \, O \left( \frac{1}{n} \frac{p}{n} \right)}.
\end{split}
\end{equation}
When $p = o (n)$, $\nu^{\sff}$ can be further simplified to 
$\nu^{\sff}_{1} = \frac{\pi_{0}}{\pi_{1}} \frac{1}{n} V_{n} \left[ y_{i} (1) - \sum_{j \neq i} H_{j, i} y_{j} (1) \right].$
\end{enumerate}
\end{theorem}

According to the second assertion of Theorem \ref{thm:HOIF-CRE, design-based}, if $p = o (n)$, $\hat{\tau}_{\adj, 2}$ always attains smaller asymptotic variance (after scaled by $n$) than $\hat{\tau}_{\unadj}$; however, if $p = O (n)$, then the \emph{iff condition} for $\hat{\tau}_{\adj, 2}$ to enjoy improved asymptotic efficiency (after scaled by $n$) compared to 
$\hat{\tau}_{\unadj}$ is simply
\begin{equation}
\label{criterion}
\begin{split}
\nu^{\sff} - \var^{\sff} (\hat{\tau}_{\unadj}) \leq 0 \Longleftrightarrow \var^{\sff} (\hat{\tau}_{\unadj}) - \nu^{\sff}_{1} \geq \nu^{\sff}_{2}.
\end{split}
\end{equation}
Since the criterion for asymptotic efficiency improvement is fully characterized, one can easily conduct numerical experiments to determine whether $\hat{\tau}_{\adj, 2}$   has smaller asymptotic variance than  $\hat{\tau}_{\unadj}$ for given baseline covariates $\bX$ and potential outcomes $\by (1)$.

\begin{remark}[Interpreting $\hat{\tau}_{\adj, 2}$ as a leave-one-out regression adjustment estimator]\leavevmode
\label{rem:projection}
The first term $\nu^{\sff}_{1}$ of $\nu^{\sff}$ in \eqref{the variance} constitutes the sample variance of $y_{i} (1) - \sum_{j \neq i} H_{j, i} y_{j} (1)$ instead of $y_{i} (1)$, as in the variance of $\hat{\tau}_{\unadj}$. The term being subtracted off, $\sum_{j \neq i} H_{j, i} y_{j} (1)$, can be represented as 
$\sum_{j \neq i} H_{j, i} y_{j} (1) \equiv (\bx_{i} - \bar{\bx})^{\top} \hat{\bm{\beta}}_{-i},$
where $\hat{\bm{\beta}}_{-i} \coloneqq \hat{\bSigma}^{-} \sum_{j \neq i} (\bx_{j} - \bar{\bx}) y_{j} (1)$, which is the leave-one-out coefficient estimator for the linear projection of the potential outcomes $y_{i} (1)$ onto the linear span of $\bx_{i} - \bar{\bx}$, resembling the construction in \citet{wu2018loop}. In an ongoing work, we apply this idea of using leave-one-out regression adjustment in broader settings, where the outcome regression is fit by a working generalized linear model.
\end{remark}

\begin{remark}
\label{rem:Bernoulli}
The proof of Theorem \ref{thm:HOIF-CRE, design-based} is in 
Appendix \ref{app:hoif results}. Under the Bernoulli sampling 
with $t_{i} \overset{\rm i.i.d.}{\sim} \mathrm{Bernoulli} (\pi_{1})$ for $i \in [n]$, we immediately have $\bias^{\sff} (\hat{\tau}_{\adj, 2}) = 0$.
\end{remark}

The next result shows that statistical guarantees of $\hat{\tau}_{\adj, 2}$ parallel to those under the design-based framework in Theorem \ref{thm:HOIF-CRE, design-based} continue to hold under the i.i.d. superpopulation framework. We mention in passing that this result has been obtained in previous works by one of the authors of this paper \citep{liu2017semiparametric, liu2020nearly, liu2023hoif}, so we omit the proof.

\begin{proposition}
\label{thm:HOIF-CRE, superpopulation}
Under CRE (Assumption \ref{as:randomization}) and the superpopulation framework, the following hold:
\begin{enumerate}[leftmargin=1cm,topsep=0.25pt]
\item The bias of $\hat{\tau}_{\adj, 2}$ has the following form: $\bias (\hat{\tau}_{\adj, 2}) = \bbE (\hat{\tau}_{\adj, 2} - \tau) = - \dfrac{\pi_{0}}{\pi_{1}} \dfrac{1}{n - 1} \bbE \left[ H_{1, 1} \b (\bx_{1}) \right]$. In addition, if Assumption \ref{as:regularity conditions} holds, then $\bias (\hat{\tau}_{\adj, 2}) = O \left( \dfrac{\pi_{0}}{\pi_{1}} \dfrac{\alpha}{n} \right)$.

\item The variance of $\hat{\tau}_{\adj, 2}$ has the following order: Under Assumption \ref{as:regularity conditions}, further suppose that $\hat{\bSigma}$ has bounded eigenvalues, $\var (\hat{\tau}_{\adj, 2}) = O \left( \dfrac{1}{n} \left\{ 1 + \dfrac{p}{n} \right\} \right)$.
\end{enumerate}
\end{proposition}

\begin{remark}[On the asymptotic distribution of $\hat{\tau}_{\adj, 2}$]
\label{rem:clt}
Under Assumptions \ref{as:randomization}--\ref{as:regularity conditions} and the superpopulation framework, if we additionally suppose that there exists a constant $\sigma^{2} > 0$ such that
\begin{equation}
\label{super stable}
\lim_{n \rightarrow \infty} n \cdot \var (\hat{\tau}_{\adj, 2}) \rightarrow \sigma^{2}, 
\end{equation}
we have $\frac{\hat{\tau}_{\adj, 2} - \tau}{\sqrt{\var (\hat{\tau}_{\adj, 2})}} \rightsquigarrow N (0, 1).$
The Gaussian limiting distribution of $\hat{\tau}_{\adj, 2}$ follows from three main steps: (i) showing that the bias of $\hat{\tau}_{\adj, 2}$ is $o (1 / \sqrt{n})$; (ii) showing that the variance of $\hat{\tau}_{\adj, 2}$ is $O (1 / n)$; and (iii) invoking the CLT of second-order $U$-statistics (Corollary 1.2 of \citet{bhattacharya1992class}). We note that, as mentioned in numerous places in \citet{liu2020nearly}, the asymptotic distribution of $\hat{\tau}_{\adj, 2}$ can be obtained by applying Corollary 1.2 of \citet{bhattacharya1992class}. 
In the design-based framework, \eqref{super stable} can be replaced by
\begin{equation}
\label{stable}
\lim_{n \rightarrow \infty} n \cdot \nu^{\sff} \rightarrow \sigma^{2}, 
\end{equation}
and one needs to further adapt the proof technique in \citet{bhattacharya1992class} to L\'{e}vy's martingale CLT by adapting the proof technique in \citet{bhattacharya1992class}, or using results in \citet{koike2023high} as in \citet{lu2025debiased}. Conditions~\ref{super stable} and \ref{stable} essentially require that the variance scaled by $n$ has a limit as $n \rightarrow \infty$; see Appendix \ref{app:clt} for details.
\end{remark}

\begin{remark}[Semiparametric efficiency under the superpopulation framework]
\label{rem:emp}
Under the superpopulation framework, 
 when  $p$ is fixed and one imposes smoothness assumption on $\b$, 
say H\"{o}lder smooth with smoothness index $s > 0$, then it is easy to see that $\hat{\tau}_{\adj, 2}$, but with $\bx - \bar{\bx}$ replaced by $\bar{\phi}_{k}(\bx) = (\phi_{1} (\bx), \cdots, \phi_{k} (\bx))$, where $\bar{\phi}_{k}(\cdot)$ denotes low-degree polynomial transformations up to degree $k \asymp \log \log n$, achieves the SVB.
\end{remark}

\subsection{A Variety of HOIF-Motivated Estimators}
\label{sec:variety}

In the previous section, we have demonstrated that the HOIF theory motivates an adjusted estimator $\hat{\tau}_{\adj, 2}$ that (1) reduces the bias of the adjusted estimator $\hat{\tau}_{\adj, 1}^{\dag}$ and (2) has asymptotic variance 
no greater than $\hat{\tau}_{\unadj}$ whenever $p = o (n)$, 
sometimes more efficient than $\hat{\tau}_{\unadj}$ when $p = O (n)$ and $p < n$. 
We now propose  other HOIF-motivated estimators in the vicinity of $\hat{\tau}_{\adj, 2}$, with slightly different statistical properties.

To motivate alternatives, 
we consider a special scenario where the potential outcomes are constant, i.e. $y_{i} (1) \equiv c, i \in [n]$.
Without loss of generality, we take $c \equiv 1$. Here both 
$\hat{\tau}_{\unadj} \equiv 1$ and 
$\hat{\tau}_{\adj, 1}^{\dag} \equiv 1$ are \emph{error-free}, but $\hat{\tau}_{\adj, 2}$ fails to be \emph{error-free} in the extreme scenario of homogeneous potential outcomes.
To restore the \emph{error-free} property in such a case, one could remove the diagonal/trace part from the error-free adjusted estimator $\hat{\tau}_{\adj, 1}^{\dag}$ defined in \eqref{adjusted estimator} instead of $\hat{\tau}_{\adj, 1}$ defined in \eqref{adjusted estimator 1}:
\begin{equation}
\label{error free HOIF}
\hat{\tau}_{\adj, 2}^{\dag} \coloneqq \hat{\tau}_{\unadj} - \widehat{\IIFF}_{\unadj, 2, 2}^{\dag}, \text{ where } \widehat{\IIFF}_{\unadj, 2, 2}^{\dag} \coloneqq \frac{1}{n} \sum_{1 \leq i \neq j \leq n} \left( \frac{t_{i}}{\pi_{1}} - 1 \right) H_{i, j} \frac{t_{j} (y_{j} - \hat{\tau}_{\unadj})}{\pi_{1}}.
\end{equation}
It is natural to conjecture that the asymptotic variance of $\hat{\tau}_{\adj, 2}^{\dag}$ should be the same as that of $\hat{\tau}_{\adj, 2}$ in \eqref{the variance}, except that $y_{i} (1)$ is replaced by $y_{i} (1) - \bar{\tau}$ for all $i \in [n]$, i.e.
\begin{align*}
& \var^{\sff} (\hat{\tau}_{\adj, 2}^{\dag}) = \ \left( \frac{\pi_{0}}{\pi_{1}} \right) \frac{1}{n} V_{n} \left[ (y_{i} (1) - \bar{\tau}) - \sum_{j \neq i} H_{j, i} (y_{j} (1) - \bar{\tau}) \right] \\
& + \left( \frac{\pi_{0}}{\pi_{1}} \right)^{2} \frac{1}{n} \left\{ \frac{1}{n} \sum_{i = 1}^{n} H_{i, i} (1 - H_{i, i}) (y_{i} (1) - \bar{\tau})^{2}  + \frac{1}{n} \sum_{1 \leq i \neq j \leq n} H_{i, j}^{2} (y_{i} (1) - \bar{\tau}) (y_{j} (1) - \bar{\tau}) \right\} + o (n^{-1}).
\end{align*}

Despite being \emph{error-free} in the case of homogeneous potential outcomes, similar to $\hat{\tau}_{\adj, 2}$, $\hat{\tau}_{\adj, 2}^{\dag}$ is not unbiased under CRE. If one insists on constructing a \emph{bias-free} adjusted estimator under CRE, the following estimator can be constructed, again building on $\hat{\tau}_{\adj, 2}$:
\begin{equation}
\label{bias free HOIF}
\hat{\tau}_{\adj, 3} \coloneqq \hat{\tau}_{\adj, 2} + \frac{\pi_{0}}{\pi_{1}} \frac{1}{n (n - 1)} \sum_{i = 1}^{n} H_{i, i} \frac{t_{i} y_{i}}{\pi_{1}}.
\end{equation}
In Appendix \ref{app:variety}, we show that $\hat{\tau}_{\adj, 3}$ is unbiased under CRE, but biased under the Bernoulli sampling. The following proposition characterizes the asymptotic variance of $\hat{\tau}_{\adj, 3}$. Of course, a similar strategy can be employed to also completely remove the bias of $\hat{\tau}_{\adj, 2}^{\dag}$, which we do not further pursue.

\begin{proposition}
\label{prop:bias free HOIF variance}
Under Assumptions \ref{as:randomization} -- \ref{as:regularity conditions}, the following holds: 
$\var^{\sff} (\hat{\tau}_{\adj, 3}) = \var^{\sff} (\hat{\tau}_{\adj, 2}) + o \left( \frac{1}{n} \right).$
\end{proposition}

In words, one could completely remove the bias of $\hat{\tau}_{\adj, 2}$ due to CRE \emph{for free asymptotically}. The proof can be found in Appendix \ref{app:variety}. The conclusion of Proposition \ref{prop:bias free HOIF variance} directly implies the asymptotic normality of $\hat{\tau}_{\adj, 3}$ under the same conditions as in Remark \ref{rem:clt}. Using a similar strategy, one can also construct an exactly unbiased version of $\hat{\tau}_{\adj, 2}^{\dag}$ under CRE, denoted as $\hat{\tau}_{\adj, 3}^{\dag}$:
\begin{align*}
\hat{\tau}_{\adj, 3}^{\dag} \coloneqq \hat{\tau}_{\adj, 2}^{\dag} - 2 \frac{\pi_{0}}{\pi_{1}} \left( 1 - \frac{\pi_{0}}{\pi_{1}} \frac{1}{n - 1} \right) \frac{1}{n - 2} \left\{ \frac{p}{n} \hat{\tau}_{\unadj} - \frac{1}{n} \sum_{i = 1}^{n} H_{i, i} \frac{t_{i} y_{i}}{\pi_{1}} \right\}.
\end{align*}
The reason for $\hat{\tau}_{\adj, 3}^{\dag}$ will become immediately clear once we reveal the bias of $\hat{\tau}_{\adj, 2}^{\dag}$ under CRE in Proposition~\ref{prop:db} in the next subsection.

\subsection{HOIF-Motivated Estimators: A Unifying Theme of Recent Proposals}
\label{sec:understanding}

As mentioned in the Introduction, HOIF-motivated estimators unify several recently proposed adjusted estimators.
Unlike the OLS-based adjusted estimator $\hat{\tau}_{\adj, 1}^{\dag}$ or $\hat{\tau}_{\adj, 1}$, these estimators are CAN and has guarantee efficiency gains 
over $\hat{\tau}_{\unadj}$ even when $p \gtrsim \sqrt{n}$. Here, we will illustrate that HOIF-motivated estimators unify those from: 
\citet{lei2021regression}, \citet{lu2025debiased}, \citet{chang2024exact}, and also \citet{jiang2025adjustments}.

We start with \citet{lu2025debiased}, in which the following debiased adjusted estimator of $\bar{\tau}$ was proposed:
\begin{equation}
\label{debiased estimator}
\hat{\tau}_{\db} \coloneqq \hat{\tau}_{\adj, 1}^{\dag} + \frac{\pi_{0}}{\pi_{1}} \frac{1}{n} \sum_{i = 1}^{n} \frac{t_{i}}{\pi_{1}} H_{i,i} \left( y_{i} - \hat{\tau}_{\unadj} \right).
\end{equation}
As pointed out in \citet{lu2025debiased}, the estimator proposed in \citet{chang2024exact} is similar to $\hat{\tau}_{\db}$ but is exactly unbiased under CRE, and thus we denote the estimator in \citet{chang2024exact} as $\hat{\tau}_{\db}^{u}$.

The following lemma immediately classifies $\hat{\tau}_{\db}$ and $\hat{\tau}_{\db}^{u}$ as HOIF-motivated estimators. The proof is deferred to Appendix \ref{app:understanding}.

\begin{lemma}
\label{lem:alternative}
The following algebraic equivalences hold: $\hat{\tau}_{\db} \equiv \hat{\tau}_{\adj, 2}^{\dag}$ and $\hat{\tau}_{\db}^{u} \equiv \hat{\tau}_{\adj, 3}^{\dag}$.
\end{lemma}


\begin{remark}
\label{rem:algebraic difference}
The sole difference between $\hat{\tau}_{\adj, 2}$ and $\hat{\tau}_{\adj, 2}^{\dag}$, and now also $\hat{\tau}_{\db}$, is that the former does not center $y_{i}$ by $\hat{\tau}_{\unadj}$. In a slightly different vein, one can also consider $\hat{\tau}_{\unadj}$ as an ``adjusted'' estimator using the linear working model with \emph{only the intercept term but no baseline covariates}:
$\hat{\tau}_{\unadj} \equiv \frac{1}{n} \sum_{i = 1}^{n} \frac{t_{i}}{\pi_{1}} \left( y_{i} - \hat{\tau}_{\unadj} \right) + \hat{\tau}_{\unadj}.$
Then by a similar line of reasoning to that of Section \ref{sec:hoif}, we should use a second-order $U$-statistic to estimate the following ``projected bias'':
\begin{align*}
\bbE \left[ \left( \frac{t}{\pi_{1}} - 1 \right) (\bx - \bmu)^{\top} \right] (n^{-1} \bSigma)^{-} \bbE \left[ (\bx - \bmu) \frac{t}{\pi_{1}} (y - \hat{\tau}_{\unadj}) \right],
\end{align*}
leading to the statistic $\widehat{\IIFF}_{\unadj, 2, 2}^{\dag}$ defined in \eqref{error free HOIF}. We hope that this algebraic equivalence sheds some new light on the bias \& variance reduction ``mechanism'' of $\hat{\tau}_{\db}$ to readers.
\end{remark}

Next, we present the statistical properties of $\hat{\tau}_{\adj, 2}^{\dag}$, and equivalently, $\hat{\tau}_{\db}$.

\begin{proposition}
\label{prop:db}
Under Assumptions~\ref{as:randomization}--\ref{as:regularity conditions}, we have the following theoretical guarantees on $\hat{\tau}_{\adj, 2}^{\dag}$ and equivalently $\hat{\tau}_{\db}$:
\begin{enumerate}[leftmargin=1cm,topsep=0.25pt]
\item The design-based bias has the following form:
\begin{equation}
\label{bias, db, design-based}
\begin{split}
\bias^{\sff} \left( \hat{\tau}_{\adj, 2}^{\dag} \right) & \equiv \bias^{\sff} \left( \hat{\tau}_{\db} \right)  = 2 \frac{\pi_{0}}{\pi_{1}^{2}} \frac{n_{1} - 1}{n - 1} \frac{1}{n - 2} \left\{ \frac{p}{n} \bar{\tau} + \frac{1}{n} \sum_{1 \leq i \neq j \leq n} H_{i, j} y_{j} (1) \right\} \\
& = 2 \frac{\pi_{0}}{\pi_{1}} \left( 1 - \frac{\pi_{0}}{\pi_{1}} \frac{1}{n - 1} \right) \frac{1}{n - 2} \left\{ \frac{p}{n} \bar{\tau} - \frac{1}{n} \sum_{i = 1}^{n} H_{i, i} y_{i} (1) \right\}.
\end{split}
\end{equation}
In addition, if Assumption \ref{as:regularity conditions} holds, then $\bias^{\sff} \left( \hat{\tau}_{\adj, 2}^{\dag} \right) \equiv \bias^{\sff} \left( \hat{\tau}_{\db} \right) = O \left( \frac{\pi_{0}}{\pi_{1}} \frac{\alpha}{n} \right).$

\item The design-based variance has the following approximation under Assumption \ref{as:regularity conditions}:
\begin{equation}
\label{var, db, design-based}
\var^{\sff} \left( \hat{\tau}_{\adj, 2}^{\dag} \right) \equiv \var^{\sff} \left( \hat{\tau}_{\db} \right) = \nu^{\sff\dag} + o (n^{-1}),
\end{equation}
\begin{align*}
\text{where }
& \nu^{\sff\dag}  \coloneqq  \ \underbrace{\left( \frac{\pi_{0}}{\pi_{1}} \right) \frac{1}{n} V_{n} \left[ (y_{i} (1) - \bar{\tau}) - \sum_{j \neq i} H_{j, i} (y_{j} (1) - \bar{\tau}) \right]}_{\eqqcolon \, \nu^{\sff}_{\db, 1}} \\
& + \underbrace{\left( \frac{\pi_{0}}{\pi_{1}} \right)^{2} \frac{1}{n} \left\{ \frac{1}{n} \sum_{i = 1}^{n} H_{i, i} (1 - H_{i, i}) (y_{i} (1) - \bar{\tau})^{2}  + \frac{1}{n} \sum_{1 \leq i \neq j \leq n} H_{i, j}^{2} (y_{i} (1) - \bar{\tau}) (y_{j} (1) - \bar{\tau}) \right\}}_{\eqqcolon \, \nu^{\sff}_{\db, 2}}.
\end{align*}

\end{enumerate}
\end{proposition}
 
The proof of Proposition \ref{prop:db} can be found in Appendix \ref{app:understanding}. The asymptotic normality of $\hat{\tau}_{\db}$ can be obtained in a fashion similar to that of $\hat{\tau}_{\adj, 2}$, and therefore omitted. It is now also clear why $\hat{\tau}_{\adj, 3}^{\dag}$ is exactly unbiased under CRE. We also refer to \citet{lu2025debiased} for analogous results on ATE. Unlike \citet{lu2025debiased}, we, in fact, characterize the exact variance of $\hat{\tau}_{\db}$ or $\hat{\tau}_{\adj, 2}^{\dag}$, and the asymptotic variance is simply a corollary. We do not further consider the statistical properties of $\hat{\tau}_{\db}$ under the superpopulation framework.

\begin{remark}
\label{rem:var comparison}
We briefly compare the variances of $\hat{\tau}_{\text{adj}, 2}$ and $\hat{\tau}_{\text{adj}, 2}^{\dag}$ (equivalent to $\hat{\tau}_{\text{db}}$). Equations \eqref{var, HOIF-CRE, design-based} and \eqref{var, db, design-based} show that their asymptotic variances differ only in whether potential outcomes are centered by $\bar{\tau}$. The advantage of $\hat{\tau}_{\text{adj}, 2}^{\dag}$ emerges when potential outcomes are homogeneous with a small coefficient of variation. Although centering is common in practice, Appendix \ref{app:var comparison} provides computer-assisted examples where centering increases asymptotic variance. 
\end{remark}

We are left to discuss the connection of the estimators proposed in \citet{lei2021regression} and \citet{jiang2025adjustments} to HOIF-motivated estimators. The estimator proposed in \citet{lei2021regression}, denoted by $\hat{\tau}_{\rm ld}$, differs from $\hat{\tau}_{\db}$ \eqref{debiased estimator} only in the way $\hat{\bbeta}_{c}$ in $\hat{\tau}_{\adj, 1}^{\dag}$ (see \eqref{adjusted estimator}) is computed. Instead of using $\hat{\bSigma}$ in $\hat{\bbeta}_{c}$, $\hat{\tau}_{\rm ld}$ estimates $\hat{\Sigma}$ using covariates only in the treated group, although the covariate distributions in the two groups should be the same by design. $\hat{\tau}_{\rm ld}$ remains CAN if $p = O (n^{2 / 3})$ up to a log-factor \citep{lei2021regression}, and the stronger dependence on the dimension $p$ is a result of the mismatch between using an estimated $\hat{\bSigma}_{1}$ \eqref{tSigma} in $\hat{\tau}_{\adj, 1}^{\dag}$ and using the true $\hat{\bSigma}$ when removing the ``diagonal''. The extra term of $\hat{\tau}_{\rm ld}$ over $\hat{\tau}_{\adj, 2}^{\dag}$ contains an error of order $\Vert \hat{\bSigma}^{-1} - \hat{\bSigma}_{1}^{-1} \Vert_{\op} = O_{\bbP} (p^{1 / 2} / n^{3 / 2})$ up to log-factors by matrix Bernstein inequality \citep{tropp2015introduction}, resulting in $p = O (n^{2 / 3})$; see the end of Appendix~\ref{app:understanding} for derivations. Finally, \citet{jiang2025adjustments} directly uses the $V$-statistic $\hat{\tau}_{\adj, 1}^{\dag}$, because they either assume $p = o (\sqrt{n})$ or assume that the linear working model for the outcome regression is correctly specified.

We now summarize our findings in this section in Table~\ref{t:one} below, which, in our opinion, is the most important message in our paper.

\begin{table}[htbp]
\centering
\caption{The connection between HOIF-motivated estimators and other recent proposals. Here $\tilde{o}(\cdot)\coloneqq o(\cdot/(\log n)^{1/3})$. In the first row, the extra term between $\hat{\tau}_{\rm ld}$ and $\hat{\tau}_{\adj, 2}^{\dag}$ can be found in the end of Appendix~\ref{app:understanding}.}
\label{t:one}
\begin{tabular}{cccc}
\Xhline{3\arrayrulewidth}
\textbf{Estimator} & \textbf{Correspondence} & \textbf{Linear Model} & \textbf{Dimension} \\ 
\Xhline{3\arrayrulewidth}
\addlinespace
\multirowcell{2}{$\hat{\tau}_{\mathrm{ld}}$ \\ {\scriptsize\citet{lei2021regression}}} 
& \multirowcell{2}{$\equiv \hat{\tau}_{\mathrm{adj}, 2}^{\dag} + \mathrm{rem}$}
& \multirowcell{2}{\xmark}
& \multirowcell{2}{$p = \tilde{o} \left( n^{2/3} \right)$} \\ [20pt]
\hline
\addlinespace 
\addlinespace 
\multirowcell{2}{$\hat{\tau}_{\mathrm{db}}$ \\ {\scriptsize\citet{lu2025debiased}}} 
&  \multirowcell{2}{$\equiv \hat{\tau}_{\mathrm{adj}, 2}^{\dag}$}
& \multirowcell{2}{\xmark}
& \multirowcell{2}{$p = o(n)$} \\[20pt]
\hline
\addlinespace
\multirowcell{2}{$\hat{\tau}_{\mathrm{db}}^{u}$ \\ {\scriptsize\citet{chang2024exact}}} 
& \multirowcell{2}{$\equiv \hat{\tau}_{\mathrm{adj}, 3}^{\dag}$} 
& \multirowcell{2}{\xmark} 
& \multirowcell{2}{$p = o(n)$} \\[20pt]
\hline
\addlinespace
\multirowcell{3}{$\hat{\tau}_{\mathrm{adj}, 1}^{\dag}$ \\ {\scriptsize\citet{ma2022regression, ye2023toward}} \\ {\scriptsize\citet{jiang2025adjustments}, and etc.}} 
& \multirowcell{3}{\shortstack{\small V-statistic \\ \small version of \normalsize$\hat{\tau}_{\adj, 2}^{\dag}$}} 
& \multirowcell{3}{\shortstack{\cmark \\ \\ \xmark}} 
& \multirowcell{3}{\shortstack{$p = o(n)$ \\ $p = o(\sqrt{n})$}} \\[30pt]
\addlinespace
\Xhline{3\arrayrulewidth}
\end{tabular}
\end{table}

\begin{remark}
\label{rem:others}
We finally comment on the connection of our work with two recent preprints: \citet{abadie2025unbiased} and \citet{song2025neumann}. \citet{abadie2025unbiased} use the ridge inverse $\hat{\bSigma}_{\lambda}^{-1} = (\bX^{\top} \bX + \lambda \bI)^{-1}$ instead of the non-regularized inverse $\hat{\bSigma}^{-1}$ to improve the numerical stability of the estimator. In the implementation of our proposed estimator, we in fact use the generalized inverse, which is also numerically stable and can be viewed as the limit of $\hat{\bSigma}_{\lambda}^{-1}$ as $\lambda \rightarrow 0$. From our analysis, the adjusted estimator will always have negligible bias compared to the sampling variability if a $U$-statistic is used instead of a $V$-statistic and $\bt$ is not used to compute the inverse Gram matrix. Therefore, we conjecture that as long as $\lambda$ is appropriately chosen such that $\hat{\bSigma}_{\lambda}^{-1} - \hat{\bSigma}^{-1}$ is sufficiently close, we can conclude that the ridge-regularized adjusted estimator of \citet{abadie2025unbiased} has an asymptotic variance sometimes smaller than and never greater than that of $\hat{\tau}_{\unadj}$. Regarding the estimator proposed in \citet{song2025neumann}, since they estimate $\hat{\bSigma}^{-1}$ by $\hat{\bSigma}_{t}^{-1}$ as in \citet{lei2021regression}, as already suggested in \citet{liu2020nearly}, additional corrections of the estimation bias is needed to relax the requirement on $p$, as done in \citet{liu2023hoif} or initially in \citet{liu2020nearly}. It should be noted that \citet{song2025neumann} focused on the design-based framework, so the true Gram matrix $\hat{\bSigma}$ is known and bias correction using Neumann series is simplified compared to \citet{liu2023hoif} or \citet{liu2020nearly}, because they assume the existence of a (possibly fictitious) superpopulation so the true $\bSigma$ is unknown.
\end{remark}

\section{Simulation Studies and Real Data Application}
\label{sec:simulations}


In this section, we conduct numerical experiments and real data analysis to explore our theoretical findings regarding the design-based bias and variance formulas for $\hat{\tau}_{\unadj}$, $\hat{\tau}_{\adj, 2}$, $\hat{\tau}_{\adj, 2}^{\dag}$ and $\hat{\tau}_{\adj, 3}$.
The accompanying \texttt{R} code for replicating our simulations is available at \href{https://github.com/Cinbo-Wang/HOIF-Car}{the GitHub repository}, which includes both the exact bias and variance formulas, as well as the asymptotic versions presented in our paper.

\subsection{Simulation studies}
\label{sec:faux sim}

We first generate a very large data matrix $\mathcal{X} \in \bbR^{N\times N}$ with $N = 5000$, with each row $\mathcal{X}_i \sim t_3(0,\Sigma)$, where $\Sigma_{k,l}=0.1^{|k-l|},k,l\in[N]$. Then we generated two types of exogenous noise $\mathcal{\epsilon}_0\in \bbR^{N}$: 
univariate $t_3$ distribution and “worst-case residual” \citep{lei2021regression}): $\mathcal{\epsilon}_0=\text{Scale}\left((\mathbb{I}_N - \bH)(H_{1,1},\cdots,H_{N,N})^\top\right),$ where $\text{Scale}(a_i)\coloneqq \left(a_i-\bar{a}\right)/\left(\sum_{i=1}^{n}(a_i-\bar{a})^2\right)^{1/2}$.
Also, we let $\bbeta_j = (-1)^j/\sqrt{j},j\in[N]$. We vary the following in our experiments:
\begin{itemize}[leftmargin=0.5cm,topsep=0.25pt]
    \item \textbf{covariates} $\bX$: sample size $n \in \{ 50,100,500,1000\}$, and covariate dimension $p = \lceil n \cdot \alpha \rceil$ where $\alpha \in \{ 0.05,\dots,0.7\}$. Then we choose $\bX \coloneqq \mathcal{X}_{1:n,1:p} \in \bbR^{n \times p}$.
    \item \textbf{potential outcome model}: linear $f(\bx)=\bx^{\top} \bbeta_{[p]}$, and nonlinear $f(\bx)=\text{sign}\left(\bx^{\top} \bbeta_{[p]}\right)|\bx^{\top} \bbeta_{[p]}|^{\frac{1}{2}} + \sin(\bx^{\top} \bbeta_{[p]})$.  Here, $\bbeta_{[p]}$ represents the first $p$ elements of $\bbeta$.
    \item \textbf{exogenous error distribution} $\epsilon$: for both $t_3$ and ``worst-case residual'' types, we  choose the first $n$ elements in $\epsilon_0$.
    \item \textbf{signal size} $ \gamma\in \{1,2\}$: the potential outcomes are generated according to $y_i (1) = 1 + f(\bx_i) + \epsilon_i\cdot \sqrt{V_n\left(f(\bX)\right)/V_n(\epsilon)} / \sqrt{\gamma}$ for $i\in [n]$.
    \item \textbf{treatment assignment ratio} $\pi_1 \in \{\frac{1}{2},\frac{2}{3}\}$: once the pre-treatment variables $\{\left(\bx_i,y_i(1)\right)\}_{i=1}^{n}$ are generated, we fix them and randomly assign $\lceil n\pi_1\rceil$ samples to the treatment arm. (When $\pi_1=\frac{2}{3}$,  the ``true'' $\pi_1$ we used for estimation is $\lceil n\pi_1\rceil / n $.) 
\end{itemize}

The relative efficiencies  of different estimators, based on their exact and approximate variances, benchmarked by $\var^{\sff} (\hat{\tau}_{\unadj})$, are presented in Figure \ref{fig:oracle_var_main}. Overall, $\hat{\tau}_{\adj, 2}$ and $\hat{\tau}_{\adj, 2}^{\dag}$($\hat{\tau}_{\db}$) exhibit very similar relative efficiencies; however, in the nonlinear outcome model setting, there are cases where $\hat{\tau}_{\adj, 2}^{\dag}$ demonstrates greater efficiency. As indicated in Remark \ref{rem:var comparison}, $\hat{\tau}_{\adj, 2}^{\dag}$  is expected to  be more efficient when CoV$^{2}$ is small, which is confirmed by the results shown in Figure \ref{fig:cv2_variance_ratio_main}.
\begin{figure}[H]
    \centering
    \includegraphics[width=\linewidth]{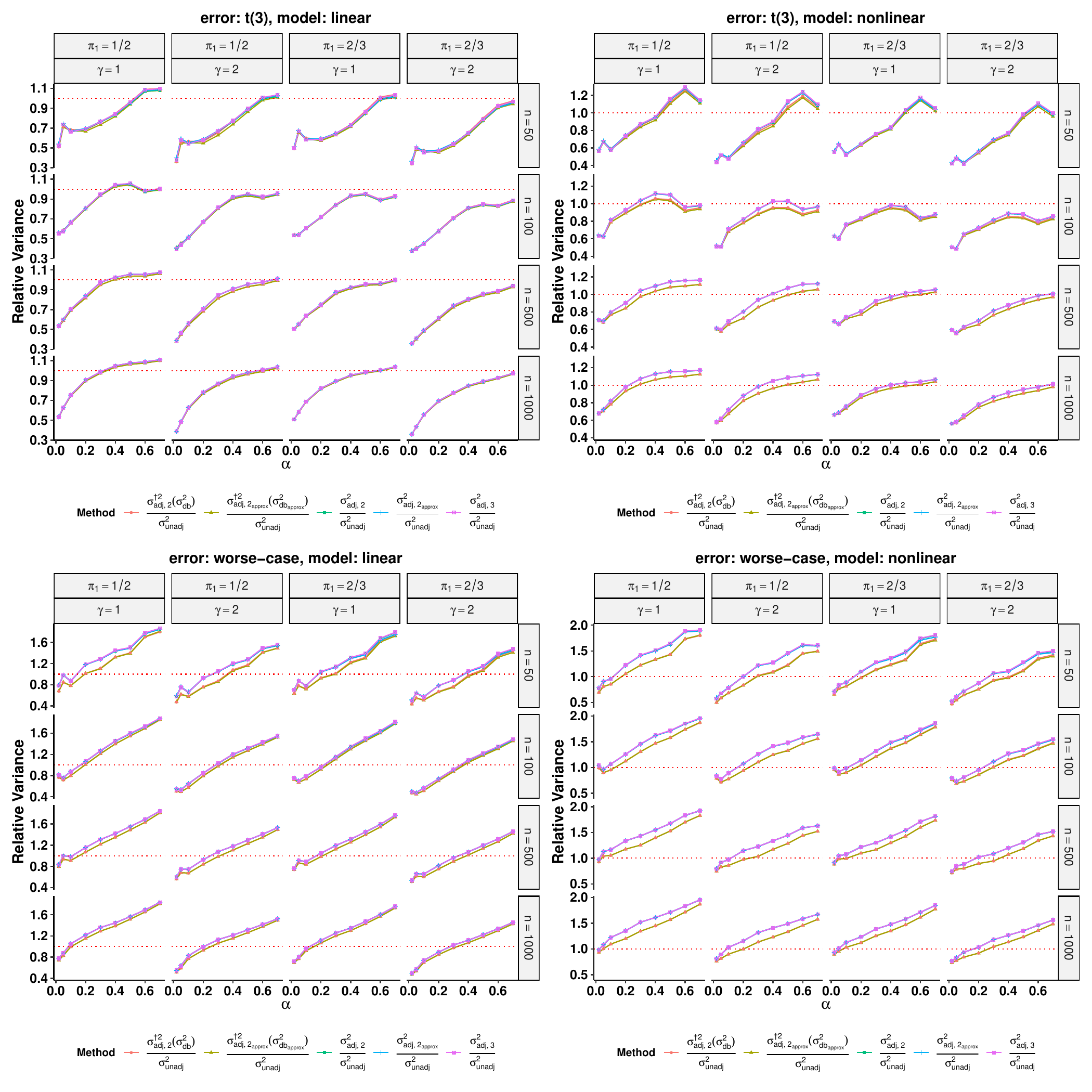}
    \caption{Relative efficiencies of $\hat{\tau}_{\adj, 2}^{\dag}$ (or equivalently $\hat{\tau}_{\db}$), $\hat{\tau}_{\adj, 2}$, $\hat{\tau}_{\adj, 3}$ based on exact and approximate formula.}
    \label{fig:oracle_var_main}
\end{figure}

\begin{figure}[H]
    \centering
    \includegraphics[width=0.6\linewidth]{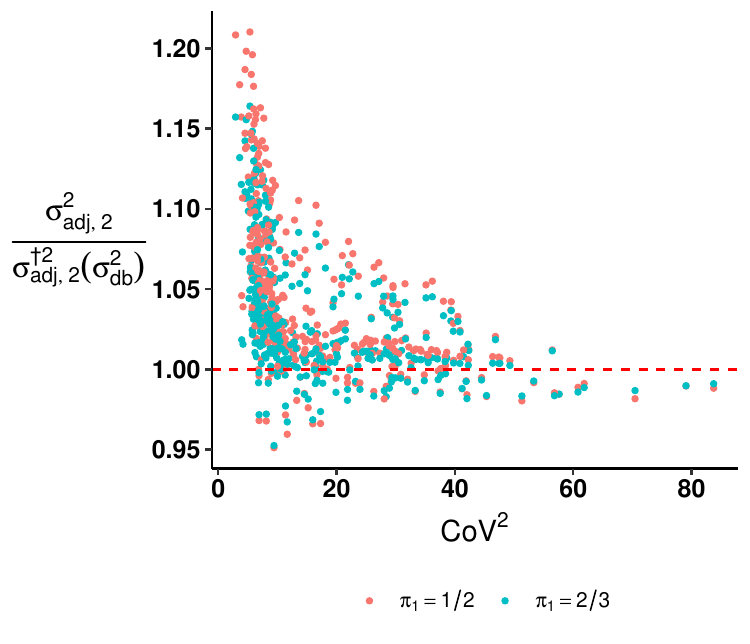}
    \caption{CoV$^{2}$ vs. $\var^{\sff} (\hat{\tau}_{\adj, 2}) / \var^{\sff} (\hat{\tau}_{\adj, 2}^{\dag})$. Each point represents a particular simulation setting in Figure \ref{fig:oracle_var_main}; only the settings with CoV$^{2} \leq 100$ are shown here.}
    \label{fig:cv2_variance_ratio_main}
\end{figure}

Furthermore, we conduct realistic simulations by actually drawing treatment assignments. Once $\{\bx_i, y_i(1)\}_{i = 1}^{n}$ are generated, they are fixed, and random treatment assignments are drawn repeatedly from CRE, with the Monte Carlo repetition size $K = 20000$. The resulting biases and variances are shown in the Appendix \ref{app:sim} and the overall message is similar.

\subsection{Real Data Application}
\label{app:real}

We next apply our proposed estimators to  
 data from the NIDA-CTN-0030 trial, which tested whether adding individual drug counseling to buprenorphine/naloxone treatment and standard medical management (SMM) improved outcomes for prescription opioid dependence. We focused on the first phase, where patients were randomized to either SMM or enhanced medical management (EMM), stratified by heroin use and chronic pain history. For simplicity, we treated the design as CRE by including the randomization stratum as baseline covariates.


The outcome of interest is the proportion of positive urine laboratory results among all tests. We include the following baseline covariates: randomization stratum, age, sex, and baseline urine test results. Missing data were replaced by medians. Following \citet{wang2023model}, outcomes were considered missing after two consecutive missed tests. We constructed a design matrix $\bX \in \bbR^{n \times p}$ with $n=587, p=6$, an outcome vector $\by \in \bbR^{n}$, and a treatment assignment vector (296 controls, 291 treated). Our target is the treatment-specific mean $\bar{\tau}$.

\begin{table}
\begin{center}
\caption{Treatment effect estimates of the proportion of positive urine laboratory results among all tests, using data from NIDA-CTN-0030. The four approaches considered are (a) $\hat{\tau}_{\unadj}$, (b) $\hat{\tau}_{\adj,2}$, (c)  $\hat{\tau}_{\adj,2}^{\dag}$ and (d) $\hat{\tau}_{\adj,3}$. All values are multiplied by 10. The superscript ``c" refers to the conservative standard error estimates.}
\label{app_CTN30}
\begin{tabular}{cccccc}
\Xhline{3\arrayrulewidth}
\textbf{Method} & \textbf{Estimate} & \textbf{SE} & \textbf{$95\%$-CI} & \textbf{SE$^{c}$} & \textbf{$95\%$-CI$^{c}$} \\ 
\Xhline{3\arrayrulewidth}
$\hat{\tau}_{\unadj}$ & $1.1836$ & $0.0372$ & $(1.1108, 1.2564)$ & NA & NA \\ 
\hline
$\hat{\tau}_{\adj,2}$ & $1.1837$ & $0.0327$ & (1.1196, 1.2489) & 0.0328 & (1.1195, 1.2479)  \\ 
\hline
$\hat{\tau}_{\adj,2}^{\dag}$ & $1.1755$ & $0.0321$ & (1.1126, 1.2383) & 0.0318 & (1.1131, 1.2379) \\ 
\hline
$\hat{\tau}_{\adj,3}$ & $1.1838$ & $0.0327$ & (1.1196, 1.2489) & 0.0328 & (1.1196, 1.2480)  \\ 
\Xhline{3\arrayrulewidth}
\end{tabular}
\end{center}
\end{table}

In Table \ref{app_CTN30}, we present the point estimate, along with the non-conservative and conservative estimates of standard deviation, and the $95\%$ confidence interval for each of the four different estimators, excluding $\hat{\tau}_{\unadj}$. Compared to other adjusted estimators, the point estimate of $\hat{\tau}_{\adj,2}$  and the bias-free $\hat{\tau}_{\adj,3}$ that we proposed are closer to the unadjusted estimator $\hat{\tau}_{\unadj}$. Furthermore, their standard errors are also lower than $\hat{\tau}_{\unadj}$.

\section{Concluding Remarks}
\label{sec:conclusions}
In this paper, we demonstrate that HOIF-motivated estimators are natural candidate treatment effect estimators adjusting for high-dimensional baseline covariates in RCT. We prove statistical properties of the HOIF-motivated estimator under CRE or Bernoulli sampling and the design-based framework, which, to the best of our knowledge, is new in the HOIF literature. More importantly, we show that HOIF-motivated estimators place several recently proposed treatment effect estimators adjusting for high-dimensional baseline covariates in a unified framework. This result further consolidates the role of HOIF as a useful template to construct ``good'' estimators from first principles, without resorting to clever tricks. Finally, there are two main future directions that worth pursuing. First, it is of theoretical interest to use the design-based Riesz representation theory of \citet{harshaw2022design} to justify that the HOIF-motivated estimators are indeed HOIF in the design-based framework. Second, it is of practical interest to extend our estimators to designs beyond CRE.

\section*{Acknowledgments}

Section~\ref{sec:hoif} of this paper is partly motivated by a conversation with Oliver Dukes, who kindly suggested to LL that a somewhat pedagogical/elementary review of HOIF could be useful to practitioners. The authors would also like to express their sincere gratitude to Yujia Gu and \href{https://maweiruc.github.io/}{Wei Ma} for enlightening discussions and Muluneh Alene Addis and Kelly Van Lancker for spotting an issue in the variance estimator in the first version of the R package \texttt{HOIFCar}.


%
\putbib[Master]

\end{bibunit}

\appendix

\begin{appendices}

\begin{bibunit}[plainnat]

\section{Proofs of Lemma in Section \ref{sec:hoif}}
\label{app:lemma1}
Here, we  prove Lemma \ref{lem:svb}:
\begin{proof}
The variance reduction part crucially relies on the following observation:
\begin{align*}
\var \left\{ \left( \frac{t}{\pi_{1}} - 1 \right) (\bx - \bmu)^{\top} \bbeta \right\} & = \cov \left\{ \frac{t}{\pi_{1}} y, \left( \frac{t}{\pi_{1}} - 1 \right) (\bx - \bmu)^{\top} \bbeta \right\} \\
& = \cov \left\{ \frac{t}{\pi_{1}} (y - \tau), \left( \frac{t}{\pi_{1}} - 1 \right) (\bx - \bmu)^{\top} \bbeta \right\},
\end{align*}
which is equivalent to saying that the augmented term $\left( \frac{t}{\pi_{1}} - 1 \right) (\bx - \bmu)^{\top} \bbeta$ is an $L_{2} (\bbP)$-projection of $\frac{t}{\pi_{1}} y$ (or equivalently $\frac{t}{\pi_{1}} (y - \tau)$) onto the following space of mean-zero random variables $\{(\frac{t}{\pi_{1}} - 1) (\bx - \bmu)^{\top} \bm{b}: \bm{b} \in \bbR^{p}\}$. Since $L_{2} (\bbP)$-projection contracts $L_{2} (\bbP)$-norm, which corresponds to the variance, $\tilde{\tau}_{\adj}$ has a reduced variance.
\end{proof}

\section{Proofs of Theorems in Section \ref{sec:main results}}
\label{app:main results}

In the derivation below, we also frequently use the following notation: $\hat{\bSigma}_{y} \coloneqq \bX_{c}^{\top} (\by (1) \circ \bX_{c})$, 
where $\circ$ denotes the Hadamard product.

\subsection{Proofs of Theorems in Section \ref{sec:hoif results}}
\label{app:hoif results}

\allowdisplaybreaks

We first prove Theorem \ref{thm:HOIF-CRE, design-based}.

\begin{proof} 
We first characterize the bias of $\hat{\tau}_{\adj, 2}$ in the design-based framework. It is trivial to see that $\hat{\tau}_{\unadj}$ is unbiased so we only need to compute the expectation of $\widehat{\IIFF}_{\unadj, 2, 2}$ using Lemma \ref{lem:CRE}:
\begin{align*}
\bbE^{\sff} (\widehat{\IIFF}_{\unadj, 2, 2}) & = \frac{1}{n} \sum_{1 \leq i \neq j \leq n} \frac{\bbE (t_{i} t_{j})}{\pi_{1}^{2}} H_{i, j} y_{j} (1) - \frac{1}{n} \sum_{1 \leq i \neq j \leq n} \frac{\bbE (t_{j})}{\pi_{1}} H_{i, j} y_{j} (1) \\
& = - \frac{1}{n} \sum_{1 \leq i \neq j \leq n} H_{i, j} y_{j} (1) \left\{ 1 - \frac{1}{\pi_{1}} \frac{n \pi_{1} - 1}{n - 1} \right\} \\
& = \frac{\pi_{0}}{\pi_{1}} \frac{1}{n (n - 1)} \sum_{1 \leq i \neq j \leq n} H_{i, j} y_{j} (1).
\end{align*}

Next, we compute the variance of $\hat{\tau}_{\adj, 2}$ by using the results from Lemma \ref{lem:HOIF-var} below.
\begin{align*}
& \ \var^{\sff} (\hat{\tau}_{\adj, 2}) \\
= & \ \var^{\sff} (\hat{\tau}_{\unadj}) + \var^{\sff} (\widehat{\IIFF}_{\unadj, 2, 2}) - 2 \cov^{\sff} (\hat{\tau}_{\unadj}, \widehat{\IIFF}_{\unadj, 2, 2}) \\
= & \ \frac{\pi_{0}}{\pi_{1}} \frac{1}{n} V_{n} (y (1)) - \frac{\pi_{0}}{\pi_{1}} \frac{1}{n^{2}} \sum_{1 \leq i \neq j \leq n} H_{i, j} (1 + 2 H_{j, j}) y_{i} (1) y_{j} (1) + \frac{\pi_{0}}{\pi_{1}} \frac{1}{n^{2}} \sum_{i = 1}^{n} H_{i, i} \left( \frac{1}{\pi_{1}} - \frac{\pi_{0} + 1}{\pi_{1}} H_{i, i} \right) y_{i} (1)^{2} \\
& + \left( \frac{\pi_{0}}{\pi_{1}} \right)^{2} \frac{1}{n^{2}} \trace \left( \left( \hat{\bSigma}_{y} \hat{\bSigma}^{-} \right)^{2} \right) - \frac{\pi_{0}}{\pi_{1}} \frac{1}{n^{3}} \trace^{2} \left( \hat{\bSigma}_{y} \hat{\bSigma}^{-} \right) - 2 \bar{\tau} \cdot \frac{\pi_{0}}{\pi_{1}} \frac{1}{n^{2}} \trace \left( \hat{\bSigma}_{y} \hat{\bSigma}^{-} \right) + \mathsf{Rem},
\end{align*}
where $\mathsf{Rem} = \mathsf{Rem}_{1} - 2 \mathsf{Rem}_{2}$. The exact forms of $\mathsf{Rem}_{1}$ and $\mathsf{Rem}_{2}$ can be deduced from the proof of Lemma \ref{lem:HOIF-var} and both are of order $o (1 / n)$ under Assumption \ref{as:regularity conditions}. As in the main text, denoting the dominating term by $\nu^{\sff}$, this term can be further simplified as follows:
\begin{align*}
\nu^{\sff} \coloneqq & \ \frac{\pi_{0}}{\pi_{1}} \frac{1}{n} V_{n} (y (1)) - \frac{\pi_{0}}{\pi_{1}} \frac{1}{n^{2}} \sum_{1 \leq i \neq j \leq n} H_{i, j} (1 + 2 H_{j, j}) y_{i} (1) y_{j} (1) \\
& + \frac{\pi_{0}}{\pi_{1}} \frac{1}{n^{2}} \sum_{i = 1}^{n} H_{i, i} \left( \frac{1}{\pi_{1}} - \frac{\pi_{0} + 1}{\pi_{1}} H_{i, i} \right) y_{i} (1)^{2} + \left( \frac{\pi_{0}}{\pi_{1}} \right)^{2} \frac{1}{n^{2}} \sum_{i = 1}^{n} \sum_{j = 1}^{n} y_{i} (1) H_{i, j}^{2} y_{j} (1) \\
& - \frac{\pi_{0}}{\pi_{1}} \frac{1}{n} \left( \frac{1}{n} \sum_{i = 1}^{n} H_{i, i} y_{i} (1) \right)^{2} - 2 \frac{\pi_{0}}{\pi_{1}} \frac{1}{n} \bar{\tau} \cdot \left( \frac{1}{n} \sum_{i = 1}^{n} H_{i, i} y_{i} (1) \right) \\
= & \ \left( \frac{\pi_{0}}{\pi_{1}} \right) \frac{1}{n} \left\{ \frac{1}{n} \sum_{i = 1}^{n} y_{i} (1)^{2} - \frac{1}{n} \sum_{1 \leq i \neq j \leq n} y_{i} (1) H_{i, j} (1 + H_{i, i} + H_{j, j}) y_{j} (1) \right\} \\
& + \left( \frac{\pi_{0}}{\pi_{1}} \right) \left\{ 1 + \left( \frac{\pi_{0}}{\pi_{1}} \right) \right\} \frac{1}{n^{2}} \sum_{i = 1}^{n} H_{i, i} (1 - H_{i, i}) y_{i} (1)^{2} - \frac{\pi_{0}}{\pi_{1}} \frac{1}{n} \left( \frac{1}{n} \sum_{i = 1}^{n} (1 + H_{i, i}) y_{i} (1) \right)^{2} \\
& + \left( \frac{\pi_{0}}{\pi_{1}} \right)^{2} \frac{1}{n^{2}} \left\{ \sum_{i = 1}^{n} \sum_{j = 1}^{n} H_{i, j}^{2} y_{i} (1) y_{j} (1) - \sum_{i = 1}^{n} H_{i, i}^{2} y_{i} (1)^{2} \right\} \\
= & \ \left( \frac{\pi_{0}}{\pi_{1}} \right) \frac{1}{n} \left\{ \frac{1}{n} \sum_{i = 1}^{n} \{(1 + H_{i, i}) y_{i} (1)\}^{2} - \left( \frac{1}{n} \sum_{i = 1}^{n} (1 + H_{i, i}) y_{i} (1) \right)^{2} \right\} \\
& + \left( \frac{\pi_{0}}{\pi_{1}} \right) \frac{1}{n} \left\{ \frac{1}{n} \sum_{i = 1}^{n} (- H_{i, i} - 2 H_{i, i}^2) y_{i} (1)^{2} - \frac{1}{n} \sum_{1 \leq i \neq j \leq n} y_{i} (1) H_{i, j} (1 + H_{i, i} + H_{j, j}) y_{j} (1) \right\} \\
& + \left( \frac{\pi_{0}}{\pi_{1}} \right)^{2} \frac{1}{n^{2}} \left\{ \sum_{i = 1}^{n} \sum_{j = 1}^{n} H_{i, j}^{2} y_{i} (1) y_{j} (1) + \sum_{i = 1}^{n} H_{i, i} (1 - 2 H_{i, i}) y_{i} (1)^{2} \right\} \\
= & \ \left( \frac{\pi_{0}}{\pi_{1}} \right) \frac{1}{n} \left\{ \frac{1}{n} \sum_{i = 1}^{n} \{(1 + H_{i, i}) y_{i} (1)\}^{2} - \left( \frac{1}{n} \sum_{i = 1}^{n} (1 + H_{i, i}) y_{i} (1) \right)^{2} \right\} \\
& + \left( \frac{\pi_{0}}{\pi_{1}} \right) \frac{1}{n} \left\{ - \frac{1}{n} \sum_{i = 1}^{n} \sum_{j = 1}^{n} y_{i} (1) H_{i, j} (1 + H_{i, i} + H_{j, j}) y_{j} (1) \right\} \\
& + \left( \frac{\pi_{0}}{\pi_{1}} \right)^{2} \frac{1}{n^{2}} \left\{ \sum_{i = 1}^{n} \sum_{j = 1}^{n} H_{i, j}^{2} y_{i} (1) y_{j} (1) + \sum_{i = 1}^{n} H_{i, i} (1 - 2 H_{i, i}) y_{i} (1)^{2} \right\} \\
= & \ \left( \frac{\pi_{0}}{\pi_{1}} \right) \frac{1}{n} \left\{ \frac{1}{n} \sum_{i = 1}^{n} \{(1 + H_{i, i}) y_{i} (1)\}^{2} - \left( \frac{1}{n} \sum_{i = 1}^{n} (1 + H_{i, i}) y_{i} (1) \right)^{2} \right\} \\
& - \left( \frac{\pi_{0}}{\pi_{1}} \right) \frac{1}{n} \left\{ \frac{1}{n} \sum_{i = 1}^{n} \sum_{j = 1}^{n} y_{i} (1) H_{i, j} (1 + H_{i, i} + H_{j, j}) y_{j} (1) \right\} \\
& + \left( \frac{\pi_{0}}{\pi_{1}} \right)^{2} \frac{1}{n} \left\{ \frac{1}{n} \sum_{1 \leq i \neq j \leq n} H_{i, j}^{2} y_{i} (1) y_{j} (1) + \frac{1}{n} \sum_{i = 1}^{n} H_{i, i} (1 - H_{i, i}) y_{i} (1)^{2} \right\} \\
= & \ \left( \frac{\pi_{0}}{\pi_{1}} \right) \frac{1}{n} \frac{1}{n} \sum_{i = 1}^{n} \left\{ (1 + H_{i, i}) y_{i} (1) - \sum_{j = 1}^{n} H_{i, j} y_{j} (1) - \left( \frac{1}{n} \sum_{l = 1}^{n} (1 + H_{l, l}) y_{l} (1) \right) \right\}^{2} \\
& + \left( \frac{\pi_{0}}{\pi_{1}} \right)^{2} \frac{1}{n} \left\{ \frac{1}{n} \sum_{1 \leq i \neq j \leq n} H_{i, j}^{2} y_{i} (1) y_{j} (1) + \frac{1}{n} \sum_{i = 1}^{n} H_{i, i} (1 - H_{i, i}) y_{i} (1)^{2} \right\} \\
= & \ \left( \frac{\pi_{0}}{\pi_{1}} \right) \frac{1}{n} V_{n} \left[ y_{i} (1) - \sum_{j \neq i} H_{j, i} y_{j} (1) \right] \\
& + \left( \frac{\pi_{0}}{\pi_{1}} \right)^{2} \frac{1}{n} \left\{ \frac{1}{n} \sum_{1 \leq i \neq j \leq n} H_{i, j}^{2} y_{i} (1) y_{j} (1) + \frac{1}{n} \sum_{i = 1}^{n} H_{i, i} (1 - H_{i, i}) y_{i} (1)^{2} \right\}.
\end{align*}
\end{proof}

\begin{lemma}
\label{lem:HOIF-var}
The following holds under Assumption \ref{as:randomization}:
\begin{equation}
\label{HOIF-var}
\begin{split}
\var^{\sff} (\hat{\tau}_{\unadj}) = & \ \frac{\pi_{0}}{\pi_{1}} \frac{1}{n} V_{n} (y (1)), \\
\var^{\sff} (\widehat{\IIFF}_{\unadj, 2, 2}) = & \ \frac{\pi_{0}}{\pi_{1}} \frac{1}{n^{2}} \sum_{i = 1}^{n} H_{i, i} \left( \frac{1}{\pi_{1}} - \frac{\pi_{0} + 1}{\pi_{1}} H_{i, i} \right) y_{i} (1)^{2} \\
& + \left( \frac{\pi_{0}}{\pi_{1}} \right)^{2} \frac{1}{n^{2}} \trace \left( \left( \hat{\bSigma}_{y} \hat{\bSigma}^{-} \right)^{2} \right) - \frac{\pi_{0}}{\pi_{1}} \frac{1}{n^{3}} \trace^{2} \left( \hat{\bSigma}_{y} \hat{\bSigma}^{-} \right) \\
& + \frac{\pi_{0}}{\pi_{1}} \frac{1}{n^{2}} \sum_{1 \leq i \neq j \leq n} H_{i, j} (1 - 2 H_{j, j}) y_{i} (1) y_{j} (1) + \mathsf{Rem}_{1}, \\
\cov^{\sff} (\hat{\tau}_{\unadj}, \widehat{\IIFF}_{\unadj, 2, 2}) = & \ \bar{\tau} \cdot \frac{\pi_{0}}{\pi_{1}} \frac{1}{n^{2}} \sum_{i = 1}^{n} H_{i, i} y_{i} (1) + \frac{\pi_{0}}{\pi_{1}} \frac{1}{n^{2}} \sum\limits_{1 \leq i \neq j \leq n} H_{i, j} y_{i} (1) y_{j} (1) + \mathsf{Rem}_{2}.
\end{split}
\end{equation}
Additionally, under Assumption \ref{as:regularity conditions}, $\mathsf{Rem}_{1}, \mathsf{Rem}_{2} = o (1 / n)$.
\end{lemma}

\begin{proof}
The following derivation follows immediately from Lemma \ref{lem:CRE}.
\begin{align*}
\var^{\sff} (\hat{\tau}_{\unadj}) & = \var^{\sff} \left( \frac{1}{n} \sum_{i = 1}^{n} \frac{t_{i}}{\pi_{1}} y_{i} (1) \right) \\
& = \frac{1}{n^{2}} \sum_{i = 1}^{n} \frac{y_{i} (1)^{2}}{\pi_{1}^{2}} \var (t_{i}) + \frac{1}{n^{2}} \sum_{1 \leq i \neq j \leq n} \frac{y_{i} (1) y_{j} (1)}{\pi_{1}^{2}} \cov (t_{i}, t_{j}) \\
& = \frac{1 - \pi_{1}}{\pi_{1}} \frac{1}{n^{2}} \sum_{i = 1}^{n} y_{i} (1)^{2} - \frac{1 - \pi_{1}}{\pi_{1}} \frac{1}{n^{2} (n - 1)} \sum_{1 \leq i \neq j \leq n} y_{i} (1) y_{j} (1) \\
& = \frac{\pi_{0}}{\pi_{1}} \frac{1}{n} \left\{ \frac{1}{n} \sum_{i = 1}^{n} y_{i} (1)^{2} - \frac{1}{n (n - 1)} \sum_{1 \leq i \neq j \leq n} y_{i} (1) y_{j} (1) \right\} \\
& = \frac{\pi_{0}}{\pi_{1}} \frac{1}{n} V_{n} (y (1)).
\end{align*}

Next,
\begin{align*}
& \ \var^{\sff} (\widehat{\IIFF}_{\unadj, 2, 2}) \equiv \var^{\sff} \left( \frac{1}{n} \sum_{1 \leq i \neq j \leq n} \left( \frac{t_{i}}{\pi_{1}} - 1 \right) H_{i, j} \frac{t_{j} y_{j} (1)}{\pi_{1}} \right) \\
= & \ \frac{1}{n^{2}} \sum_{1 \leq i \neq j \leq n} \sum_{1 \leq k \neq l \leq n} \cov^{\sff} \left\{ \left( \frac{t_{i}}{\pi_{1}} - 1 \right) H_{i, j} \frac{t_{j} y_{j} (1)}{\pi_{1}}, \left( \frac{t_{k}}{\pi_{1}} - 1 \right) H_{k, l} \frac{t_{l} y_{l} (1)}{\pi_{1}} \right\} \\
= & \ \frac{1}{n^{2}} \sum_{1 \leq i \neq j \leq n} \cov^{\sff} \left\{ \left( \frac{t_{i}}{\pi_{1}} - 1 \right) H_{i, j} \frac{t_{j} y_{j} (1)}{\pi_{1}}, \left( \frac{t_{i}}{\pi_{1}} - 1 \right) H_{i, j} \frac{t_{j} y_{j} (1)}{\pi_{1}} \right\} \\
& + \frac{1}{n^{2}} \sum_{1 \leq i \neq j \leq n} \cov^{\sff} \left\{ \left( \frac{t_{i}}{\pi_{1}} - 1 \right) H_{i, j} \frac{t_{j} y_{j} (1)}{\pi_{1}}, \frac{t_{i} y_{i} (1)}{\pi_{1}} H_{i, j} \left( \frac{t_{j}}{\pi_{1}} - 1 \right) \right\} \\
& + \frac{1}{n^{2}} \sum_{1 \leq i \neq j \neq k \leq n} \cov^{\sff} \left\{ \left( \frac{t_{i}}{\pi_{1}} - 1 \right) H_{i, j} \frac{t_{j} y_{j} (1)}{\pi_{1}}, \left( \frac{t_{i}}{\pi_{1}} - 1 \right) H_{i, k} \frac{t_{k} y_{k} (1)}{\pi_{1}} + \left( \frac{t_{j}}{\pi_{1}} - 1 \right) H_{j, k} \frac{t_{k} y_{k} (1)}{\pi_{1}} \right\} \\
& + \frac{1}{n^{2}} \sum_{1 \leq i \neq j \neq k \leq n} \cov^{\sff} \left\{ \left( \frac{t_{i}}{\pi_{1}} - 1 \right) H_{i, j} \frac{t_{j} y_{j} (1)}{\pi_{1}}, \frac{t_{i} y_{i} (1)}{\pi_{1}} H_{i, k} \left( \frac{t_{k}}{\pi_{1}} - 1 \right) + \frac{t_{j} y_{j} (1)}{\pi_{1}} H_{j, k} \left( \frac{t_{k}}{\pi_{1}} - 1 \right)  \right\} \\
& + \frac{1}{n^{2}} \sum_{1 \leq i \neq j \neq k \neq l \leq n} \cov^{\sff} \left\{ \left( \frac{t_{i}}{\pi_{1}} - 1 \right) H_{i, j} \frac{t_{j} y_{j} (1)}{\pi_{1}}, \left( \frac{t_{k}}{\pi_{1}} - 1 \right) H_{k, l} \frac{t_{l} y_{l} (1)}{\pi_{1}} \right\} \\
= & \ \frac{1}{\pi_{1}^{2}} \frac{1}{n^{2}} \sum_{1 \leq i \neq j \leq n} H_{i, j}^{2} \left\{ y_{j} (1)^{2} \var \left( \left( \frac{t_{i}}{\pi_{1}} - 1 \right) t_{j} \right) + y_{i} (1) y_{j} (1) \cov \left( \left( \frac{t_{i}}{\pi_{1}} - 1 \right) t_{j}, t_{i} \left( \frac{t_{j}}{\pi_{1}} - 1 \right) \right) \right\} \\
& + \frac{1}{\pi_{1}^{2}} \frac{1}{n^{2}} \sum_{1 \leq i \neq j \neq k \leq n} H_{i, j} H_{i, k} \left\{  y_{j} (1)  y_{k} (1) \cov \left( \left( \frac{t_{i}}{\pi_{1}} - 1 \right) t_{j}, \left( \frac{t_{i}}{\pi_{1}} - 1 \right) t_{k} \right)  \right\} \\
& + \frac{1}{\pi_{1}^{2}} \frac{1}{n^{2}} \sum_{1 \leq i \neq j \neq k \leq n} H_{i, j} H_{j, k}  \left\{  y_{j} (1) y_{k} (1) \cov \left( \left( \frac{t_{i}}{\pi_{1}} - 1 \right) t_{j}, \left( \frac{t_{j}}{\pi_{1}} - 1 \right) t_{k} \right)  \right\} \\ 
& + \frac{1}{\pi_{1}^{2}} \frac{1}{n^{2}} \sum_{1 \leq i \neq j \neq k \leq n} H_{i, j} H_{i, k} \left\{ y_{i} (1) y_{j} (1) \cov \left( \left( \frac{t_{i}}{\pi_{1}} - 1 \right) t_{j}, t_{i} \left( \frac{t_{k}}{\pi_{1}} - 1 \right) \right) \right\} \\
& + \frac{1}{\pi_{1}^{2}} \frac{1}{n^{2}} \sum_{1 \leq i \neq j \neq k \leq n} H_{i, j} H_{j, k} \left\{ y_{j} (1)^{2}   \cov \left( \left( \frac{t_{i}}{\pi_{1}} - 1 \right) t_{j}, t_{j} \left( \frac{t_{k}}{\pi_{1}} - 1 \right) \right) \right\} \\
& + \frac{1}{\pi_{1}^{2}} \frac{1}{n^{2}} \sum_{1 \leq i \neq j \neq k \neq l \leq n} H_{i, j} H_{k, l} y_{j} (1) y_{l} (1) \cov \left( \left( \frac{t_{i}}{\pi_{1}} - 1 \right) t_{j}, \left( \frac{t_{k}}{\pi_{1}} - 1 \right) t_{l} \right) \\
\eqqcolon & \ V_{1} \var \left( \left( \frac{t_{1}}{\pi_{1}} - 1 \right) t_{2} \right) + V_{2} \cov \left( \left( \frac{t_{1}}{\pi_{1}} - 1 \right) t_{2}, \left( \frac{t_{2}}{\pi_{1}} - 1 \right) t_{1} \right) + V_{3} \cov \left( \left( \frac{t_{1}}{\pi_{1}} - 1 \right) t_{2}, \left( \frac{t_{1}}{\pi_{1}} - 1 \right) t_{3} \right) \\
& + V_{4} \cov \left( \left( \frac{t_{1}}{\pi_{1}} - 1 \right) t_{2}, \left( \frac{t_{2}}{\pi_{1}} - 1 \right) t_{3} \right) + V_{5} \cov \left( \left( \frac{t_{1}}{\pi_{1}} - 1 \right) t_{2}, \left( \frac{t_{3}}{\pi_{1}} - 1 \right) t_{1} \right) \\
& + V_{6} \cov \left( \left( \frac{t_{1}}{\pi_{1}} - 1 \right) t_{2}, \left( \frac{t_{3}}{\pi_{1}} - 1 \right) t_{2} \right) + V_{7} \cov \left( \left( \frac{t_{1}}{\pi_{1}} - 1 \right) t_{2}, \left( \frac{t_{3}}{\pi_{1}} - 1 \right) t_{4} \right).
\end{align*}

We next simplify $V_{1}$ -- $V_{7}$ step by step.
\begin{align*}
V_{1} & = \frac{1}{\pi_{1}^{2}} \frac{1}{n^{2}} \sum_{i = 1}^{n} \sum_{j = 1}^{n} H_{j, i} H_{i, j} y_{j} (1)^{2} - \frac{1}{\pi_{1}^{2}} \frac{1}{n^{2}} \sum_{i = 1}^{n} H_{i, i}^{2} y_{i} (1)^{2} \\
& = \frac{1}{\pi_{1}^{2}} \frac{1}{n^{2}} \sum_{j = 1}^{n} H_{j, j} y_{j} (1)^{2} - \frac{1}{\pi_{1}^{2}} \frac{1}{n^{2}} \sum_{i = 1}^{n} H_{i, i}^{2} y_{i} (1)^{2} \\
& = \frac{1}{\pi_{1}^{2}} \frac{1}{n^{2}} \sum_{i = 1}^{n} H_{i, i} (1 - H_{i, i}) y_{i} (1)^{2}.
\end{align*}

Let $\hat{\bSigma}_{y} \coloneqq \sum_{i = 1}^{n} (\bx_{i} - \bar{\bx}) (\bx_{i} - \bar{\bx})^{\top} y_{i} (1)$ and $\hat{\bSigma} \coloneqq \sum_{i = 1}^{n} (\bx_{i} - \bar{\bx}) (\bx_{i} - \bar{\bx})^{\top}$. Then
\begin{align*}
V_{2} = & \ \frac{1}{\pi_{1}^{2}} \frac{1}{n} \sum_{i = 1}^{n} y_{i} (1) (\bx_{i} - \bar{\bx})^{\top} \hat{\bSigma}^{-} \left\{ \begin{array}{c}
\dfrac{1}{n} \sum\limits_{j = 1}^{n}(\bx_{j} - \bar{\bx}) (\bx_{j} - \bar{\bx})^{\top} y_{j} (1) \\
- \, \dfrac{1}{n} (\bx_{i} - \bar{\bx}) (\bx_{i} - \bar{\bx})^{\top} y_{i} (1)
\end{array} \right\} \hat{\bSigma}^{-} (\bx_{i} - \bar{\bx}) \\
= & \ \frac{1}{\pi_{1}^{2}} \frac{1}{n^{2}} \trace \left( \hat{\bSigma}_{y} \hat{\bSigma}^{-} \hat{\bSigma}_{y} \hat{\bSigma}^{-} \right) - \frac{1}{\pi_{1}^{2}} \frac{1}{n^{2}} \sum_{i = 1}^{n} H_{i, i}^{2} y_{i} (1)^{2}.
\end{align*}

\begin{align*}
V_{3} & = \frac{1}{\pi_{1}^{2}} \frac{1}{n^{2}} \sum_{1 \leq i \neq j \neq k \leq n} H_{i, j} H_{i, k} y_{j} (1) y_{k} (1) \\
& = \frac{1}{\pi_{1}^{2}} \frac{1}{n^{2}} \sum_{1 \leq j \neq k \leq n} y_{j} (1) \left\{ \sum_{i = 1}^{n} H_{j, i} H_{i, k} - H_{j, j} H_{j, k} - H_{j, k} H_{k, k} \right\} y_{k} (1) \\
& = \frac{1}{\pi_{1}^{2}} \frac{1}{n^{2}} \sum_{1 \leq i \neq j \leq n} H_{i, j} y_{i} (1) y_{j} (1) - \frac{1}{\pi_{1}^{2}} \frac{1}{n^{2}} \sum_{1 \leq i \neq j \leq n} H_{i, i} H_{i, j} y_{i} (1) y_{j} (1) - \frac{1}{\pi_{1}^{2}} \frac{1}{n^{2}} \sum_{1 \leq i \neq j \leq n} H_{i, j} H_{j, j} y_{i} (1) y_{j} (1) \\
& = \frac{1}{\pi_{1}^{2}} \frac{1}{n^{2}} \sum_{1 \leq i \neq j \leq n} H_{i, j} y_{i} (1) y_{j} (1) - \frac{1}{\pi_{1}^{2}} \frac{2}{n^{2}} \sum_{1 \leq i \neq j \leq n} H_{i, j} H_{j, j} y_{i} (1) y_{j} (1) \\
& = \frac{1}{\pi_{1}^{2}} \frac{1}{n^{2}} \sum_{1 \leq i \neq j \leq n} H_{i, j} (1 - 2 H_{j, j}) y_{i} (1) y_{j} (1).
\end{align*}

\begin{align*}
V_{4} = & \ \frac{1}{\pi_{1}^{2}} \frac{1}{n^{2}} \sum_{1 \leq i \neq j \neq k \leq n} H_{i, j} H_{j, k} y_{j} (1) y_{k} (1) \\
= & \ \frac{1}{\pi_{1}^{2}} \frac{1}{n^{2}} \sum_{1 \leq j \neq k \leq n} \sum_{i \neq \{j, k\}} H_{i, j} H_{j, k} y_{j} (1) y_{k} (1) \\
= & \ \frac{1}{\pi_{1}^{2}} \frac{1}{n^{2}} \sum_{1 \leq j \neq k \leq n} \left( \sum_{i = 1}^{n} H_{i, j} H_{j, k} - H_{j, j} H_{j, k} - H_{j, k}^{2} \right) y_{j} (1) y_{k} (1) \\
= & - \frac{1}{\pi_{1}^{2}} \frac{1}{n^{2}} \sum_{1 \leq j \neq k \leq n} \left( H_{j, j} H_{j, k} + H_{j, k}^{2} \right) y_{j} (1) y_{k} (1) \\
= & - \frac{1}{\pi_{1}^{2}} \frac{1}{n^{2}} \sum_{1 \leq i \neq j \leq n} H_{i, i} H_{i, j} y_{i} (1) y_{j} (1) - \frac{1}{\pi_{1}^{2}} \frac{1}{n^{2}} \sum_{i = 1}^{n} \sum_{j = 1}^{n} y_{i} (1) H_{i, j} y_{j} (1) H_{j, i} + \frac{1}{\pi_{1}^{2}} \frac{1}{n^{2}} \sum_{i = 1}^{n} y_{i} (1)^{2} H_{i, i}^{2} \\
= & - \frac{1}{\pi_{1}^{2}} \frac{1}{n^{2}} \sum_{1 \leq i \neq j \leq n} H_{i, i} H_{i, j} y_{i} (1) y_{j} (1) - \frac{1}{\pi_{1}^{2}} \frac{1}{n^{2}} \sum_{i = 1}^{n} \sum_{j = 1}^{n} H_{i, j}^{2} y_{i} (1) y_{j} (1) + \frac{1}{\pi_{1}^{2}} \frac{1}{n^{2}} \sum_{i = 1}^{n} y_{i} (1)^{2} H_{i, i}^{2} \\
= & - \frac{1}{\pi_{1}^{2}} \frac{1}{n^{2}} \sum_{1 \leq i \neq j \leq n} H_{i, i} H_{i, j} y_{i} (1) y_{j} (1) - \frac{1}{\pi_{1}^{2}} \frac{1}{n^{2}} \trace \left( \hat{\bSigma}_{y} \hat{\bSigma}^{-} \hat{\bSigma}_{y} \hat{\bSigma}^{-} \right) + \frac{1}{\pi_{1}^{2}} \frac{1}{n^{2}} \sum_{i = 1}^{n} y_{i} (1)^{2} H_{i, i}^{2}.
\end{align*}

By symmetry, $V_{5} \equiv V_{4}$, so
\begin{align*}
V_{5} = & \ \frac{1}{\pi_{1}^{2}} \frac{1}{n^{2}} \sum_{1 \leq i \neq j \neq k \leq n} H_{i, j} H_{i, k} y_{i} (1) y_{j} (1) \\
= & - \frac{1}{\pi_{1}^{2}} \frac{1}{n^{2}} \sum_{1 \leq i \neq j \leq n} H_{i, i} H_{i, j} y_{i} (1) y_{j} (1) - \frac{1}{\pi_{1}^{2}} \frac{1}{n^{2}} \trace \left( \hat{\bSigma}_{y} \hat{\bSigma}^{-} \hat{\bSigma}_{y} \hat{\bSigma}^{-} \right) + \frac{1}{\pi_{1}^{2}} \frac{1}{n^{2}} \sum_{i = 1}^{n} y_{i} (1)^{2} H_{i, i}^{2}.
\end{align*}

\begin{align*}
V_{6} = & \ \frac{1}{\pi_{1}^{2}} \frac{1}{n^{2}} \sum_{1 \leq i \neq j \neq k \leq n} H_{i, j} H_{j, k} y_{j} (1)^{2} = \frac{1}{\pi_{1}^{2}} \frac{1}{n^{2}} \sum_{1 \leq i \neq j \neq k \leq n} H_{i, j} H_{k, i} y_{i} (1)^{2} \\
= & \ \frac{1}{\pi_{1}^{2}} \frac{1}{n^{2}} \sum_{i = 1}^{n} y_{i} (1)^{2} \sum_{j \neq k \neq i} H_{i, j} H_{k, i} \\
= & \ \frac{1}{\pi_{1}^{2}} \frac{1}{n^{2}} \sum_{i = 1}^{n} y_{i} (1)^{2} \left( \sum_{1 \leq j \neq k \leq n} H_{i, j} H_{k, i} - 2 \sum_{k \neq i} H_{i, i} H_{k, i} \right) \\
= & \ \frac{1}{\pi_{1}^{2}} \frac{1}{n^{2}} \sum_{i = 1}^{n} y_{i} (1)^{2} \left( \sum_{j = 1}^{n} \sum_{k = 1}^{n} H_{i, j} H_{k, i} - \sum_{j = 1}^{n} H_{i, j}^{2} - 2 \sum_{k = 1}^{n} H_{i, i} H_{k, i} + 2 H_{i, i}^{2} \right) \\
= & \ \frac{1}{\pi_{1}^{2}} \frac{1}{n^{2}} \sum_{i = 1}^{n} y_{i} (1)^{2} \left( 2 H_{i, i}^{2} - \sum_{j = 1}^{n} H_{i, j}^{2} \right) \\
= & \ \frac{1}{\pi_{1}^{2}} \frac{1}{n^{2}} \sum_{i = 1}^{n} y_{i} (1)^{2} \left( 2 H_{i, i}^{2} - H_{i, i} \right) \\
= & \ \frac{1}{\pi_{1}^{2}} \frac{1}{n^{2}} \sum_{i = 1}^{n} H_{i, i} \left( 2 H_{i, i} - 1 \right) y_{i} (1)^{2}.
\end{align*}

\begin{align*}
V_{7} = & \ \frac{1}{\pi_{1}^{2}} \frac{1}{n^{2}} \sum_{1 \leq i \neq j \neq k \neq l \leq n} H_{i, j} H_{k, l} y_{i} (1) y_{l} (1) \\
= & \ \frac{1}{\pi_{1}^{2}} \frac{1}{n^{2}} \sum_{1 \leq i \neq l \leq n} y_{i} (1) y_{l} (1) \sum_{j \neq k \neq \{i, l\}} H_{i, j} H_{k, l} \\
= & \ \frac{1}{\pi_{1}^{2}} \frac{1}{n^{2}} \sum_{1 \leq i \neq l \leq n} y_{i} (1) y_{l} (1) \left( \begin{array}{c}
\sum\limits_{1 \leq j \neq k \leq n} H_{i, j} H_{k, l} - \sum\limits_{k \not\in \{i, l\}} H_{i, i} H_{k, l} - \sum\limits_{k \not\in \{i, l\}} H_{i, k} H_{l, l} \\
- \, \sum\limits_{k \not\in \{i, l\}} H_{i, k} H_{i, l} - \sum\limits_{k \not\in \{i, l\}} H_{i, l} H_{k, l} - H_{i, i} H_{l, l} - H_{i, l}^{2}
\end{array} \right) \\
= & - \frac{1}{\pi_{1}^{2}} \frac{1}{n^{2}} \sum_{1 \leq i \neq j \leq n} y_{i} (1) y_{j} (1) (H_{i, i} H_{j, j} + H_{i, j}^{2}) \\
& - \frac{1}{\pi_{1}^{2}} \frac{1}{n^{2}} \sum_{1 \leq i \neq j \leq n} y_{i} (1) y_{j} (1) \left( H_{i, i} \sum_{k = 1}^{n} H_{k, j} - H_{i, i} H_{i, j} - H_{i, i} H_{j, j} \right) \\
& - \frac{1}{\pi_{1}^{2}} \frac{1}{n^{2}} \sum_{1 \leq i \neq j \leq n} y_{i} (1) y_{j} (1) \left( \sum_{k = 1}^{n} H_{i, k} H_{j, j} - H_{i, i} H_{j, j} - H_{i, j} H_{j, j} \right) \\
& - \frac{1}{\pi_{1}^{2}} \frac{1}{n^{2}} \sum_{1 \leq i \neq j \leq n} y_{i} (1) y_{j} (1) \left( \sum_{k = 1}^{n} H_{i, k} H_{i, j} - H_{i, i} H_{i, j} - H_{i, j} H_{i, j} \right) \\
& - \frac{1}{\pi_{1}^{2}} \frac{1}{n^{2}} \sum_{1 \leq i \neq j \leq n} y_{i} (1) y_{j} (1) \left( \sum_{k = 1}^{n} H_{i, j} H_{k, j} - H_{i, j} H_{i, j} - H_{i, j} H_{j, j} \right) \\
& + \frac{1}{\pi_{1}^{2}} \frac{1}{n^{2}} \sum_{1 \leq i \neq j \leq n} y_{i} (1) y_{j} (1) \left\{ \sum_{k = 1}^{n} H_{i, k} \sum_{l = 1}^{n} H_{l, j} - \sum_{k = 1}^{n} H_{i, k} H_{k, j} \right\} \\
= & - \frac{1}{\pi_{1}^{2}} \frac{1}{n^{2}} \sum_{1 \leq i \neq j \leq n} y_{i} (1) y_{j} (1) \left( \begin{array}{c}
H_{i, i} H_{j, j} + H_{i, j}^{2} - H_{i, i} H_{i, j} - H_{i, i} H_{j, j} - H_{i, i} H_{j, j} - H_{i, j} H_{j, j} \\
- \, H_{i, i} H_{i, j} - H_{i, j} H_{i, j} - H_{i, j}^{2} - H_{i, j}^{2} + H_{i, j}
\end{array} \right) \\
= & \ \frac{1}{\pi_{1}^{2}} \frac{1}{n^{2}} \sum_{1 \leq i \neq j \leq n} y_{i} (1) y_{j} (1) \left( 4 H_{i, i} H_{i, j} + H_{i, i} H_{j, j} + H_{i, j}^{2} - H_{i, j} \right) \\
= & \ \frac{1}{\pi_{1}^{2}} \frac{1}{n^{2}} \sum_{1 \leq i \neq j \leq n} y_{i} (1) y_{j} (1) H_{i, j} \left( 4 H_{i, i} - 1 \right) + \frac{1}{\pi_{1}^{2}} \frac{1}{n^{2}} \sum_{i = 1}^{n} y_{i} (1) H_{i, i} \sum_{j = 1}^{n} y_{j} (1) H_{j, j} \\
& - \frac{1}{\pi_{1}^{2}} \frac{1}{n^{2}} \sum_{i = 1}^{n} y_{i} (1)^{2} H_{i, i}^{2} + \frac{1}{\pi_{1}^{2}} \frac{1}{n^{2}} \sum_{i = 1}^{n} \sum_{j = 1}^{n} y_{i} (1) y_{j} (1) H_{i, j}^{2} - \frac{1}{\pi_{1}^{2}} \frac{1}{n^{2}} \sum_{i = 1}^{n} y_{i} (1)^{2} H_{i, i}^{2} \\
= & \ \frac{1}{\pi_{1}^{2}} \frac{1}{n^{2}} \sum_{1 \leq i \neq j \leq n} y_{i} (1) y_{j} (1) H_{i, j} \left( 4 H_{i, i} - 1 \right) + \frac{1}{\pi_{1}^{2}} \frac{1}{n^{2}} \left\{ \trace^{2} \left( \hat{\bSigma}_{y} \hat{\bSigma}^{-} \right) + \trace \left( \hat{\bSigma}_{y} \hat{\bSigma}^{-} \hat{\bSigma}_{y} \hat{\bSigma}^{-} \right) \right\} \\
& - \frac{2}{\pi_{1}^{2}} \frac{1}{n^{2}} \sum_{i = 1}^{n} H_{i, i}^{2} y_{i} (1)^{2}.
\end{align*}
Combining $V_{1}$ -- $V_{7}$, we have
\begin{align*}
& \ \var^{\sff} (\widehat{\IIFF}_{\unadj, 2, 2}) \\
= & \ \frac{\pi_{0}}{\pi_{1}^{2}} \frac{1}{n^{2}} \sum_{i = 1}^{n} H_{i, i} (1 - H_{i, i}) y_{i} (1)^{2} \left\{ 1 + O \left( \frac{1}{n} \right) \right\} + \left( \frac{\pi_{0}}{\pi_{1}} \right)^{2} \frac{1}{n^{2}} \trace \left( \hat{\bSigma}_{y} \hat{\bSigma}^{-} \hat{\bSigma}_{y} \hat{\bSigma}^{-} \right) \left\{ 1 + O \left( \frac{1}{n} \right) \right\} \\
& - \left( \frac{\pi_{0}}{\pi_{1}} \right)^{2} \frac{1}{n^{2}} \sum_{i = 1}^{n} H_{i, i}^{2} y_{i} (1)^{2} \left\{ 1 + O \left( \frac{1}{n} \right) \right\} + \frac{\pi_{0}}{\pi_{1}} \frac{1}{n^{2}} \sum_{1 \leq i \neq j \leq n} H_{i, j} (1 - 2 H_{j, j}) y_{i} (1) y_{j} (1) \left\{ 1 + O \left( \frac{1}{n} \right) \right\} \\
& + 4 \left( \frac{\pi_{0}}{\pi_{1}} \right)^{2} \frac{1}{n^{2} (n - 2)} \sum_{1 \leq i \neq j \leq n} H_{i, i} H_{i, j} y_{i} (1) y_{j} (1) \left\{ 1 + O \left( \frac{1}{n} \right) \right\} \\
& + 4 \left( \frac{\pi_{0}}{\pi_{1}} \right)^{2} \frac{1}{n^{2} (n - 2)} \trace \left( \hat{\bSigma}_{y} \hat{\bSigma}^{-} \hat{\bSigma}_{y} \hat{\bSigma}^{-} \right) \left\{ 1 + O \left( \frac{1}{n} \right) \right\} \\
& - 4 \left( \frac{\pi_{0}}{\pi_{1}} \right)^{2} \frac{1}{n^{2} (n - 2)} \sum_{i = 1}^{n} H_{i, i}^{2} y_{i} (1)^{2} \left\{ 1 + O \left( \frac{1}{n} \right) \right\} \\
& - \frac{\pi_{0}}{\pi_{1}^{2}} \frac{1}{n (n - 1) (n - 2)} \sum_{i = 1}^{n} H_{i, i} (2 H_{i, i} - 1) y_{i} (1)^{2} \left\{ 1 + O \left( \frac{1}{n} \right) \right\} \\
& - \frac{\pi_{0}}{\pi_{1}} \frac{1}{n^{2} (n - 2)} \sum_{1 \leq i \neq j \leq n} H_{i, j} \left( 4 H_{i, i} - 1 \right) y_{i} (1) y_{j} (1) \left\{ 1 + O \left( \frac{1}{n} \right) \right\} \\
& - \frac{\pi_{0}}{\pi_{1}} \frac{1}{n^{2} (n - 2)} \left\{ \trace^{2} \left( \hat{\bSigma}_{y} \hat{\bSigma}^{-} \right) + \trace \left( \hat{\bSigma}_{y} \hat{\bSigma}^{-} \hat{\bSigma}_{y} \hat{\bSigma}^{-} \right) \right\} \left\{ 1 + O \left( \frac{1}{n} \right) \right\} \\
& + 2 \frac{\pi_{0}}{\pi_{1}} \frac{1}{n^{2} (n - 2)} \sum_{i = 1}^{n} H_{i, i}^{2} y_{i} (1)^{2} \left\{ 1 + O \left( \frac{1}{n} \right) \right\} \\
= & \ \frac{\pi_{0}}{\pi_{1}} \frac{1}{n^{2}} \sum_{i = 1}^{n} H_{i, i} \left( \frac{1}{\pi_{1}} - \frac{\pi_{0} + 1}{\pi_{1}} H_{i, i} \right) y_{i} (1)^{2} + \left( \frac{\pi_{0}}{\pi_{1}} \right)^{2} \frac{1}{n^{2}} \trace \left( \left( \hat{\bSigma}_{y} \hat{\bSigma}^{-} \right)^{2} \right) - \frac{\pi_{0}}{\pi_{1}} \frac{1}{n^{3}} \trace^{2} \left( \hat{\bSigma}_{y} \hat{\bSigma}^{-} \right) \\
& + \frac{\pi_{0}}{\pi_{1}} \frac{1}{n^{2}} \sum_{1 \leq i \neq j \leq n} H_{i, j} (1 - 2 H_{j, j}) y_{i} (1) y_{j} (1) + \mathsf{Rem}_{1},
\end{align*}
where it is easy to see that $\mathsf{Rem}_{1} = o (1 / n)$ under Assumption \ref{as:regularity conditions} as $n \rightarrow \infty$.

Once again by Lemma \ref{lem:var intermediate},
\begin{align*}
& \ \cov^{\sff} \left( \hat{\tau}_{\unadj}, \widehat{\IIFF}_{\unadj, 2, 2} \right) \equiv \cov^{\sff} \left( \frac{1}{n} \sum_{k = 1}^{n} \frac{t_{k}}{\pi_{1}} y_{k} (1), \frac{1}{n} \sum_{1 \leq i \neq j \leq n} \left( \frac{t_{i}}{\pi_{1}} - 1 \right) H_{i, j} \frac{t_{j} y_{j} (1)}{\pi_{1}} \right) \\
= & \ \frac{1}{\pi_{1}^{2}} \frac{1}{n^{2}} \sum_{1 \leq i \neq j \leq n} \sum_{k = 1}^{n} \cov^{\sff} \left( \left( \frac{t_{i}}{\pi_{1}} - 1 \right) H_{i, j} t_{j} y_{j} (1), t_{k} y_{k} (1) \right) \\
= & \ \frac{1}{\pi_{1}^{2}} \frac{1}{n^{2}} \sum_{1 \leq i \neq j \neq k \leq n} H_{i, j} y_{j} (1) y_{k} (1) \cov \left( \left( \frac{t_{i}}{\pi_{1}} - 1 \right) t_{j}, t_{k} \right) \\
& + \frac{1}{\pi_{1}^{2}} \frac{1}{n^{2}} \sum_{1 \leq i \neq j \leq n} H_{i, j} y_{i} (1) y_{j} (1) \cov \left( \left( \frac{t_{i}}{\pi_{1}} - 1 \right) t_{j}, t_{i} \right) \\
& + \frac{1}{\pi_{1}^{2}} \frac{1}{n^{2}} \sum_{1 \leq i \neq j \leq n} H_{i, j} y_{j} (1)^{2} \cov \left( \left( \frac{t_{i}}{\pi_{1}} - 1 \right) t_{j}, t_{j} \right) \\
\eqqcolon & \ U_{1} \cov \left( \left( \frac{t_{1}}{\pi_{1}} - 1 \right) t_{2}, t_{3} \right) + U_{2} \cov \left( \left( \frac{t_{1}}{\pi_{1}} - 1 \right) t_{2}, t_{1} \right) + U_{3} \cov \left( \left( \frac{t_{1}}{\pi_{1}} - 1 \right) t_{2}, t_{2} \right) \\
= & \ \bar{\tau} \frac{\pi_{0}}{\pi_{1}} \frac{1}{n^{2}} \sum_{i = 1}^{n} H_{i, i} y_{i} (1) \left\{ 1 + O \left( \frac{1}{n} \right) \right\} + \frac{\pi_{0}}{\pi_{1}} \frac{1}{n^{2}} \sum_{1 \leq i \neq j \leq n} H_{i, j} y_{i} (1) y_{j} (1) \left\{ 1 + O \left( \frac{1}{n} \right) \right\} \\
& + \left( \frac{\pi_{0}}{\pi_{1}} \right) \left(  \frac{\pi_{0}}{\pi_{1}} - 1 \right) \frac{1}{n^{3}} \sum_{i = 1}^{n} H_{i, i} y_{i} (1)^{2} \left\{ 1 + O \left( \frac{1}{n} \right) \right\} \\
= & \ \bar{\tau} \cdot \frac{\pi_{0}}{\pi_{1}} \frac{1}{n^{2}} \trace \left( \hat{\bSigma}_{y} \hat{\bSigma}^{-} \right) + \frac{\pi_{0}}{\pi_{1}} \frac{1}{n^{2}} \sum_{1 \leq i \neq j \leq n} H_{i, j} y_{i} (1) y_{j} (1) + \mathsf{Rem}_{2},
\end{align*}
where it is easy to see that $\mathsf{Rem}_{2} = o (1 / n)$ under Assumption \ref{as:regularity conditions} as $n \rightarrow \infty$.

Here
\begin{align*}
U_{1} = & - \bar{\tau} \cdot \frac{1}{\pi_{1}^{2}} \frac{1}{n} \sum_{i = 1}^{n} H_{i, i} y_{i} (1) + \frac{1}{\pi_{1}^{2}} \frac{1}{n^{2}} \sum_{i = 1}^{n} H_{i, i} y_{i} (1)^{2} - \frac{1}{\pi_{1}^{2}} \frac{1}{n^{2}} \sum_{1 \leq i \neq j \leq n} H_{i, j} y_{i} (1) y_{j} (1),
\end{align*}
\begin{align*}
U_{2} = \frac{1}{\pi_{1}^{2}} \frac{1}{n^{2}} \sum_{1 \leq i \neq j \leq n} H_{i, j} y_{i} (1) y_{j} (1),
\end{align*}
and
\begin{align*}
U_{3} = - \frac{1}{\pi_{1}^{2}} \frac{1}{n^{2}} \sum_{i = 1}^{n} H_{i, i} y_{i} (1)^{2}.
\end{align*}
\end{proof}

\begin{lemma}
\label{lem:var intermediate}
Under CRE, the following assertions hold:
\begin{align*}
\var \left( \left( \frac{t_{i}}{\pi_{1}} - 1 \right) t_{j} \right) = & \ \pi_{1} + \left( \frac{1 - \pi_{1}}{\pi_{1}} - 1 \right) \frac{n \pi_{1} - 1}{n - 1} - \frac{(1 - \pi_{1})^{2}}{(n - 1)^{2}} \\
= & \ 1 - \pi_{1} - \frac{(1 - \pi_{1}) (1 - 2 \pi_{1})}{\pi_{1} (n - 1)} - \frac{(1 - \pi_{1})^{2}}{(n - 1)^{2}}, \\
\cov \left( \left( \frac{t_{i}}{\pi_{1}} - 1 \right) t_{j}, t_{k} \right) = & \ \frac{(n \pi_{1} - 1) (n \pi_{1} - 2)}{(n - 1) (n - 2)} - 2 \pi_{1} \frac{n \pi_{1} - 1}{n - 1} + \pi_{1}^{2} \\
= & - \frac{n \pi_{1} (1 - \pi_{1})}{(n - 1) (n - 2)} + \frac{2 (1 - \pi_{1})^{2}}{(n - 1) (n - 2)}, \\
\cov \left( \left( \frac{t_{i}}{\pi_{1}} - 1 \right) t_{j}, t_{i} \right) = & \ (1 - 2 \pi_{1}) \frac{n \pi_{1} - 1}{n - 1} + \pi_{1}^{2} = \pi_{1} (1 - \pi_{1}) - \frac{(1 - \pi_{1}) (1 - 2 \pi_{1})}{n - 1}, \\
\cov \left( \left( \frac{t_{i}}{\pi_{1}} - 1 \right) t_{j}, t_{j} \right) = & \ (1 - \pi_{1}) \frac{n \pi_{1} - 1}{n - 1} - \pi_{1} (1 - \pi_{1}) = - \frac{(1 - \pi_{1})^{2}}{n - 1}, \\
\cov \left( \left( \frac{t_{i}}{\pi_{1}} - 1 \right) t_{j}, t_{i} \left( \frac{t_{j}}{\pi_{1}} - 1 \right) \right) = & \ \frac{(1 - \pi_{1})^{2}}{\pi_{1}} \frac{n \pi_{1} - 1}{n - 1} - \frac{(1 - \pi_{1})^{2}}{(n - 1)^{2}} \\
= & \ (1 - \pi_{1})^{2} - \frac{(1 - \pi_{1})^{3}}{\pi_{1} (n - 1)} - \frac{(1 - \pi_{1})^{2}}{(n - 1)^{2}}, \\
\cov \left( \left( \frac{t_{i}}{\pi_{1}} - 1 \right) t_{j}, \left( \frac{t_{i}}{\pi_{1}} - 1 \right) t_{k} \right) = & \ \left( \frac{1}{\pi_{1}} - 2 \right) \frac{(n \pi_{1} - 1) (n \pi_{1} - 2)}{(n - 1) (n - 2)} + \pi_{1} \frac{n \pi_{1} - 1}{n - 1} - \frac{(1 - \pi_{1})^{2}}{(n - 1)^{2}} \\
= & \ \pi_{1} (1 - \pi_{1}) - \frac{(1 - \pi_{1})^{2}}{n - 1} - \frac{2 (1 - \pi_{1}) (1 - 2 \pi_{1})}{n - 2} \\
& + \left( \frac{2}{\pi_{1}} - 5 + \frac{1}{n - 1} \right) \frac{(1 - \pi_{1})^{2}}{(n - 1) (n - 2)}, \\
\cov \left( \left( \frac{t_{i}}{\pi_{1}} - 1 \right) t_{j}, t_{i} \left( \frac{t_{k}}{\pi_{1}} - 1 \right) \right) = & - 2 \frac{(1 - \pi_{1})^{2}}{\pi_{1}} \frac{n \pi_{1} - 1}{n - 1} \frac{1}{n - 2} - \frac{(1 - \pi_{1})^{2}}{(n - 1)^{2}} \\
= & - 2 \frac{(1 - \pi_{1})^{2}}{n - 2} + \left\{ \frac{2}{\pi_{1}} - 3 + \frac{1}{n - 1} \right\} \frac{(1 - \pi_{1})^{2}}{(n - 1) (n - 2)}, \\
\cov \left( \left( \frac{t_{i}}{\pi_{1}} - 1 \right) t_{j}, t_{j} \left( \frac{t_{k}}{\pi_{1}} - 1 \right) \right) = & \ \frac{1}{\pi_{1}} \frac{(n \pi_{1} - 1) (n \pi_{1} - 2)}{(n - 1) (n - 2)} - 2 \frac{n \pi_{1} - 1}{n - 1} + \pi_{1} - \frac{(1 - \pi_{1})^{2}}{(n - 1)^{2}} \\
= & - \frac{n (1 - \pi_{1})}{(n - 1) (n - 2)} - \frac{(1 - \pi_{1})^{2}}{(n - 1)^{2}} + \frac{2 (1 - \pi_{1})^{2}}{\pi_{1} (n - 1) (n - 2)} \\
\cov \left( \left( \frac{t_{i}}{\pi_{1}} - 1 \right) t_{j}, \left( \frac{t_{k}}{\pi_{1}} - 1 \right) t_{l} \right) = & \ \frac{1}{\pi_{1}} \frac{(n \pi_{1} - 1) (n \pi_{1} - 2) (n \pi_{1} - 3)}{(n - 1) (n - 2) (n - 3)} - 2 \frac{(n \pi_{1} - 1) (n \pi_{1} - 2)}{(n - 1) (n - 2)} \\
& + \pi_{1} \frac{n \pi_{1} - 1}{n - 1} - \frac{(1 - \pi_{1})^{2}}{(n - 1)^{2}} \\
= & - \frac{\pi_1 (1 - \pi_1)}{n - 2} + \frac{3 (2 - 3 \pi_{1}) (1 - \pi_{1})}{(n - 2) (n - 3)} \\
& - \frac{3 (2 - 3 \pi_{1}) (1 - \pi_{1})^{2}}{(n - 1) (n - 2) (n - 3) \pi_{1}} + \frac{(1 - \pi_1)^2}{(n - 1)^2 (n - 2)}.
\end{align*}
Here different letters $i, j, k, l$ denote distinct indices.
\end{lemma}

\begin{proof}
The proof is a consequence of Lemma \ref{lem:CRE}. First,
\begin{align*}
& \ \var \left( \left( \frac{t_{i}}{\pi_{1}} - 1 \right) t_{j} \right) = \bbE \left[ \left( \frac{t_{i}}{\pi_{1}} - 1 \right)^{2} t_{j} \right] - \bbE^{2} \left[ \left( \frac{t_{i}}{\pi_{1}} - 1 \right) t_{j} \right] \\
= & \ \bbE \left[ \left( \frac{t_{i}}{\pi_{1}^{2}} - \frac{2 t_{i}}{\pi_{1}} + 1 \right) t_{j} \right] - \left( \frac{n \pi_{1} - 1}{n - 1} - \pi_{1} \right)^{2} \\
= & \ \left\{ \left( \frac{1}{\pi_{1}} - 2 \right) \frac{n \pi_{1} - 1}{n - 1} + \pi_{1} \right\} - \left( \frac{n \pi_{1} - 1}{n - 1} - \pi_{1} \right)^{2} \\
= & \ \pi_{1} + \left( \frac{1 - \pi_{1}}{\pi_{1}} - 1 \right) \frac{n \pi_{1} - 1}{n - 1} - \frac{(1 - \pi_{1})^{2}}{(n - 1)^{2}}.
\end{align*}

Second,
\begin{align*}
& \ \cov \left( \left( \frac{t_{i}}{\pi_{1}} - 1 \right) t_{j}, t_{k} \right) = \bbE \left[ \left( \frac{t_{i}}{\pi_{1}} - 1 \right) t_{j} t_{k} \right] - \bbE \left[ \left( \frac{t_{i}}{\pi_{1}} - 1 \right) t_{j} \right] \pi_{1} \\
= & \ \bbE \left[ \frac{t_{i} t_{j} t_{k}}{\pi_{1}} - t_{j} t_{k} \right] - \left( \frac{n \pi_{1} - 1}{n - 1} - \pi_{1} \right) \pi_{1} \\
= & \ \frac{(n \pi_{1} - 1) (n \pi_{1} - 2)}{(n - 1) (n - 2)} - \pi_{1} \frac{n \pi_{1} - 1}{n - 1} - \left( \frac{n \pi_{1} - 1}{n - 1} - \pi_{1} \right) \pi_{1} \\
= & \ \frac{(n \pi_{1} - 1) (n \pi_{1} - 2)}{(n - 1) (n - 2)} - 2 \pi_{1} \frac{n \pi_{1} - 1}{n - 1} + \pi_{1}^{2}.
\end{align*}

Third,
\begin{align*}
& \ \cov \left( \left( \frac{t_{i}}{\pi_{1}} - 1 \right) t_{j}, t_{i} \right) = \bbE \left[ \left( \frac{t_{i}}{\pi_{1}} - 1 \right) t_{j} t_{i} \right] - \bbE \left[ \left( \frac{t_{i}}{\pi_{1}} - 1 \right) t_{j} \right] \pi_{1} \\
= & \ \bbE \left[ \frac{t_{i} t_{j}}{\pi_{1}} - t_{i} t_{j} \right] - \left( \frac{n \pi_{1} - 1}{n - 1} - \pi_{1} \right) \pi_{1} \\
= & \ (1 - \pi_{1}) \frac{n \pi_{1} - 1}{n - 1} - \pi_{1} \frac{n \pi_{1} - 1}{n - 1} + \pi_{1}^{2} \\
= & \ (1 - 2 \pi_{1}) \frac{n \pi_{1} - 1}{n - 1} + \pi_{1}^{2}.
\end{align*}

Fourth,
\begin{align*}
& \ \cov \left( \left( \frac{t_{i}}{\pi_{1}} - 1 \right) t_{j}, t_{j} \right) = \bbE \left[ \left( \frac{t_{i}}{\pi_{1}} - 1 \right) t_{j} t_{j} \right] - \bbE \left[ \left( \frac{t_{i}}{\pi_{1}} - 1 \right) t_{j} \right] \pi_{1} \\
= & \ \bbE \left[ \frac{t_{i} t_{j}}{\pi_{1}} - t_{j} \right] - \left( \frac{n \pi_{1} - 1}{n - 1} - \pi_{1} \right) \pi_{1} \\
= & \ \frac{n \pi_{1} - 1}{n - 1} - \pi_{1} - \frac{n \pi_{1} - 1}{n - 1} \pi_{1} + \pi_{1}^{2} \\
= & \ (1 - \pi_{1}) \frac{n \pi_{1} - 1}{n - 1} - \pi_{1} (1 - \pi_{1}).
\end{align*}

We then proceed to the fifth assertion.
\begin{align*}
& \ \cov \left( \left( \frac{t_{i}}{\pi_{1}} - 1 \right) t_{j}, t_{i} \left( \frac{t_{j}}{\pi_{1}} - 1 \right) \right) = \bbE \left[ t_{i} t_{j} \left( \frac{t_{i} t_{j}}{\pi_{1}^{2}} - 2 \frac{t_{i}}{\pi_{1}} + 1 \right) \right] - \bbE^{2} \left[ \left( \frac{t_{i}}{\pi_{1}} - 1 \right) t_{j} \right] \\
= & \ \bbE \left[ \frac{t_{i} t_{j}}{\pi_{1}^{2}} - 2 \frac{t_{i} t_{j}}{\pi_{1}} + t_{i} t_{j} \right] - \left( \frac{n \pi_{1} - 1}{n - 1} - \pi_{1} \right)^{2} \\
= & \ \left( \pi_{1} + \frac{1}{\pi_{1}} - 2 \right) \frac{n \pi_{1} - 1}{n - 1} - \frac{(1 - \pi_{1})^{2}}{(n - 1)^{2}} \\
= & \ \frac{1}{\pi_{1}} \frac{(1 - \pi_{1})^{2} (n \pi_{1} - 1) (n - 1) - \pi_{1} (1 - \pi_{1})^{2}}{(n - 1)^{2}} \\
= & \ \frac{(1 - \pi_{1})^{2}}{\pi_{1}} \frac{(n \pi_{1} - 1) (n - 1) - \pi_{1}}{(n - 1)^{2}} \\
= & \ \frac{(1 - \pi_{1})^{2}}{\pi_{1}} \frac{n \pi_{1} - 1}{n - 1} - \frac{(1 - \pi_{1})^{2}}{(n - 1)^{2}}.
\end{align*}

We next prove the sixth assertion.
\begin{align*}
& \ \cov \left( \left( \frac{t_{i}}{\pi_{1}} - 1 \right) t_{j}, \left( \frac{t_{i}}{\pi_{1}} - 1 \right) t_{k} \right) = \bbE \left[ \left( \frac{t_{i}}{\pi_{1}} - 1 \right)^{2} t_{j} t_{k} \right] - \bbE^{2} \left[ \left( \frac{t_{i}}{\pi_{1}} - 1 \right) t_{j} \right] \\
= & \ \bbE \left[ \frac{t_{i} t_{j} t_{k}}{\pi_{1}^{2}} - 2 \frac{t_{i} t_{j} t_{k}}{\pi_{1}} + t_{j} t_{k} \right] - \frac{(1 - \pi_{1})^{2}}{(n - 1)^{2}} \\
= & \ \left( \frac{1}{\pi_{1}} - 2 \right) \frac{(n \pi_{1} - 1) (n \pi_{1} - 2)}{(n - 1) (n - 2)} + \pi_{1} \frac{n \pi_{1} - 1}{n - 1} - \frac{(1 - \pi_{1})^{2}}{(n - 1)^{2}}.
\end{align*}

The seventh assertion follows similarly.
\begin{align*}
& \ \cov \left( \left( \frac{t_{i}}{\pi_{1}} - 1 \right) t_{j}, t_{i} \left( \frac{t_{k}}{\pi_{1}} - 1 \right) \right) = \bbE \left[ \left( \frac{t_{i}}{\pi_{1}} - 1 \right) \left( \frac{t_{k}}{\pi_{1}} - 1 \right) t_{i} t_{j} \right] - \bbE^{2} \left[ \left( \frac{t_{i}}{\pi_{1}} - 1 \right) t_{j} \right] \\
= & \ \bbE \left[ \frac{t_{i} t_{j} t_{k}}{\pi_{1}^{2}} - \frac{t_{i} t_{j} t_{k}}{\pi_{1}} - \frac{t_{i} t_{j}}{\pi_{1}} + t_{i} t_{j} \right] - \frac{(1 - \pi_{1})^{2}}{(n - 1)^{2}} \\
= & \ \frac{1 - \pi_{1}}{\pi_{1}} \frac{(n \pi_{1} - 1) (n \pi_{1} - 2)}{(n - 1) (n - 2)} - (1 - \pi_{1}) \frac{n \pi_{1} - 1}{n - 1} - \frac{(1 - \pi_{1})^{2}}{(n - 1)^{2}} \\
= & \ (1 - \pi_{1}) \frac{n \pi_{1} - 1}{n - 1} \left( \frac{1}{\pi_{1}} \frac{n \pi_{1} - 2}{n - 2} - 1 \right) - \frac{(1 - \pi_{1})^{2}}{(n - 1)^{2}} \\
= & - 2 \frac{(1 - \pi_{1})^{2}}{\pi_{1}} \frac{n \pi_{1} - 1}{n - 1} \frac{1}{n - 2} - \frac{(1 - \pi_{1})^{2}}{(n - 1)^{2}}.
\end{align*}

We then proceed to the eighth statement:
\begin{align*}
& \ \cov \left( \left( \frac{t_{i}}{\pi_{1}} - 1 \right) t_{j}, t_{j} \left( \frac{t_{k}}{\pi_{1}} - 1 \right) \right) = \bbE \left[ \left( \frac{t_{i}}{\pi_{1}} - 1 \right) \left( \frac{t_{k}}{\pi_{1}} - 1 \right) t_{j} \right] - \bbE^{2} \left[ \left( \frac{t_{i}}{\pi_{1}} - 1 \right) t_{j} \right] \\
= & \ \bbE \left[ \frac{t_{i} t_{j} t_{k}}{\pi_{1}^{2}} - 2 \frac{t_{i} t_{j}}{\pi_{1}} + t_{j} \right] - \frac{(1 - \pi_{1})^{2}}{(n - 1)^{2}} \\
= & \ \frac{1}{\pi_{1}} \frac{(n \pi_{1} - 1) (n \pi_{1} - 2)}{(n - 1) (n - 2)} - 2 \frac{n \pi_{1} - 1}{n - 1} + \pi_{1} - \frac{(1 - \pi_{1})^{2}}{(n - 1)^{2}}.
\end{align*}

Finally, the last assertion can be shown as follows.
\begin{align*}
& \ \cov \left( \left( \frac{t_{i}}{\pi_{1}} - 1 \right) t_{j}, \left( \frac{t_{k}}{\pi_{1}} - 1 \right) t_{l} \right) = \bbE \left[ \left( \frac{t_{i}}{\pi_{1}} - 1 \right) \left( \frac{t_{k}}{\pi_{1}} - 1 \right) t_{j} t_{l} \right] - \bbE^{2} \left[ \left( \frac{t_{i}}{\pi_{1}} - 1 \right) t_{j} \right] \\
= & \ \bbE \left[ \frac{t_{i} t_{j} t_{k} t_{l}}{\pi_{1}^{2}} - 2 \frac{t_{i} t_{j} t_{k}}{\pi_{1}} + t_{j} t_{l} \right] - \frac{(1 - \pi_{1})^{2}}{(n - 1)^{2}} \\
= & \ \frac{1}{\pi_{1}} \frac{(n \pi_{1} - 1) (n \pi_{1} - 2) (n \pi_{1} - 3)}{(n - 1) (n - 2) (n - 3)} - 2 \frac{(n \pi_{1} - 1) (n \pi_{1} - 2)}{(n - 1) (n - 2)} + \pi_{1} \frac{n \pi_{1} - 1}{n - 1} - \frac{(1 - \pi_{1})^{2}}{(n - 1)^{2}}.
\end{align*}
\end{proof}

\subsection{On the Asymptotic Distribution of \texorpdfstring{$\hat{\tau}_{\adj, 2}$}{estimator} under the Randomization-based Framework}
\label{app:clt}

In this section, we discuss the asymptotic normality of $\hat{\tau}_{\adj, 2}$. The asymptotic normality of $\hat{\tau}_{\unadj}$ is straightforward by applying design-based CTL of \citet{hajek1960limiting}. Hence our discussion focuses on the second-order $U$-statistic component $\widehat{\IIFF}_{\unadj, 2, 2}$.

To simplify our argument, we consider the Bernoulli sampling. We first perform Hoeffding decomposition of $\widehat{\IIFF}_{\unadj, 2, 2}$: Let $G_{i, j} \coloneqq H_{i, j} y_{j}$ for $i, j = 1, \cdots, n$, and then
\begin{align*}
\widehat{\IIFF}_{\unadj, 2, 2} & \equiv \frac{1}{n} \sum_{i = 1}^{n} \left( \frac{t_{i}}{\pi_{1}} - 1 \right) \sum_{j \neq i} G_{i, j} + \frac{1}{n} \sum_{1 \leq i \neq j \leq n} \left( \frac{t_{i}}{\pi_{1}} - 1 \right) G_{i, j} \left( \frac{t_{j}}{\pi_{1}} - 1 \right) \\
& \eqqcolon M_{1} + M_{2}.
\end{align*}
$M_{1}$ is uncorrelated with $M_{2}$. Denote $\bar{t} \coloneqq t / \pi_{1} - 1$ and $W_{i} \coloneqq \sum_{j \neq i} G_{i, j}$ so obviously $\bbE \bar{t} = 0$. Then $M_{1}$ and $M_{2}$ can be, respectively, represented as
\begin{align*}
M_{1} = \frac{1}{n} \sum_{i = 1}^{n} W_{i} \bar{t}_{i}, M_{2} = \frac{1}{n} \sum_{1 \leq i \neq j \leq n} G_{i, j} \bar{t}_{i} \bar{t}_{j}.
\end{align*}

$M_{1}$ is a sum of non-identically distributed independent random variables so we need to consider the Lindeberg version of CLT. Here $\bbE^{\sff} (n^{- 1 / 2} W_{i} \bar{t}_{i}) = 0$, 
\begin{align*}
\sigma_{1, i}^{2} \coloneqq & \ \bbE (n^{- 1 / 2} W_{i} \bar{t}_{i})^{2} = n^{-1} W_{i}^{2} \pi_{0} / \pi_{1} = \frac{\pi_{0}}{\pi_{1}} \frac{1}{n} \left( \sum_{j = 1}^{n} G_{i, j} - G_{i, i} \right)^{2} \\
= & \ \frac{\pi_{0}}{\pi_{1}} \left\{ \frac{1}{n} \left( \sum_{j = 1}^{n} G_{i, j} \right)^{2} - 2 \frac{1}{n} G_{i, i} \sum_{j = 1}^{n} G_{i, j} + \frac{1}{n} G_{i, i}^{2} \right\}
\end{align*} and
\begin{align*}
\kappa_{1, i}^{3} \coloneqq & \ \bbE^{\sff} |n^{- 1 / 2} W_{i} \bar{t}_{i}|^{3} = n^{- 3 / 2} |W_{i}|^{3} \bbE |\bar{t}|^{3} = n^{- 3 / 2} |W_{i}|^{3} \bbE \left[ \left\vert \frac{t}{\pi_{1}^{3}} - \frac{3 t}{\pi_{1}^{2}} + \frac{3 t}{\pi_{1}} - 1 \right\vert \right] \\
= & \ \frac{1}{n^{3 / 2}} |W_{i}|^{3} \left( \frac{\pi_{0}^{3}}{\pi_{1}^{2}} + \pi_{0} \right) = \pi_{0} \left( 1 + \left( \frac{\pi_{0}}{\pi_{1}} \right)^{2} \right) \frac{|W_{i}|^{3}}{n^{3 / 2}}.
\end{align*}
Now let $\bar{\kappa}_{1, n}^{3} \coloneqq \sum_{i = 1}^{n} \kappa_{1, i}^{3}$ and
\begin{align*}
& \ \bar{\sigma}_{1, n}^{2} \coloneqq \sum_{i = 1}^{n} \sigma_{1, i}^{2} = \frac{\pi_{0}}{\pi_{1}} \sum_{i = 1}^{n} \left\{ \frac{1}{n} \left( \sum_{j = 1}^{n} G_{i, j} \right)^{2} - 2 \frac{1}{n} G_{i, i} \sum_{j = 1}^{n} G_{i, j} + \frac{1}{n} G_{i, i}^{2} \right\} \\
= & \ \frac{\pi_{0}}{\pi_{1}} \left\{ \frac{1}{n} \sum_{i = 1}^{n} \left( \sum_{j = 1}^{n} G_{i, j} \right)^{2} - \frac{1}{n} \sum_{i = 1}^{n} G_{i, i} \sum_{j = 1}^{n} G_{i, j} + \frac{1}{n} \sum_{i = 1}^{n} G_{i, i}^{2} \right\}.
\end{align*}
Under Assumption \ref{as:regularity conditions}, one can easily show that $\bar{\kappa}_{1, n}^{3} = O (n^{- 1 / 2})$ and $\bar{\sigma}_{1, n}^{2} = O (1)$. Of course, the weakest regularity condition for Lindeberg CLT to hold is on the tail probability rather than on the third moment. But again, we sacrifice the mathematical generality for ease of presentation. By the Lindeberg version of the CLT, we can establish
\begin{align*}
\sqrt{n} M_{1} \rightsquigarrow N (0, \sigma_{1}^{2}),
\end{align*}
as long as there exists a unique constant $\sigma_{1}^{2} > 0$ such that $\lim_{n \rightarrow \infty} \bar{\sigma}_{1, n}^{2} = \sigma_{1}^{2}$.

As for $M_{2}$, we first symmetrize $M_{2}$:
\begin{align*}
M_{2} = \frac{1}{n} \sum_{1 \leq i < j \leq n} \bar{G}_{i, j} \bar{t}_{i} \bar{t}_{j},
\end{align*}
where $\bar{G}_{i, j} \coloneqq (G_{i, j} + G_{j, i}) / 2$. Then we further represent $M_{2}$ as a martingale sum:
\begin{align*}
M_{2} \equiv \frac{1}{n} \sum_{i = 2}^{n} \sum_{j = 1}^{i - 1} \bar{G}_{i, j} \bar{t}_{i} \bar{t}_{j},
\end{align*}
where $M_{2, i} \coloneqq \sum_{j = 1}^{i - 1} \bar{G}_{i, j} \bar{t}_{i} \bar{t}_{j}$ satisfies $\bbE [M_{2, i} | \calF_{i - 1}] = 0$ where $\calF_{i - 1}$ denotes the filtration formed by $\{\bo_{1}, \cdots, \bo_{i - 1}\}$, so $\{M_{2, i}, i = 2, \cdots, n\}$ indeed forms a martingale sequence. Therefore, we will invoke L\'{e}vy's martingale CLT to prove the asymptotic normality of $\sqrt{n} M_{2}$. Let
\begin{align*}
& \sigma_{2, n}^{2} \coloneqq \sum_{i = 2}^{n} \sigma_{2, i}^{2}, \text{ where } \sigma_{2, i}^{2} \coloneqq \bbE^{\sff} [n^{-1} M_{2, i}^{2} | \calF_{i - 1}].
\end{align*}
In particular,
\begin{align*}
\sigma_{2, i}^{2} = & \frac{1}{n} \bbE^{\sff} \left[ \left( \sum_{j = 1}^{i - 1} \bar{G}_{i, j} \bar{t}_{j} \right)^{2} \bar{t}_{i}^{2} | \calF_{i - 1} \right] = \frac{\pi_{0}}{\pi_{1}} \frac{1}{n} \left( \sum_{j = 1}^{i - 1} \bar{G}_{i, j} \bar{t}_{j} \right)^{2} \\
= &\ \frac{\pi_{0}}{\pi_{1}} \frac{1}{n} \left\{ \sum_{j = 1}^{i - 1} \bar{G}_{i, j}^{2} \bar{t}_{j}^{2} + \sum_{1 \leq j_{1} \neq j_{2} \leq i - 1} \bar{G}_{i, j_{1}} \bar{G}_{i, j_{2}} \bar{t}_{j_{1}} \bar{t}_{j_{2}} \right\}.
\end{align*}
Then
\begin{align*}
\sigma_{2, n}^{2} = & \ \frac{1}{n} \sum_{i = 2}^{n} \frac{\pi_{0}}{\pi_{1}} \sum_{j = 1}^{i - 1} \bar{G}_{i, j}^{2} \bar{t}_{j}^{2} + \frac{1}{n} \sum_{i = 2}^{n} \frac{\pi_{0}}{\pi_{1}} \sum_{1 \leq j_{1} \neq j_{2} \leq i - 1} \bar{G}_{i, j_{1}} \bar{G}_{i, j_{2}} \bar{t}_{j_{1}} \bar{t}_{j_{2}} \eqqcolon C_{1} + C_{2}.
\end{align*}
First, it is not difficult to see that
\begin{align*}
\bbE^{\sff} C_{1} = \frac{1}{n} \sum_{i = 2}^{n} \left( \frac{\pi_{0}}{\pi_{1}} \right)^{2} \sum_{j = 1}^{i - 1} \bar{G}_{i, j}^{2} = \left( \frac{\pi_{0}}{\pi_{1}} \right)^{2} \frac{1}{n} \sum_{1 \leq i < j \leq n} \bar{G}_{i, j}^{2}.
\end{align*}
and $\bbE^{\sff} C_{2} = 0$. Thus as long as there exists a unique $\sigma_{2}^{2} > 0$ such that $\lim_{n \rightarrow \infty} \bbE^{\sff} C_{1} = \sigma_{2}^{2}$, we can show that
\begin{align*}
\sqrt{n} M_{2} \rightsquigarrow N (0, \sigma_{2}^{2}).
\end{align*}

Combining the above arguments, we need the following assumption to hold for $\sqrt{n} \widehat{\IIFF}_{\unadj, 2, 2}$ to have a Gaussian limiting distribution.

\begin{assumption}
\label{as:clt}
There exists a unique $\sigma_{1}^{2} > 0$ such that $\lim_{n \rightarrow \infty} \bar{\sigma}_{1, n}^{2} = \sigma_{1}^{2}$ and a unique $\sigma_{2}^{2} > 0$ such that $\lim_{n \rightarrow \infty} \bar{\sigma}_{2, n}^{2} = \sigma_{2}^{2}$.
\end{assumption}

Since $\sigma_{1, n}^{2} + \sigma_{2, n}^{2}$ is the design-based variance of $\widehat{\IIFF}_{\unadj, 2, 2}$, the extra condition for Gaussian limiting distribution is essentially that the design-based variance converges to a $O (1)$ constant.

\subsection{Proofs of Results in Section \ref{sec:variety}}
\label{app:variety}

We first prove that $\hat{\tau}_{\adj, 3}$ is unbiased under CRE but biased under the Bernoulli sampling.

\begin{proof}
Recall that $\hat{\tau}_{\adj, 2}$ has bias of the form
\begin{align*}
- \frac{\pi_{0}}{\pi_{1}} \frac{1}{n (n - 1)} \sum_{i = 1}^{n} H_{i, i} y_{i} (1).
\end{align*}
It remains to show that the extra term in $\hat{\tau}_{\adj, 3}$ is an unbiased estimator of the above bias, which is immediate because $\bbE t \equiv \pi_{1}$. The biasedness of $\hat{\tau}_{\adj, 3}$ under the Bernoulli sampling stems from the unbiasedness of $\hat{\tau}_{\adj, 2}$ under the Bernoulli sampling and $\hat{\tau}_{\adj, 3}$ adds to $\hat{\tau}_{\adj, 2}$ a random variable with non-zero mean.
\end{proof}

We next prove Proposition \ref{prop:bias free HOIF variance}.

\begin{proof}[Proof of Proposition \ref{prop:bias free HOIF variance}]
Denote the extra term in $\hat{\tau}_{\adj, 3}$ as $\hat{\alpha}$. We first compute its variance.
\begin{align*}
\var^{\sff} (\hat{\alpha}) = \left( \frac{\pi_{0}}{\pi_{1}} \right)^{3} \frac{1}{n (n - 1)^{3}} \left\{ \sum_{i = 1}^{n} H_{i, i}^{2} y_{i} (1)^{2} - \frac{1}{n} \trace^{2} \left( \hat{\bSigma}_{y} \hat{\bSigma}^{-} \right) \right\}.
\end{align*}

Next, we compute the covariance between $\hat{\tau}_{\unadj}$ and $\hat{\alpha}$.
\begin{align*}
& \ \cov^{\sff} \left( \hat{\tau}_{\unadj}, \hat{\alpha} \right) = \cov^{\sff} \left( \frac{1}{n} \sum_{i = 1}^{n} \frac{t_{i}}{\pi_{1}} y_{i} (1), \frac{\pi_{0}}{\pi_{1}} \frac{1}{n (n - 1)} \sum_{i = 1}^{n} H_{i, i} \frac{t_{i}}{\pi_{1}} y_{i} (1) \right) \\
= & \ \frac{1}{n^{2} (n - 1)} \frac{\pi_{0}}{\pi_{1}^{3}} \sum_{i = 1}^{n} y_{i} (1) \sum_{j = 1}^{n} H_{j, j} y_{j} (1) \cov \left( t_{i}, t_{j} \right) \\
= & - \frac{1}{n^{2} (n - 1)} \frac{\pi_{0}}{\pi_{1}^{3}} \sum_{1 \leq i \neq j \leq n} H_{i, i} y_{i} (1) y_{j} (1) \frac{\pi_{1} \pi_{0}}{n - 1} + \frac{1}{n^{2} (n - 1)} \frac{\pi_{0}}{\pi_{1}^{3}} \sum_{i = 1}^{n} H_{i, i} y_{i} (1)^{2} \pi_{1} \pi_{0} \\
= & \ \left( \frac{\pi_{0}}{\pi_{1}} \right)^{2} \frac{1}{n^{2} (n - 1)} \left\{ - \frac{1}{n - 1} \sum_{i = 1}^{n} H_{i, i} y_{i} (1) \sum_{j = 1}^{n} y_{j} (1) + \frac{1}{n - 1} \sum_{i = 1}^{n} H_{i, i} y_{i} (1)^{2} + \sum_{i = 1}^{n} H_{i, i} y_{i} (1)^{2} \right\} \\
= & \ \left( \frac{\pi_{0}}{\pi_{1}} \right)^{2} \frac{1}{n^{2} (n - 1)} \left\{ \frac{n}{n - 1} \sum_{i = 1}^{n} H_{i, i} y_{i} (1)^{2} - \bar{\tau} \cdot \frac{n}{n - 1} \sum_{i = 1}^{n} H_{i, i} y_{i} (1) \right\} \\
= & \ \left( \frac{\pi_{0}}{\pi_{1}} \right)^{2} \frac{1}{n (n - 1)^{2}} \sum_{i = 1}^{n} H_{i, i} y_{i} (1) \cdot (y_{i} (1) - \bar{\tau}).
\end{align*}

Finally, we are left to compute the covariance between $\hat{\alpha}$ and $\widehat{\IIFF}_{\unadj, 2, 2}$.
\begin{align*}
& \ \cov^{\sff} \left( \widehat{\IIFF}_{\unadj, 2, 2}, \hat{\alpha} \right) = \cov^{\sff} \left( \frac{1}{n} \sum_{1 \leq i \neq j \leq n} \left( \frac{t_{i}}{\pi_{1}} - 1 \right) H_{i, j} \frac{t_{j} y_{j} (1)}{\pi_{1}}, \frac{\pi_{0}}{\pi_{1}} \frac{1}{n (n - 1)} \sum_{i = 1}^{n} H_{i, i} \frac{t_{i}}{\pi_{1}} y_{i} (1) \right) \\
= & \ \frac{\pi_{0}}{\pi_{1}^{3}} \frac{1}{n^{2} (n - 1)} \sum_{1 \leq i \neq j \leq n} H_{i, j} y_{j} (1) \sum_{k = 1}^{n} H_{k, k} y_{k} (1) \cov \left( \left( \frac{t_{i}}{\pi_{1}} - 1 \right) t_{j} , t_{k} \right) \\
= & \ \frac{\pi_{0}}{\pi_{1}^{3}} \frac{1}{n^{2} (n - 1)} \sum_{1 \leq i \neq j \neq k \leq n} H_{i, j} y_{j} (1) H_{k, k} y_{k} (1) \frac{\pi_{0}^{2} - n \pi_{1} \pi_{0}}{(n - 1) (n - 2)} \\
& + \frac{\pi_{0}}{\pi_{1}^{3}} \frac{1}{n^{2} (n - 1)} \sum_{1 \leq i \neq j \leq n} H_{i, j} y_{j} (1) H_{i, i} y_{i} (1) \left( \pi_{1} \pi_{0} - \frac{\pi_{0} (\pi_{0} - \pi_{1})}{n - 1} \right) \\
& - \frac{\pi_{0}}{\pi_{1}^{3}} \frac{1}{n^{2} (n - 1)} \sum_{1 \leq i \neq j \leq n} H_{i, j} y_{j} (1) H_{j, j} y_{j} (1) \frac{\pi_{0}^{2}}{n - 1} \\
= & \ \left\{ \left( \frac{\pi_{0}}{\pi_{1}} \right)^{3} \frac{1}{n^{2} (n - 1)^{2} (n - 2)} - \left( \frac{\pi_{0}}{\pi_{1}} \right)^{2} \frac{1}{n (n - 1)^{2} (n - 2)} \right\} \sum_{1 \leq i \neq j \neq k \leq n} H_{i, j} y_{j} (1) H_{k, k} y_{k} (1) \\
& + \left\{ \left( \frac{\pi_{0}}{\pi_{1}} \right)^{2} \frac{1}{n^{2} (n - 1)} - \left( \frac{\pi_{0}}{\pi_{1}} \right)^{3} \frac{1}{n^{2} (n - 1)^{2}} + \left( \frac{\pi_{0}}{\pi_{1}} \right)^{2} \frac{1}{n^{2} (n - 1)^{2}} \right\} \sum_{1 \leq i \neq j \leq n} y_{i} (1) H_{i, i} H_{i, j} y_{j} (1) \\
& - \left( \frac{\pi_{0}}{\pi_{1}} \right)^{3} \frac{1}{n^{2} (n - 1)^{2}} \sum_{1 \leq i \neq j \leq n} H_{i, j} H_{j, j} y_{j} (1)^{2} \\
= & \ \left( \frac{\pi_{0}}{\pi_{1}} \right)^{2} \frac{1}{n (n - 1)^{2} (n - 2)} \left\{ 1 - \frac{\pi_{0}}{\pi_{1}} \frac{1}{n} \right\} \sum_{1 \leq i \neq j \leq n} \left\{ H_{i, i} y_{i} (1) H_{j, j} y_{j} (1) + y_{i} (1) H_{i, i} H_{i, j} y_{j} (1) \right\} \\
& + \left( \frac{\pi_{0}}{\pi_{1}} \right)^{2} \frac{1}{n^{2} (n - 1)} \left\{ 1 -  \frac{\pi_{0} - \pi_{1}}{\pi_{1}} \frac{1}{n - 1} \right\} \sum_{1 \leq i \neq j \leq n} y_{i} (1) H_{i, i} H_{i, j} y_{j} (1) \\
& + \left( \frac{\pi_{0}}{\pi_{1}} \right)^{3} \frac{1}{n^{2} (n - 1)^{2}} \sum_{i = 1}^{n} H_{i, i}^{2} y_{i} (1)^{2} \\
= 
 & \ \left( \frac{\pi_{0}}{\pi_{1}} \right)^{2} \frac{1}{n^{2} (n - 1)^{2}} \frac{\pi_1 n^2 - n + \pi_0}{(n-2)\pi_1} \sum_{1 \leq i \neq j \leq n} y_{i} (1) H_{i, i} H_{i, j} y_{j} (1) \\
& + \left( \frac{\pi_{0}}{\pi_{1}} \right)^{2} \frac{1}{n (n - 1)^{2} (n - 2)} \left\{ 1 - \frac{\pi_{0}}{\pi_{1}} \frac{1}{n} \right\} \trace^{2} \left( \hat{\bSigma}_{y} \hat{\bSigma}^{-} \right) \\
& + \left( \frac{\pi_{0}}{\pi_{1}} \right)^{2} \frac{1}{n (n - 1)^{2}} \left\{ \frac{\pi_{0}}{\pi_{1}} \frac{1}{n} - \frac{1}{n - 2} \left( 1 - \frac{\pi_{0}}{\pi_{1}} \frac{1}{n} \right) \right\} \sum_{i = 1}^{n} H_{i, i}^{2} y_{i} (1)^{2}.
\end{align*}
\end{proof}

\subsection{A Brief Discussion on ``Gadgets''}
\label{app:gadgets}

In Remark \ref{rem:gadget}, we mention that we will provide some statistical intuition on how to read off the order of $\var^{\sff} (\hat{\tau}_{\adj, 2})$ by simply inspecting each term involved in \eqref{var, HOIF-CRE, design-based} using a few ``gadgets''. The discussion here is not formal or rigorous, so we use the non-rigorous asymptotic notation $\tilde{O} (\cdot)$. The following gadgets are often involved in the bias and variance formula:
\begin{enumerate}[label = (\roman*)]
\item $\frac{1}{n} \sum_{1 \leq i \neq j \leq n} y_{i} (1) H_{i, j} (H_{j, j})^{k} y_{j} (1) = \tilde{O} \left( \left( \frac{p}{n} \right)^{k} \right)$ for any bounded non-negative integer $k$: this can be deduced by viewing this $U$-statistic as the product of the projection of $y (1)$ onto the space spanned by $\bx - \bar{\bx}$ and the projection of $y (1)$ onto the space spanned by $(\bx - \bar{\bx}) (\bx^{\top} \hat{\bSigma}^{-} \bx)^{k}$ and the $(p / n)^{k}$ factor comes from $(\bx^{\top} \hat{\bSigma}^{-} \bx)^{k}$ (note that the projection of $y (1)$ onto $\bx - \bar{\bx}$ should be of order $\tilde{O} (1)$);

\item $\sum_{i = 1}^{n} H_{i, i}^{k_{1}} y_{i} (1)^{k_{2}} = \tilde{O} (p)$ for any bounded non-negative integers $k_{1}, k_{2}$: this can be directly deduced using the cyclic permutation invariance of the trace operator.
\end{enumerate}

By using the above two gadgets, it is then straightforward to conclude the following (just looking at each term in \eqref{var, HOIF-CRE, design-based} one by one):
\begin{itemize}
\item $\frac{\pi_{0}}{\pi_{1}} \frac{1}{n} V_{n} (y (1)) = \tilde{O} (1 / n)$ without even using the above gadgets;

\item $\frac{\pi_{0}}{\pi_{1}} \frac{1}{n^{2}} \sum_{1 \leq i \neq j \leq n} H_{i, j} (1 + 2 H_{j, j}) y_{i} (1) y_{j} (1) = \tilde{O} (1 / n + p / n^{2})$ using gadget (i);

\item $\frac{\pi_{0}}{\pi_{1}} \frac{1}{n^{2}} \sum_{i = 1}^{n} H_{i, i} (\frac{1}{\pi_{1}} - \frac{\pi_{0} + 1}{\pi_{1}} H_{i, i}) y_{i} (1)^{2} = \tilde{O} (p / n^{2})$ using gadget (ii);

\item $(\frac{\pi_{0}}{\pi_{1}})^{2} \frac{1}{n^{2}} \trace ((\hat{\bSigma}_{y} \hat{\bSigma}^{-})^{2}) = \tilde{O} (p / n^{2})$ using gadget (ii);

\item $\frac{\pi_{0}}{\pi_{1}} \frac{1}{n^{3}} \trace^{2} (\hat{\bSigma}_{y} \hat{\bSigma}^{-}) = \tilde{O} (p^{2} / n^{3}) = O (p / n^{2})$ using gadget (ii);

\item $\bar{\tau} \cdot \frac{\pi_{0}}{\pi_{1}} \frac{1}{n^{2}} \trace (\hat{\bSigma}_{y} \hat{\bSigma}^{-}) = \tilde{O} (p / n^{2})$ using gadget (ii).
\end{itemize}

\subsection{Proofs of Theorems in Section \ref{sec:understanding}}
\label{app:understanding}

We now prove Lemma \ref{lem:alternative}.

\begin{proof}[Proof of Lemma \ref{lem:alternative}]
The alternative decomposition of $\hat{\tau}_{\db}$ is derived as follows:
\begin{align*}
\hat{\tau}_{\db} = & \ \hat{\tau}_{\unadj} - \frac{1}{n} \sum_{i = 1}^{n} \sum_{j = 1}^{n} \frac{t_{i} t_{j}}{\pi_{1}^{2}} H_{i, j} y_{j} (1) + \frac{1}{n^{2}} \sum_{i = 1}^{n} \sum_{j = 1}^{n} \sum_{k = 1}^{n} \frac{t_{i} t_{j} t_{k}}{\pi_{1}^{3}} H_{i, j} y_{k} (1) \\
& + \frac{1 - \pi_{1}}{\pi_{1}} \frac{1}{n} \sum_{i = 1}^{n} \frac{t_{i}}{\pi_{1}} H_{i, i} y_{i} (1) - \frac{1 - \pi_{1}}{\pi_{1}} \frac{1}{n^{2}} \sum_{i = 1}^{n} \sum_{j = 1}^{n} \frac{t_{i} t_{j}}{\pi_{1}^{2}} H_{i, i} y_{j} (1) \\
= & \ \hat{\tau}_{\unadj} - \frac{1}{n} \sum_{1 \leq i \neq j \leq n} \frac{t_{i} t_{j}}{\pi_{1}^{2}} H_{i, j} y_{j} (1) - \frac{1}{n} \sum_{i = 1}^{n} \frac{t_{i}}{\pi_{1}} H_{i, i} y_{i} (1) \\
& + \frac{1}{n^{2}} \sum_{1 \leq i \neq j \leq n} \sum_{k = 1}^{n} \frac{t_{i} t_{j} t_{k}}{\pi_{1}^{3}} H_{i, j} y_{k} (1) + \frac{1}{n^{2}} \sum_{i = 1}^{n} \sum_{k = 1}^{n} \frac{t_{i} t_{k}}{\pi_{1}^{2}} H_{i, i} y_{k} (1) \\
= & \ \hat{\tau}_{\unadj} - \frac{1}{n} \sum_{1 \leq i \neq j \leq n} \frac{t_{i} t_{j}}{\pi_{1}^{2}} H_{i, j} (y_{j} (1) - \hat{\tau}_{\unadj}) - \frac{1}{n} \sum_{i = 1}^{n} \frac{t_{i}}{\pi_{1}} H_{i, i} (y_{i} (1) - \hat{\tau}_{\unadj}) \\
= & \ \hat{\tau}_{\unadj} - \frac{1}{n} \sum_{1 \leq i \neq j \leq n} \left( \frac{t_{i}}{\pi_{1}} - 1 \right) H_{i, j} \frac{t_{j}}{\pi_{1}} (y_{j} (1) - \hat{\tau}_{\unadj}) \\
& - \frac{1}{n} \sum_{1 \leq i \neq j \leq n} H_{i, j} \frac{t_{j}}{\pi_{1}} (y_{j} (1) - \hat{\tau}_{\unadj}) - \frac{1}{n} \sum_{i = 1}^{n} \frac{t_{i}}{\pi_{1}} H_{i, i} (y_{i} (1) - \hat{\tau}_{\unadj}) \\
= & \ \hat{\tau}_{\unadj} - \frac{1}{n} \sum_{1 \leq i \neq j \leq n} \left( \frac{t_{i}}{\pi_{1}} - 1 \right) H_{i, j} \frac{t_{j}}{\pi_{1}} (y_{j} (1) - \hat{\tau}_{\unadj}) - \frac{1}{n} \sum_{i = 1}^{n} \sum_{j = 1}^{n} H_{i, j} \frac{t_{j}}{\pi_{1}} (y_{j} (1) - \hat{\tau}_{\unadj}) \\
= & \ \hat{\tau}_{\unadj} - \frac{1}{n} \sum_{1 \leq i \neq j \leq n} \left( \frac{t_{i}}{\pi_{1}} - 1 \right) H_{i, j} \frac{t_{j} (y_{j} (1) - \hat{\tau}_{\unadj})}{\pi_{1}} \equiv \hat{\tau}_{\adj, 2}^{\dag},
\end{align*}
where the last step follows from Lemma \ref{lem:hat}.
\end{proof}

We next prove Proposition \ref{prop:db}.

\begin{proof}[Proof of Proposition \ref{prop:db}]
In the first part, we compute the bias formula of $\hat{\tau}_{\adj, 2}^{\dag}$ under the design-based framework, based on Lemma \ref{lem:alternative}, and Lemma \ref{lem:hat} and Lemma \ref{lem:CRE} in Appendix \ref{app:technical lemmas}. By the alternative decomposition $\hat{\tau}_{\db} \equiv \hat{\tau}_{\adj, 2}^{\dag}$, we have
\begin{align*}
& \ \bbE^{\sff} (\hat{\tau}_{\adj, 2}^{\dag} - \bar{\tau}) \\
= & \ \bbE^{\sff} \left[ \frac{n - 1}{n} \widehat{\IIFF}_{\unadj, 2, 2} \right] + \bbE^{\sff} \left[ \frac{1}{n} \sum_{1 \leq i \neq j \leq n} \left( \frac{t_{i}}{\pi_{1}} - 1 \right) H_{i, j} \frac{t_{j}}{\pi_{1}} \hat{\tau}_{\unadj} \right] \\
= & - \frac{\pi_{0}}{\pi_{1}} \frac{1}{n (n - 1)} \sum_{i = 1}^{n} H_{i, i} y_{i} (1) + \bbE^{\sff} \left[ \frac{1}{n} \sum_{1 \leq i \neq j \leq n} \left( \frac{t_{i}}{\pi_{1}} - 1 \right) H_{i, j} \frac{t_{j}}{\pi_{1}} \hat{\tau}_{\unadj} \right].
\end{align*}
Thus we are left to compute the second term in the above display.
\begin{align*}
& \ \bbE^{\sff} \left[ \frac{1}{n} \sum_{1 \leq i \neq j \leq n} \left( \frac{t_{i}}{\pi_{1}} - 1 \right) H_{i, j} \frac{t_{j}}{\pi_{1}} \hat{\tau}_{\unadj} \right] \\
= & \ \frac{1}{n^{2}} \sum_{1 \leq i \neq j \leq n} \sum_{k = 1}^{n} \bbE^{\sff} \left[ \left( \frac{t_{i}}{\pi_{1}} - 1 \right) H_{i, j} \frac{t_{j} t_{k}}{\pi_{1}^{2}} y_{k} (1) \right] \\
= & \ \frac{1}{n^{2}} \sum_{1 \leq i \neq j \neq k \leq n} \bbE^{\sff} \left[ \left( \frac{t_{i}}{\pi_{1}} - 1 \right) H_{i, j} \frac{t_{j} t_{k}}{\pi_{1}^{2}} y_{k} (1) \right] + \frac{1}{n^{2}} \sum_{1 \leq i \neq j \leq n} \bbE^{\sff} \left[ \left( \frac{t_{i}}{\pi_{1}} - 1 \right) H_{i, j} \frac{t_{j} t_{i}}{\pi_{1}^{2}} y_{i} (1) \right] \\
& + \frac{1}{n^{2}} \sum_{1 \leq i \neq j \leq n} \bbE^{\sff} \left[ \left( \frac{t_{i}}{\pi_{1}} - 1 \right) H_{i, j} \frac{t_{j}}{\pi_{1}^{2}} y_{j} (1) \right] \\
= & \ \frac{1}{n^{2}} \sum_{1 \leq i \neq j \neq k \leq n} H_{i, j} y_{k} (1) \left( \frac{1}{\pi_{1}^{2}} \frac{(n \pi_{1} - 1) (n \pi_{1} - 2)}{(n - 1) (n - 2)} - \frac{1}{\pi_{1}} \frac{n \pi_{1} - 1}{n - 1} \right) \\
& + \frac{1}{n^{2}} \sum_{1 \leq i \neq j \leq n} H_{i, j} \left\{ y_{i} (1) \left( \frac{1}{\pi_{1}^{2}} \frac{n \pi_{1} - 1}{n - 1} - \frac{1}{\pi_{1}} \frac{n \pi_{1} - 1}{n - 1} \right) + y_{j} (1) \left( \frac{1}{\pi_{1}^{2}} \frac{n \pi_{1} - 1}{n - 1} - \frac{1}{\pi_{1}} \right) \right\} \\
= & - \frac{2 (1 - \pi_{1})}{\pi_{1}^{2}} \frac{1}{n^{2}} \sum_{1 \leq i \neq j \neq k \leq n} H_{i, j} y_{k} (1) \frac{n \pi_{1} - 1}{n - 1} \frac{1}{n - 2} \\
& + \frac{1 - \pi_{1}}{\pi_{1}^{2}} \frac{1}{n^{2}} \sum_{1 \leq i \neq j \leq n} H_{i, j} y_{i} (1) \frac{n \pi_{1} - 1}{n - 1} - \frac{1 - \pi_{1}}{\pi_{1}^{2}} \frac{1}{n^{2}} \sum_{1 \leq i \neq j \leq n} H_{i, j} y_{j} (1) \frac{1}{n - 1} \\
= & - \frac{2 \pi_{0}}{\pi_{1}^{2}} \frac{n \pi_{1} - 1}{n - 1} \frac{1}{n (n - 2)} \sum_{1 \leq i \neq j \leq n} H_{i, j} \left\{ \bar{\tau} - \frac{1}{n} y_{i} (1) - \frac{1}{n} y_{j} (1) \right\} \\
& - \frac{\pi_{0}}{\pi_{1}^{2}} \frac{n \pi_{1} - 1}{n - 1} \frac{1}{n^{2}} \sum_{i = 1}^{n} H_{i, i} y_{i} (1) + \frac{\pi_{0}}{\pi_{1}^{2}} \frac{1}{n^{2} (n - 1)} \sum_{i = 1}^{n} H_{i, i} y_{i} (1) \\
= & \ 2 \frac{\pi_{0}}{\pi_{1}^{2}} \frac{n \pi_{1} - 1}{n - 1} \frac{1}{n - 2} \frac{p}{n} \bar{\tau} - 4 \frac{\pi_{0}}{\pi_{1}^{2}} \frac{n \pi_{1} - 1}{n - 1} \frac{1}{n - 2} \frac{1}{n^{2}} \sum_{i = 1}^{n} H_{i, i} y_{i} (1) - \frac{\pi_{0}}{\pi_{1}^{2}} \frac{n \pi_{1} - 2}{n - 1} \frac{1}{n^{2}} \sum_{i = 1}^{n} H_{i, i} y_{i} (1).
\end{align*}
Then we have
\begin{align*}
& \ \bbE^{\sff} (\hat{\tau}_{\adj, 2}^{\dag} - \bar{\tau}) \\
= & \ 2 \frac{\pi_{0}}{\pi_{1}^{2}} \frac{n \pi_{1} - 1}{n - 1} \frac{1}{n - 2} \frac{p}{n} \bar{\tau} - \frac{\pi_{0}}{\pi_{1}^{2}} \left\{ \frac{n}{n - 1} \pi_{1} + \frac{n \pi_{1} - 2}{n - 1} + \frac{4 (n \pi_{1} - 1)}{(n - 1) (n -2)} \right\} \frac{1}{n^{2}} \sum_{i = 1}^{n} H_{i, i} y_{i} (1) \\
= & \ 2 \frac{\pi_{0}}{\pi_{1}^{2}} \frac{n \pi_{1} - 1}{n - 1} \frac{1}{n - 2} \frac{p}{n} \bar{\tau} - \frac{\pi_{0}}{\pi_{1}^{2}} \frac{\pi_{1} n (n - 2) + (n \pi_{1} - 2) (n - 2) + 4 n \pi_{1} - 4}{(n - 1) (n - 2)} \frac{1}{n^{2}} \sum_{i = 1}^{n} H_{i, i} y_{i} (1) \\
= & \ 2 \frac{\pi_{0}}{\pi_{1}^{2}} \frac{n \pi_{1} - 1}{n - 1} \frac{1}{n - 2} \frac{p}{n} \bar{\tau} - \frac{\pi_{0}}{\pi_{1}^{2}} \frac{n^{2} \pi_{1} - 2 n \pi_{1} + n^{2} \pi_{1} - 2 n + 2 n \pi_{1}}{(n - 1) (n - 2)} \frac{1}{n^{2}} \sum_{i = 1}^{n} H_{i, i} y_{i} (1) \\
= & \ 2 \frac{\pi_{0}}{\pi_{1}^{2}} \frac{n \pi_{1} - 1}{n - 1} \frac{1}{n - 2} \frac{p}{n} \bar{\tau} - 2 \frac{\pi_{0}}{\pi_{1}^{2}} \frac{n^{2} \pi_{1} - n}{(n - 1) (n - 2)} \frac{1}{n^{2}} \sum_{i = 1}^{n} H_{i, i} y_{i} (1) \\
= & \ 2 \frac{\pi_{0}}{\pi_{1}^{2}} \frac{n \pi_{1} - 1}{n - 1} \frac{1}{n - 2} \left( \frac{p}{n} \bar{\tau} - \frac{1}{n} \sum_{i = 1}^{n} H_{i, i} y_{i} (1) \right).
\end{align*}

We next directly provide the exact variance of $\hat{\tau}_{\adj, 2}^{\dag}$ under the design-based framework. The derivation is elementary but quite tedious. The key intermediate steps can be found in the lemma below. But we refer interested readers to \href{https://github.com/Cinbo-Wang/HOIF-Car/blob/main/var-db.pdf}{this GitHub link} for the complete derivation.
\begin{align*}
& \var^{\sff} (\hat{\tau}_{\adj, 2}^{\dag}) = \var^{\sff} (\hat{\tau}_{\unadj}) + \var^{\sff} (\widehat{\IIFF}_{\unadj, 2, 2}^{\dag}) - 2 \cov^{\sff} (\hat{\tau}_{\unadj}, \widehat{\IIFF}_{\unadj, 2, 2}^{\dag}) \\
= & \ \left( \frac{\pi_{0}}{\pi_{1}} \right) \frac{1}{n} V_{n} \left[ (y_{i} (1) - \bar{\tau}) - \sum_{j \neq i} H_{j, i} (y_{j} (1) - \bar{\tau}) \right] \\
& \ + \left( \frac{\pi_{0}}{\pi_{1}} \right)^{2} \frac{1}{n} \left\{ \frac{1}{n} \sum_{i = 1}^{n} H_{i, i} (1 - H_{i, i}) (y_{i} (1) - \bar{\tau})^{2}  + \frac{1}{n} \sum_{1 \leq i \neq j \leq n} H_{i, j}^{2} (y_{i} (1) - \bar{\tau}) (y_{j} (1) - \bar{\tau}) \right\} + o \left( n^{-1} \right) \\
= & \ \nu^{\sff\dag} +  o \left( n^{-1} \right) .
\end{align*}

\end{proof}

\begin{lemma}
\label{lem:HOIF-db-var}
The variance of $\hat{\tau}_{\adj, 2}^{\dag}$ can be approximated as follows under Assumption \ref{as:regularity conditions}:
\begin{equation}
\label{HOIF-db-var}
\begin{split}
    \var^{\sff} (\hat{\tau}_{\adj, 2}^{\dag}) = & \left( \frac{\pi_{0}}{\pi_{1}} \right) \frac{1}{n} V_{n} \left[ (y_{i} (1) - \bar{\tau}) - \sum_{j \neq i} H_{j, i} (y_{j} (1) - \bar{\tau}) \right] \\
    & \ + \left( \frac{\pi_{0}}{\pi_{1}} \right)^{2} \frac{1}{n} \left\{ \frac{1}{n} \sum_{i = 1}^{n} H_{i, i} (1 - H_{i, i}) (y_{i} (1) - \bar{\tau})^{2}  + \frac{1}{n} \sum_{1 \leq i \neq j \leq n} H_{i, j}^{2} (y_{i} (1) - \bar{\tau}) (y_{j} (1) - \bar{\tau}) \right\}\\ & \ + o \left( n^{-1} \right) \\
    = & \ \nu^{\sff\dag} +  o \left( n^{-1} \right).
\end{split}
\end{equation}
\end{lemma}

\begin{proof}
We use the following elementary decomposition.
\begin{align*}
& \var^{\sff} (\hat{\tau}_{\adj, 2}^{\dag}) = \var^{\sff} (\hat{\tau}_{\unadj}) + \var^{\sff} (\widehat{\IIFF}_{\unadj, 2, 2}^{\dag}) - 2 \cov^{\sff} (\hat{\tau}_{\unadj}, \widehat{\IIFF}_{\unadj, 2, 2}^{\dag}).
\end{align*}
where $ \var^{\sff} (\hat{\tau}_{\unadj})$ has been derived before. We only need to derive $\var^{\sff} (\widehat{\IIFF}_{\unadj, 2, 2}^{\dag})$ and $\cov^{\sff} (\hat{\tau}_{\unadj}, \widehat{\IIFF}_{\unadj, 2, 2}^{\dag})$.

Here we only give the exact form of $\var^{\sff} (\hat{\tau}_{\adj, 2}^{\dag})$, and more details of the derivation can be found in \href{https://github.com/Cinbo-Wang/HOIF-Car/var-db.pdf}{this GitHub link}, here $\bar{\tau}^{(2)} = \frac{1}{n}\sum_{i=1}^{n} y_{i}^{2}$,
\begin{align*}
    & \var^{\sff} (\hat{\tau}_{\adj, 2}^{\dag}) = \var^{\sff} (\hat{\tau}_{\unadj}) + \var^{\sff} (\widehat{\IIFF}_{\unadj, 2, 2}^{\dag}) - 2 \cov^{\sff} (\hat{\tau}_{\unadj}, \widehat{\IIFF}_{\unadj, 2, 2}^{\dag})\\
    = & \ \frac{\pi_{0}}{\pi_{1}} \frac{1}{n} V_{n} (y (1)) \\
    + & \ \frac{\bar{\tau}^{2}}{\pi_{1}^{4}} \frac{\pi_{0}}{\pi_{1}} \frac{p}{n^2} \left( \pi_{1} - \frac{\pi_{0}}{n - 1} \right) \left\{ \begin{aligned}
       & \frac{2n^6\pi_{1}^2+n^5(-(p+6)\pi_{1}^{3}-10\pi_{1}^2-10\pi_{1})}{(n-1)(n-2)^{2}(n-3)(n-4)(n-5)} \\
       & + \frac{n^4(-(13p+50)\pi_{1}^3+(21p+124)\pi_{1}^2+28\pi_{1}+12)}{(n-1)(n-2)^{2}(n-3)(n-4)(n-5)} \\
       & + \frac{n^3((10p+432)\pi_{1}^3+(69p-572)\pi_{1}^2-(102p+14)\pi_{1}-48)}{(n-1)(n-2)^{2}(n-3)(n-4)(n-5)} \\
       & + \frac{n^2((148p-856)\pi_{1}^3+(-390p+936)\pi_{1}^2+(150p-4)\pi_{1}+120p+60)}{(n-1)(n-2)^{2}(n-3)(n-4)(n-5)} \\
       & + \frac{n((-240p+480)\pi_{1}^3+(252p-480)\pi_{1}^2+336p\pi_{1}-360p-24)}{(n-1)(n-2)^{2}(n-3)(n-4)(n-5)} \\
       & + \frac{240p\pi_{1}^2-480p\pi_{1}+240p}{(n-1)(n-2)^{2}(n-3)(n-4)(n-5)}
    \end{aligned} \right\} \\
    & \ + \frac{1}{\pi_{1}^{4}} \frac{\pi_{0}}{\pi_{1}} \frac{p}{n^2} \left( \pi_{1} - \frac{\pi_{0}}{n - 1} \right) \left\{ \begin{aligned}
        & \frac{n^{3}(-4\pi_{1}^{3}+2\pi_{1})}{(n-2)(n-3)(n-4)(n-5)} \\
        & + \frac{n^{2}((p+54)\pi_{1}^{3}-(34+p)\pi_{1}^{2}+2\pi_{1}-4)}{(n-2)(n-3)(n-4)(n-5)} \\
       & + \frac{n((20p-120)\pi_{1}^{3}+(-32p+80)\pi_{1}^{2}+(14p-8)\pi_{1}+12)}{(n-2)(n-3)(n-4)(n-5)}\\
       & + \frac{-40p\pi_{1}^{2}+60p\pi_{1}-24p-8}{(n-2)(n-3)(n-4)(n-5)} 
    \end{aligned}\right\} \bar{\tau}^{(2)} \\
    & \ + \frac{1}{\pi_{1}^{4}} \frac{\pi_{0}}{\pi_{1}} \frac{1}{n^2} \left( \pi_{1} - \frac{\pi_{0}}{n - 1} \right) \left( \pi_{1} - \frac{2 \pi_{0}}{n - 2} \right) \left\{ \begin{aligned}
       & \frac{n^{3}(3\pi_{1}^{2}-2\pi_{1})}{(n-3)(n-4)(n-5)} \\
       & + \frac{n^{2}(11\pi_{1}^{2}-20\pi_{1}+6)}{(n-3)(n-4)(n-5)} \\
       & + \frac{n(20\pi_{1}^{2}-60\pi_{1}+42)}{(n-3)(n-4)(n-5)} 
    \end{aligned} \right\} \bar{\tau}^{2}  \sum_{i=1}^{n} H_{i,i}^{2} \\
    & \ + \frac{1}{\pi_{1}^{4}} \frac{\pi_{0}}{\pi_{1}} \frac{1}{n^2} \left( \pi_{1} - \frac{\pi_{0}}{n - 1} \right) \left\{ \begin{aligned}
        & \frac{n^{3}(-3\pi_{1}^{3}+5\pi_{1}^{2}-2\pi_{1})}{(n-2)(n-3)(n-4)(n-5)}\\
        & + \frac{n^{2}(-11\pi_{1}^{3}+20\pi_{1}^{2}-14\pi_{1}+4)}{(n-2)(n-3)(n-4)(n-5)} \\
       & + \frac{n(-20\pi_{1}^{3}+49\pi_{1}^{2}-44\pi_{1}+12}{(n-2)(n-3)(n-4)(n-5)}\\
       & + \frac{-20\pi_{1}^{2}+40\pi_{1}-16}{(n-2)(n-3)(n-4)(n-5)} 
    \end{aligned}\right\} \bar{\tau}^{(2)}  \sum_{i=1}^{n} H_{i,i}^{2} \\
    & \ + \frac{1}{\pi_{1}^{4}} \frac{\pi_{0}}{\pi_{1}} \frac{1}{n^2} \left( \pi_{1} - \frac{\pi_{0}}{n - 1} \right) \left\{ \begin{aligned}
        & \frac{-4n^6\pi_{1}^2+n^5((2p+20)\pi_{1}^3+20\pi_{1}^2+16\pi_{1})}{(n-1)(n-2)^2(n-3)(n-4)(n-5)} \\
        & + \frac{n^4((24p-32)\pi_{1}^3-(40p+180)\pi_{1}^2-48\pi_{1}-16)}{(n-1)(n-2)^2(n-3)(n-4)(n-5)}\\
        & + \frac{n^3(-(54p+284)\pi_{1}^3+(-80p+780)\pi_{1}^2+(176p+48)\pi_{1}+48)}{(n-1)(n-2)^2(n-3)(n-4)(n-5)} \\
        & + \frac{n^2((-180p+776)\pi_{1}^3+(672p-1256)\pi_{1}^2)}{(n-1)(n-2)^2(n-3)(n-4)(n-5)}\\
        & + \frac{n^2(-(336p+48)\pi_{1}-192p-16)}{(n-1)(n-2)^2(n-3)(n-4)(n-5)}\\
        & + \frac{n((400p-480)\pi_{1}^3+(-616p+640)\pi_{1}^2)}{(n-1)(n-2)^2(n-3)(n-4)(n-5)}\\
        & + \frac{n((-368p+32)\pi_{1}+576p-48)}{(n-1)(n-2)^2(n-3)(n-4)(n-5)}\\
        & + \frac{-320p\pi_{1}^2+720p\pi_{1}-384p+32}{(n-1)(n-2)^2(n-3)(n-4)(n-5)}\\
    \end{aligned}\right\} \\
    & \quad \times \bar{\tau} \sum_{i=1}^{n} H_{i,i} y_{i} (1) \\
     & \ + \frac{1}{\pi_{1}^{4}} \frac{\pi_{0}}{\pi_{1}} \frac{1}{n^2} \left( \pi_{1} - \frac{\pi_{0}}{n - 1} \right) \left\{ \begin{aligned}
        & \frac{n^{4}(-4\pi_{1}^{3}+4\pi_{1}^{2})}{(n-2)(n-3)(n-4)(n-5)} \\
        & + \frac{n^{3}(-10\pi_{1}^{3}+18\pi_{1}^{2}-16\pi_{1})}{(n-2)(n-3)(n-4)(n-5)}\\
        & + \frac{n^{2}(-18\pi_{1}^{3}+92\pi_{1}^{2}-48\pi_{1}+16)}{(n-2)(n-3)(n-4)(n-5)}\\
        & + \frac{n(40\pi_{1}^{3}+22\pi_{1}^{2}-160\pi_{1}+48)}{(n-2)(n-3)(n-4)(n-5)}\\
        & + \frac{40\pi_{1}^{2}-160\pi_{1}+128}{(n-2)(n-3)(n-4)(n-5)}
    \end{aligned}\right\} \bar{\tau} \sum_{i=1}^{n} H_{i,i}^2  y_{i} (1)\\
    & \ + \frac{1}{\pi_{1}^{4}} \frac{\pi_{0}}{\pi_{1}} \frac{1}{n^2} \left( \pi_{1} - \frac{\pi_{0}}{n - 1} \right) \left\{ \begin{aligned}
        & \frac{2n^{4}\pi_{1}^{3}}{(n-2)(n-3)(n-4)(n-5)} + \frac{n^{3}(6\pi_{1}^{3}-24\pi_{1}^{2})}{(n-2)(n-3)(n-4)(n-5)}\\
        & + \frac{n^{2}(-32\pi_{1}^{2}+88\pi_{1})}{(n-2)(n-3)(n-4)(n-5)} + \frac{n(40\pi_{1}^{2}-8\pi_{1}-96)}{(n-2)(n-3)(n-4)(n-5)}\\
        & + \frac{-80\pi_{1}+96}{(n-2)(n-3)(n-4)(n-5)}
    \end{aligned}\right\} \\
    & \quad \times \bar{\tau} \sum_{1\leq i\neq j\leq n} H_{i,i} H_{i,j} y_{j} (1) \\
    & \ + \left\{ \frac{1}{\pi_{1}^{2}} \frac{\pi_{0}}{\pi_{1}} \frac{n\pi_{1}-1}{n^2(n-1)} + \frac{1}{\pi_{1}^{4}} \frac{\pi_{0}}{\pi_{1}} \frac{1}{n^2} \left( \pi_{1} - \frac{\pi_{0}}{n - 1} \right) \left( \begin{aligned}
        & \frac{n^{4}(4\pi_{1}^{3}+6\pi_{1}^{2}-4\pi_{1})}{n(n-2)(n-3)(n-4)(n-5)} \\ 
        & + \frac{n^{3}(-(2p+68)\pi_{1}^{3}+(2p-12)\pi_{1}^{2}+12\pi_{1}+4)}{n(n-2)(n-3)(n-4)(n-5)} \\
        & + \frac{n^{2}((-38p+160)\pi_{1}^{3}+(52p+94)\pi_{1}^{2})}{n(n-2)(n-3)(n-4)(n-5)} \\
        & + \frac{n^{2}(-(20p+48)\pi_{1}-8)}{n(n-2)(n-3)(n-4)(n-5)} \\
        & + \frac{n(40p\pi_{1}^{3}+(42p-200)\pi_{1}^{2})}{n(n-2)(n-3)(n-4)(n-5)}\\
         & + \frac{n((-84p+24)\pi_{1}+24p+20)}{n(n-2)(n-3)(n-4)(n-5)}\\
        & + \frac{40p\pi_{1}^{2}-120p\pi_{1}+72p-16}{n(n-2)(n-3)(n-4)(n-5)}
    \end{aligned}\right) \right\}\\
    & \quad \times \sum_{i=1}^{n} H_{i,i} y_{i} (1)^{2} \\
    & \ + \left\{ \frac{1}{\pi_{1}^2} \frac{\pi_{0}}{\pi_{1}}\frac{-n^{2}\pi_{1}+n\pi_{0}+1}{n^3(n-1)} + \frac{1}{\pi_{1}^{4}} \frac{\pi_{0}}{\pi_{1}} \frac{1}{n^2} \left( \pi_{1} - \frac{\pi_{0}}{n - 1} \right) \left( \begin{aligned}
        & \frac{n^{5}(\pi_{1}^{3}-\pi_{1}^{2})}{n(n-2)(n-3)(n-4)(n-5)} \\
        & + \frac{n^{4}(4\pi_{1}^{3}-12\pi_{1}^{2}+8\pi_{1})}{n(n-2)(n-3)(n-4)(n-5)} \\
        & + \frac{n^{3}(15\pi_{1}^{3}+11\pi_{1}^{2}-8\pi_{1}-8)}{n(n-2)(n-3)(n-4)(n-5)} \\
        & + \frac{n^{2}(-20\pi_{1}^{3}-128\pi_{1}^{2}+108\pi_{1}+8)}{n(n-2)(n-3)(n-4)(n-5)} \\
        & + \frac{n(-54\pi_{1}^{2}+116\pi_{1}-112)}{n(n-2)(n-3)(n-4)(n-5)} \\
        & + \frac{120\pi_{1}^{2}+16}{n(n-2)(n-3)(n-4)(n-5)}
    \end{aligned}\right) \right\}\\
    & \quad \times \sum_{i=1}^{n} H_{i,i}^{2} y_{i} (1)^{2} \\
    & \ +  \frac{1}{\pi_{1}^{4}} \frac{\pi_{0}}{\pi_{1}} \frac{1}{n^2} \left( \pi_{1} - \frac{\pi_{0}}{n - 1} \right) \left( \begin{aligned}
        & \frac{-n^{6}\pi_{1}^{3}+n^5(-2\pi_{1}^3+12\pi_{1}^2)+n^4(15\pi_{1}^3-6\pi_{1}^2-44\pi_{1})}{n(n-1)(n-2)^2(n-3)(n-4)(n-5)}\\
        & + \frac{n^3(16\pi_{1}^3-124\pi_{1}^2+112\pi_{1}+48)}{n(n-1)(n-2)^2(n-3)(n-4)(n-5)}\\
        & +\frac{n^2(-44\pi_{1}^3+74\pi_{1}^2+124\pi_{1}-192)}{n(n-1)(n-2)^2(n-3)(n-4)(n-5)}\\ 
        & + \frac{n(-80\pi_{1}^3+316\pi_{1}^2-448\pi_{1}+240)}{n(n-1)(n-2)^2(n-3)(n-4)(n-5)}\\
        & + \frac{-80\pi_{1}^2+160\pi_{1}-96}{n(n-1)(n-2)^2(n-3)(n-4)(n-5)}
    \end{aligned}\right) \trace^2 \left( \hat{\bSigma}_{y} \hat{\bSigma}^{-} \right)\\
    & \ + \frac{1}{\pi_{1}^{4}} \frac{\pi_{0}}{\pi_{1}} \frac{1}{n^2} \left( \pi_{1} - \frac{\pi_{0}}{n - 1} \right) \left\{ \begin{aligned}
        & \frac{n^{5}(-\pi_{1}^{3}+\pi_{1}^{2})+n^{4}(\pi_{1}^{3}+2\pi_{1}^{2}-4\pi_{1})}{n(n-2)(n-3)(n-4)(n-5)}  \\
        & + \frac{n^{3}(4\pi_{1}^{3}-7\pi_{1}^{2}+4\pi_{1}+4)+n^{2}(4\pi_{1}^{2}-8)}{n(n-2)(n-3)(n-4)(n-5)} \\
        & + \frac{n(-16\pi_{1}+20)-16}{n(n-2)(n-3)(n-4)(n-5)} \\
    \end{aligned}\right\} \trace \left( \hat{\bSigma}_{y} \hat{\bSigma}^{-} \hat{\bSigma}_{y} \hat{\bSigma}^{-} \right)\\
    & \ + \frac{1}{\pi_{1}^{4}} \frac{\pi_{0}}{\pi_{1}} \frac{1}{n^2} \left( \pi_{1} - \frac{\pi_{0}}{n - 1} \right) \left\{ \begin{aligned}
        & \frac{-n^{5}\pi_{1}^{3}+n^{4}(15\pi_{1}^{3}-2\pi_{1}^{2})+n^{3}((8p-50)\pi_{1}^{3}-(6p+22)\pi_{1}^{2}+16\pi_{1})}{n(n-2)(n-3)(n-4)(n-5)}\\
        & + \frac{n^{2}((-20p+40)\pi_{1}^{3}+(-18p+104)\pi_{1}^{2}+(24p-40)\pi_{1}-16)}{n(n-2)(n-3)(n-4)(n-5)}\\
        & + \frac{n((60p-80)\pi_{1}^{2}+(-8p+8)\pi_{1}-24p+48)-40p\pi_{1}+24p-32}{n(n-2)(n-3)(n-4)(n-5)}
    \end{aligned}\right\}  \\
    & \quad \times \sum_{1\leq i\neq j\leq n} H_{i,j} y_{i} (1) y_{j} (1)  \\
    & \ + \frac{1}{\pi_{1}^{4}} \frac{\pi_{0}}{\pi_{1}} \frac{1}{n^2} \left( \pi_{1} - \frac{\pi_{0}}{n - 1} \right) \left\{ \begin{aligned}
        & \frac{n^{4}(-2\pi_{1}^{3}+2\pi_{1}^{2})+n^{3}(-6\pi_{1}^{3}+16\pi_{1}^{2}-12\pi_{1})}{n(n-2)(n-3)(n-4)(n-5)} \\
        & + \frac{n^{2}(14\pi_{1}^{2}-24\pi_{1}+16)+4n\pi_{1}-16}{n(n-2)(n-3)(n-4)(n-5)}
    \end{aligned} \right\} \sum_{1\leq i\neq j\leq n} H_{i,i} H_{i,j} y_{j} (1)^{2} \\
    & \ + \frac{1}{\pi_{1}^{4}} \frac{\pi_{0}}{\pi_{1}} \frac{1}{n^2} \left( \pi_{1} - \frac{\pi_{0}}{n - 1} \right) \left\{ \begin{aligned}
        & \frac{-2n^{5}\pi_{1}^{3}+n^{4}(-2\pi_{1}^{3}+20\pi_{1}^{2})}{n(n-2)(n-3)(n-4)(n-5)} \\
        & \frac{n^{3}(12\pi_{1}^{3}-64\pi_{1})}{n(n-2)(n-3)(n-4)(n-5)} \\
        & + \frac{n^{2}(-68\pi_{1}^{2}+56\pi_{1}+64)}{n(n-2)(n-3)(n-4)(n-5)}\\
         & + \frac{n(72\pi_{1}-96)+32}{n(n-2)(n-3)(n-4)(n-5)}
    \end{aligned} \right\} \sum_{1\leq i\neq j\leq n} H_{i,i} H_{i,j} y_{i} (1) y_{j} (1).
\end{align*}

Then $\var^{\sff} (\hat{\tau}_{\adj, 2}^{\dag})$ can be approximated as follows
\begin{align*}
     & \ \var^{\sff} (\hat{\tau}_{\adj, 2}^{\dag}) \\
     = & \ \frac{\pi_{0}}{\pi_{1}} \frac{1}{n} V_{n} (y (1)) + \left( 2 \frac{\pi_{0}}{\pi_{1}^{2}} \frac{1}{n^2} - \frac{\pi_{0}}{\pi_{1}} \frac{p}{n^3} \right) \bar{\tau}^{2} \sum_{i=1}^{n} H_{i,i} + \left( - 2 \frac{\pi_{0}}{\pi_{1}^{2}} \frac{1}{n^2} + 3 \frac{\pi_{0}}{\pi_{1}} \frac{1}{n^2} \right) \bar{\tau}^{2} \sum_{i=1}^{n} H_{i,i}^{2} \\
     & \ + \left( - 4 \frac{\pi_{0}}{\pi_{1}^{2}} \frac{1}{n^2} + 2 \frac{\pi_{0}}{\pi_{1}} \frac{p}{n^3} \right) \bar{\tau} \sum_{i=1}^{n} H_{i,i} y_{i} (1) + 4 \left( \frac{\pi_{0}}{\pi_{1}} \right)^2 \frac{1}{n^2} \bar{\tau} \sum_{i=1}^{n} H_{i,i}^{2} y_{i} (1) \\
     & \ + 2 \frac{\pi_{0}}{\pi_{1}} \frac{1}{n^{2}} \bar{\tau} \sum_{1\leq i\neq j\leq n} H_{i,i} H_{i,j} y_{j} (1) +  \frac{\pi_{0}}{\pi_{1}^{2}} \frac{1}{n^2}  \sum_{i=1}^{n} H_{i,i} y_{i} (1)^{2} - \left( \frac{\pi_{0}}{\pi_{1}^{2}} \frac{1}{n^{2}} + \left( \frac{\pi_{0}}{\pi_{1}} \right)^{2}  \frac{1}{n^2} \right) \sum_{i=1}^{n} H_{i,i}^{2} y_{i} (1)^{2} \\
     & \ - \frac{\pi_{0}}{\pi_{1}} \frac{1}{n^3} \trace^{2} \left( \hat{\bSigma}_{y} \hat{\bSigma}^{-} \right) + \left( \frac{\pi_{0}}{\pi_{1}} \right)^{2} \frac{1}{n^2} \trace \left( \hat{\bSigma}_{y} \hat{\bSigma}^{-} \hat{\bSigma}_{y} \hat{\bSigma}^{-} \right) - \frac{\pi_{0}}{\pi_{1}} \frac{1}{n^2} \sum_{1\leq i\neq j\leq n} H_{i,j} y_{i} (1) y_{j} (1) \\
     & \ - 2 \frac{\pi_{0}}{\pi_{1}} \frac{1}{n^2} \sum_{1\leq i\neq j\leq n} H_{i,i} H_{i,j} y_{i} (1) y_{j} (1) 
      + o \left( n^{-1} \right)  \\
    = & \ \frac{\pi_{0}}{\pi_{1}} \frac{1}{n} V_{n} (y (1)) + 2 \frac{\pi_{0}}{\pi_{1}^{2}} \frac{1}{n^2} \bar{\tau}^{2} \left( p - \sum_{i=1}^{n} H_{i,i}^{2} \right) + \frac{\pi_{0}}{\pi_{1}} \frac{1}{n^2} \bar{\tau}^{2} \left( 3 \sum_{i=1}^{n} H_{i,i}^{2} - \frac{p^{2}}{n} \right) \\
    & \ + 4 \frac{\pi_{0}}{\pi_{1}} \frac{1}{n^2} \bar{\tau} \sum_{i=1}^{n} H_{i,i} \left( \frac{\pi_{0}}{\pi_{1}} H_{i,i} - \frac{1}{\pi_{1}} \right) y_{i} (1) + 2 \frac{\pi_{0}}{\pi_{1}} \frac{1}{n^2} \bar{\tau} \left( \frac{p}{n} \sum_{i=1}^{n} H_{i,i} y_{i} (1) + \sum_{1\leq i\neq j \leq n} H_{i,i} H_{i,j} y_{j} (1) \right)\\
    & \ + \frac{\pi_{0}}{\pi_{1}} \frac{1}{n^2} \sum_{i=1}^{n} H_{i,i} \left( \frac{1}{\pi_{1}} - \frac{1+\pi_{0}}{\pi_{1}} H_{i,i} \right) y_{i} (1)^2 - \frac{\pi_{0}}{\pi_{1}} \frac{1}{n} \left( \frac{1}{n} \sum_{i=1}^{n} H_{i,i} y_{i} (1) \right)^{2} \\
    & \ + \left( \frac{\pi_{0}}{\pi_{1}} \right)^{2} \frac{1}{n^2} \sum_{i=1}^{n} \sum_{j=1}^{n} H_{i,j}^{2} y_{i} (1) y_{j} (1) - \frac{\pi_{0}}{\pi_{1}} \frac{1}{n^2} \sum_{1\leq i\neq j\leq n} H_{i,j} \left( 1 + 2 H_{j,j} \right) y_{i}(1) y_{j}(1) + o \left( n^{-1} \right) \\
    = & \  \frac{\pi_{0}}{\pi_{1}} \frac{1}{n} V_{n} (y (1)) + 2 \frac{\pi_{0}}{\pi_{1}} \left( 1 + \frac{\pi_{0}}{\pi_{1}} \right) \frac{1}{n^2} \bar{\tau}^{2} \left( p - \sum_{i=1}^{n} H_{i,i}^{2} \right) + \frac{\pi_{0}}{\pi_{1}} \frac{1}{n^2} \bar{\tau}^{2} \left( 3 \sum_{i=1}^{n} H_{i,i}^{2} - \frac{p^{2}}{n} \right) \\
    & \ - 4 \left( \frac{\pi_{0}}{\pi_{1}} \right)^{2} \frac{1}{n^2} \bar{\tau} \sum_{i=1}^{n} H_{i,i} \left( 1 - H_{i,i} \right) y_{i} (1) - 4 \frac{\pi_{0}}{\pi_{1}} \frac{1}{n^2} \bar{\tau} \sum_{i=1}^{n} H_{i,i} y_{i}(1) \\
    & \ + 2 \frac{\pi_{0}}{\pi_{1}} \frac{1}{n^2} \bar{\tau} \left( \frac{p}{n} \sum_{i=1}^{n} H_{i,i} y_{i} (1) + \sum_{1\leq i\neq j \leq n} H_{i,i} H_{i,j} y_{j} (1) \right)\\
    & \ + \frac{\pi_{0}}{\pi_{1}} \left( 1 + \frac{\pi_{0}}{\pi_{1}} \right) \frac{1}{n^2} \sum_{i=1}^{n} H_{i,i} \left( 1 - H_{i,i} \right) y_{i} (1)^2 - \left( \frac{\pi_{0}}{\pi_{1}} \right)^{2} \frac{1}{n^2} \sum_{i=1}^{n} H_{ii}^2 y_{i} (1)^2  - \frac{\pi_{0}}{\pi_{1}} \frac{1}{n} \left( \frac{1}{n} \sum_{i=1}^{n} H_{i,i} y_{i} (1) \right)^{2} \\
    & \ + \left( \frac{\pi_{0}}{\pi_{1}} \right)^{2} \frac{1}{n^2} \sum_{i=1}^{n} \sum_{j=1}^{n} H_{i,j}^{2} y_{i} (1) y_{j} (1) - \frac{\pi_{0}}{\pi_{1}} \frac{1}{n^2} \sum_{1\leq i\neq j\leq n} H_{i,j} \left( 1 + 2 H_{j,j} \right) y_{i}(1) y_{j}(1) + o \left( n^{-1} \right) \\
    = & \ \frac{\pi_{0}}{\pi_{1}} \frac{1}{n} \left\{ \frac{1}{n} \sum_{i=1}^{n} \left( y_{i}(1) - \bar{\tau} \right)^{2}\right\} - \frac{\pi_{0}}{\pi_{1}} \frac{1}{n} \left\{ \frac{1}{n} \sum_{1\leq i\neq j\leq n} H_{i,j} \left( 1 + 2 H_{j,j} \right) \left( y_{i}(1) - \bar{\tau} \right) \left( y_{j}(1) - \bar{\tau} \right) \right\}\\
    & \ + 2 \frac{\pi_{0}}{\pi_{1}} \frac{1}{n^2} \bar{\tau} \sum_{i=1}^{n} H_{i,i}^2 y_{i} (1) - \frac{\pi_{0}}{\pi_{1}} \frac{1}{n^2} \bar{\tau}^{2} \left( p + 2 \sum_{i=1}^{n} H_{i,i} ^2 \right) 
    + \frac{\pi_{0}}{\pi_{1}} \frac{1}{n^2} \bar{\tau}^{2} \left( 3 \sum_{i=1}^{n} H_{i,i}^2 \right)\\
    & \ + \frac{\pi_{0}}{\pi_{1}} \left\{ 1 + \frac{\pi_{0}}{\pi_{1}} \right\} \frac{1}{n^2} \sum_{i=1}^{n} H_{i,i} \left( 1 - H_{i,i} \right) \left( y_{i} - \bar{\tau} \right)^{2} - 2 \frac{\pi_{0}}{\pi_{1}} \frac{1}{n^2} \bar{\tau} \sum_{i=1}^{n} H_{i,i}^2 y_{i} (1) + \frac{\pi_{0}}{\pi_{1}} \frac{1}{n^2} \bar{\tau}^{2}  \left( p - \sum_{i=1}^{n} H_{i,i}^2\right)  \\
    & \ - \frac{\pi_{0}}{\pi_{1}} \left\{ \frac{1}{n} \sum_{i=1}^{n} \left( 1 + H_{i,i} \right) \left( y_{i} (1) -\bar{\tau} \right) \right\}^{2} \\
    & +  \left( \frac{\pi_{0}}{\pi_{1}} \right)^2 \frac{1}{n^2} \left\{ \sum_{i=1}^{n} \sum_{j=1}^{n} H_{i,j}^{2} \left( y_{i} - \bar{\tau} \right) \left( y_{j} - \bar{\tau} \right) - \sum_{i=1}^{n} H_{i,i}^2 \left( y_{i} - \bar{\tau} \right)^{2} \right\} + o \left( n^{-1} \right)\\
    = & \ \frac{\pi_{0}}{\pi_{1}} \frac{1}{n} \left\{ \frac{1}{n} \sum_{i=1}^{n} \left( y_{i}(1) - \bar{\tau} \right)^{2} - \frac{1}{n} \sum_{1\leq i\neq j\leq n} H_{i,j} \left( 1 + 2 H_{j,j} \right) \left( y_{i}(1) - \bar{\tau} \right) \left( y_{j}(1) - \bar{\tau} \right) \right\}\\
    & \ + \frac{\pi_{0}}{\pi_{1}} \left\{ 1 + \frac{\pi_{0}}{\pi_{1}} \right\} \frac{1}{n^2} \sum_{i=1}^{n} H_{i,i} \left( 1 - H_{i,i} \right) \left( y_{i} - \bar{\tau} \right)^{2} - \frac{\pi_{0}}{\pi_{1}} \left\{ \frac{1}{n} \sum_{i=1}^{n} \left( 1 + H_{i,i} \right) \left( y_{i} (1) -\bar{\tau} \right) \right\}^{2} \\
    & \ + \left( \frac{\pi_{0}}{\pi_{1}} \right)^2 \frac{1}{n^2} \left\{ \sum_{i=1}^{n} \sum_{j=1}^{n} H_{i,j}^{2} \left( y_{i} - \bar{\tau} \right) \left( y_{j} - \bar{\tau} \right) - \sum_{i=1}^{n} H_{i,i}^2 \left( y_{i} - \bar{\tau} \right)^{2} \right\} + o \left( n^{-1} \right) \\
    = & \ \left( \frac{\pi_{0}}{\pi_{1}} \right) \frac{1}{n} V_{n} \left[ (y_{i} (1) - \bar{\tau}) - \sum_{j \neq i} H_{j, i} (y_{j} (1) - \bar{\tau}) \right] \\
    & \ + \left( \frac{\pi_{0}}{\pi_{1}} \right)^{2} \frac{1}{n} \left\{ \frac{1}{n} \sum_{i = 1}^{n} H_{i, i} (1 - H_{i, i}) (y_{i} (1) - \bar{\tau})^{2}  + \frac{1}{n} \sum_{1 \leq i \neq j \leq n} H_{i, j}^{2} (y_{i} (1) - \bar{\tau}) (y_{j} (1) - \bar{\tau}) \right\} + o \left( n^{-1} \right) \\
    = & \ \nu^{\sff\dag} + o \left( n^{-1} \right).
\end{align*}
\end{proof}

Finally, we discuss the difference between $\hat{\tau}_{\rm ld}$ in \citet{lei2021regression} and $\hat{\tau}_{\adj, 2}$. For simplicity, we do not center $y$ or $\bx$ and focus only on the treatment arm. $\hat{\tau}_{\rm ld}$ in \citet{lei2021regression} starts with $\hat{\tau}_{\adj, 1}$ but with $\hat{\bSigma}$ replaced by $\hat{\bSigma}_{1}$, which is the standard OLS adjusted estimator, and then remove the diagonal component by using $\hat{\bSigma}$. Therefore, the difference between these two estimators takes the following form:
\begin{align*}
\frac{1}{n} \sum_{i = 1}^{n} \left( \frac{t_{i}}{\pi_{1}} - 1 \right) \bx_{i}^{\top} (\hat{\bSigma}^{-1} - \hat{\bSigma}_{1}^{-1}) \bx_{i} \frac{t_{i} y_{i}}{\pi_{1}},
\end{align*}
which has order $\Vert \hat{\bSigma}^{-1} - \hat{\bSigma}_{1}^{-1} \Vert_{\op} = O_{\bbP} (p^{1 / 2} / n^{3 / 2})$ by standard matrix concentration inequalities such as matrix Bernstein inequality \citep{tropp2015introduction}.

\section{Variance Estimators}
\label{app:variance estimators}

In this section, we briefly discuss how to estimate the variance of the proposed HOIF-motivated estimator $\hat{\tau}_{\adj, 2}$ or $\hat{\tau}_{\adj, 3}$. In terms of $\hat{\tau}_{\db}$ or equivalently $\hat{\tau}_{\adj, 2}^{\dag}$, we refer readers to \citet{lu2025debiased}. Since we have been focusing on $\bar{\tau}$ rather than the ATE, we decide not to consider the issue of non-identifiability of $\cov (y (0), y (1))$ in the main text, which has been studied quite extensively in works such as \citet{robins1988confidence} and \citet{aronow2014sharp}. See also \citet{rosenbaum2002covariance} including the comment by \citet{robins2002covariance} for more related discussions.

To simplify our discussion, we only estimate the main terms in $\var^{\sff} (\hat{\tau}_{\adj, 2})$ and $\var^{\sff} (\hat{\tau}_{\adj, 3})$ and ignore all the terms of $o (1 / n)$ under Assumption \ref{as:regularity conditions}, as indicated in Theorem \ref{thm:HOIF-CRE, design-based} and Proposition \ref{prop:bias free HOIF variance}. Since $\var^{\sff} (\hat{\tau}_{\adj, 2})$ and $\var^{\sff} (\hat{\tau}_{\adj, 3})$ share the same main term, for ease of exposition, we denote this main term as

\begin{equation*}
\begin{split}
\nu^{\sff} \coloneqq & \ \underbrace{\left( \frac{\pi_{0}}{\pi_{1}} \right) \frac{1}{n} V_{n} \left[ y_{i} (1) - \sum_{j \neq i} H_{j, i} y_{j} (1) \right]}_{\eqqcolon \, \nu^{\sff}_{1} \, = \, O \left( \frac{1}{n} \right)} \\
& + \underbrace{\left( \frac{\pi_{0}}{\pi_{1}} \right)^{2} \frac{1}{n} \left\{ \frac{1}{n} \sum_{i = 1}^{n} H_{i, i} (1 - H_{i, i}) y_{i} (1)^{2}  + \frac{1}{n} \sum_{1 \leq i \neq j \leq n} H_{i, j}^{2} y_{i} (1) y_{j} (1) \right\}}_{\eqqcolon \, \nu^{\sff}_{2} \, = \, O \left( \frac{1}{n} \frac{p}{n} \right)}.
\end{split}
\end{equation*}
Note that \begin{align*}
    \nu_1^{\sff} & \ = \left(\frac{\pi_0}{\pi_1}\right) \frac{1}{n} \frac{1}{n} \|P\left(y(1) - Hy(1) + \diag\{H\}y(1)\right)\|_2^2\\
    & \ = \left(\frac{\pi_0}{\pi_1}\right)\frac{1}{n}\frac{1}{n} \|P\left(I_n - H + \diag\{H\}\right)y(1)\|_2^2\\
    &\ = \left(\frac{\pi_0}{\pi_1}\right)\frac{1}{n}\frac{1}{n} \left(\sum_{i=1}^{n}B_{i,i}y_i(1)^2 + \sum_{1\leq i \neq j\leq n}B_{i,j}y_{i}(1)y_{j}(1)\right).
\end{align*}
where $P=I_n - \frac{1}{n}11^\top,M = P - H + P\diag\{H\}, B = M^\top M$. So we can give an unbiased variance estimator as follows
\begin{align*}
    \hat{\nu}^{\sff} \coloneqq &\ \left( \frac{\pi_{0}}{\pi_{1}} \right) \frac{1}{n} \frac{1}{n} \left(
    \sum_{i=1}^{n}B_{i,i}\frac{t_i y_i^2}{\pi_1} + \sum_{1\leq i\neq j \leq n}B_{i,j}\frac{t_i y_i}{\pi_1}\frac{t_j y_j}{\pi_1} \right) \\
    &\ + \left( \frac{\pi_{0}}{\pi_{1}} \right)^{2} \frac{1}{n}\left\{ \frac{1}{n} \sum_{i=1}^{n} H_{i, i} (1 - H_{i, i}) \frac{t_i y_{i}^{2}}{\pi_1}  +\frac{1}{n} \sum_{1\leq i\neq j\leq n} H_{i, j}^{2}\frac{t_i y_{i}}{\pi_1} \frac{t_j y_{j}}{\pi_j} \right\}
\end{align*}
However, this unbiased variance estimator is under-conservative when the sample size is small, as shown in the simulation to be found in Appendix \ref{app:sim}.

Note that we can also rewrite $\nu_{1}^{\sff}$ as follows:
    \begin{align*}
        \nu_{1}^{\sff} = \frac{\pi_0}{\pi_1}\frac{1}{n}\frac{1}{n}\sum_{i=1}^{n}\left(
        y_i(1) - \sum_{j\neq i}H_{j,i}y_j(1) -
        \frac{1}{n}\sum_{i=1}^{n}(1+H_{i,i})y_i(1)
        \right)^2
    \end{align*}
As a result, we can consider the following more conservative variance estimator
\begin{align*}
    \hat{\nu}^{'\sff}  = &\ \frac{\pi_0}{\pi_1}\frac{1}{n}\frac{1}{n}\sum_{i=1}^{n}\frac{t_i}{\pi_1}\left(
    y_i -\sum_{j\neq i}H_{j,i}\frac{t_jy_j}{\pi_1} - \frac{1}{n}\sum_{k=1}^{n}(1+H_{k,k})\frac{t_ky_k}{\pi_1}
    \right)^2 \\
    &\ + \left( \frac{\pi_{0}}{\pi_{1}} \right)^{2} \frac{1}{n^2}\left\{  \sum_{i=1}^{n} H_{i, i} (1 - H_{i, i})\frac{ t_iy_{i}^{2}}{\pi_1}  +\sum_{i\neq j} H_{i, j}^{2}\frac{t_iy_i}{\pi_1} \frac{t_jy_j}{\pi_1} \right\}.         
\end{align*}

We then have the following result regarding the means of $\hat{\nu}^{\sff}$ and $\hat{\nu}^{'\sff}$ as an estimator of $\nu^{\sff}$.

To compute $\bbE[\hat{\nu}^{\sff}]$, we divide it into two parts to facilitate calculations.
    \begin{align*}
         \hat{\nu}^{\sff} \coloneqq &\ \underbrace{ \left( \frac{\pi_{0}}{\pi_{1}} \right) \frac{1}{n} \frac{1}{n} \left(
         \sum_{i=1}^{n}B_{i,i}\frac{t_i y_i^2}{\pi_1} + \sum_{1\leq i\neq j \leq n}B_{i,j}\frac{t_i y_i}{\pi_1}\frac{t_j y_j}{\pi_1} \right)}_{\eqqcolon \, \hat{\nu}^{\sff}_{1} } \\
         &\ + \underbrace{ \left( \frac{\pi_{0}}{\pi_{1}} \right)^{2} \frac{1}{n}\left\{ \frac{1}{n} \sum_{i=1}^{n} H_{i, i} (1 - H_{i, i}) \frac{t_i y_{i}^{2}}{\pi_1}  +\frac{1}{n} \sum_{1\leq i\neq j\leq n} H_{i, j}^{2}\frac{t_i y_{i}}{\pi_1} \frac{t_j y_{j}}{\pi_j} \right\}}_{\eqqcolon \, \hat{\nu}^{\sff}_{2}}.
    \end{align*}
    For the first term,
    \begin{align*}
        & \bbE\left[ \hat{\nu}^{\sff}_{1}\right] = \bbE \left[ \left( \frac{\pi_{0}}{\pi_{1}} \right) \frac{1}{n} \frac{1}{n} \left( \sum_{i=1}^{n}B_{i,i}\frac{t_i y_i^2}{\pi_1} + \sum_{1\leq i\neq j \leq n}B_{i,j}\frac{t_i y_i}{\pi_1}\frac{t_j y_j}{\pi_1} \right) \right] \\
        = & \ \left( \frac{\pi_{0}}{\pi_{1}} \right) \frac{1}{n} \frac{1}{n} \left(  \sum_{i=1}^{n} B_{i,i} y_{i}^{2} \frac{1}{\pi_{1}} \bbE \left[ t_{i} \right] +  \sum_{1\leq i\neq j \leq n} B_{i,j} y_{i}y_{j} \frac{1}{\pi_{1}^{2}} \bbE \left[ t_{i}t_{j} \right] \right)\\
        = & \ \left( \frac{\pi_{0}}{\pi_{1}} \right) \frac{1}{n} \frac{1}{n} \left( \sum_{i=1}^{n} B_{i,i} y_{i}^{2} + \sum_{1\leq i\neq j \leq n} B_{i,j} y_{i}y_{j}  \right) + O \left( n^{-2} \right) \\
        = & \ \nu^{\sff}_{1} + O \left( n^{-2} \right),
    \end{align*}
    the second term,
    \begin{align*}
        & \bbE \left[ \hat{\nu}^{\sff}_{2} \right] = \bbE \left[ \left( \frac{\pi_{0}}{\pi_{1}} \right)^{2} \frac{1}{n}\left\{ \frac{1}{n} \sum_{i=1}^{n} H_{i, i} (1 - H_{i, i}) \frac{t_i y_{i}^{2}}{\pi_1}  + \frac{1}{n} \sum_{1\leq i\neq j\leq n} H_{i, j}^{2} \frac{t_i y_{i}}{\pi_1} \frac{t_j y_{j}}{\pi_j} \right\} \right] \\
        = & \ \left( \frac{\pi_{0}}{\pi_{1}} \right)^{2} \frac{1}{n} \left\{ \frac{1}{n} \sum_{i=1}^{n} H_{i, i} (1 - H_{i, i}) y_{i}^{2} \frac{1}{\pi_{1}} \bbE \left[ t_{i} \right] + \frac{1}{n} \sum_{1\leq i\neq j\leq n} H_{i, j}^{2} y_{i}y_{j} \frac{1}{\pi_{1}^{2}} \bbE \left[ t_{i}t_{j} \right] \right\}\\
        = & \ \left( \frac{\pi_{0}}{\pi_{1}} \right)^{2} \frac{1}{n} \left\{ \frac{1}{n} \sum_{i=1}^{n} H_{i, i} (1 - H_{i, i}) y_{i}^{2} + \frac{1}{n} \sum_{1\leq i\neq j\leq n} H_{i, j}^{2} y_{i}y_{j} \right\}  + O \left( n^{-2} \right)\\
        = & \ \nu^{\sff}_{1} + O \left( n^{-2} \right).
    \end{align*}
    Combine with the two terms, we have 
    \begin{align*}
        \bbE[\hat{\nu}^{\sff}] & = \nu^{\sff} + O\left( n^{-2} \right).
    \end{align*}
    Next we compute the expectation $\bbE[\hat{\nu}^{'\sff}]$, 
    \begin{align*}
        \hat{\nu}^{'\sff} \coloneqq &\ \underbrace{ \frac{\pi_0}{\pi_1}\frac{1}{n}\frac{1}{n}\sum_{i=1}^{n}\frac{t_i}{\pi_1}\left( y_i -\sum_{j\neq i}H_{j,i}\frac{t_jy_j}{\pi_1} - \frac{1}{n}\sum_{k=1}^{n}(1+H_{k,k})\frac{t_ky_k}{\pi_1} \right)^2}_{\eqqcolon \, \hat{\nu}^{'\sff}_{1} } \\
         &\ + \underbrace{ \left( \frac{\pi_{0}}{\pi_{1}} \right)^{2} \frac{1}{n}\left\{ \frac{1}{n} \sum_{i=1}^{n} H_{i, i} (1 - H_{i, i}) \frac{t_i y_{i}^{2}}{\pi_1}  +\frac{1}{n} \sum_{1\leq i\neq j\leq n} H_{i, j}^{2}\frac{t_i y_{i}}{\pi_1} \frac{t_j y_{j}}{\pi_j} \right\}}_{\eqqcolon \, \hat{\nu}^{'\sff}_{2}}.
    \end{align*}
    We first give the decomposition of $\nu^{\sff}_{1}$ ,
    \begin{align*}
       &  \nu^{\sff}_{1} =  \frac{\pi_{0}}{\pi_{1}} \frac{1}{n} \frac{1}{n} \sum_{i = 1}^{n} \left\{ y_{i} + H_{i, i} y_{i} - \sum_{j = 1}^{n} H_{i, j} y_{j} - \left( \frac{1}{n} \sum_{k = 1}^{n} (1 + H_{k, k}) y_{k} \right) \right\}^{2}\\
       = & \ \frac{\pi_{0}}{\pi_{1}} \frac{1}{n} \frac{1}{n}  \sum_{i = 1}^{n} \left\{ \begin{aligned}
        & y_{i}^{2} + H_{i, i}^{2} y_{i}^{2} + \sum_{j = 1}^{n} H_{i, j}^{2} y_{j}^{2} + \sum_{1\leq j\neq k\leq n} H_{i,j} H_{i,k} y_{j}y_{k} + \frac{1}{n^{2}} \sum_{j = 1}^{n} (1 + H_{j, j})^{2} y_{j}^{2}\\
        & + \frac{1}{n^{2}} \sum_{1\leq j\neq k\leq n} (1 + H_{j, j}) (1+H_{k,k}) y_{j}y_{k} + 2 H_{i, i} y_{i}^{2} - 2 \sum_{j = 1}^{n} H_{i, j} y_{i}y_{j}\\
        & - 2 \frac{1}{n} \sum_{j = 1}^{n} (1 + H_{j, j}) y_{i}y_{j} - 2 \sum_{j = 1}^{n} H_{i,i} H_{i, j} y_{i}y_{j} - 2  \frac{1}{n} \sum_{j = 1}^{n} H_{i,i} (1 + H_{j, j}) y_{i}y_{j} \\
        & + 2 \frac{1}{n} \sum_{j = 1}^{n} H_{i, j} (1 + H_{j, j}) y_{j}^{2} +  2 \frac{1}{n} \sum_{1\leq j\neq k\leq n} H_{i,j} (1 + H_{k, k}) y_{j}y_{k}
        \end{aligned} \right\}\\
       = & \ \frac{\pi_{0}}{\pi_{1}} \frac{1}{n} \frac{1}{n} \left\{ \begin{aligned}
           & \sum_{i = 1}^{n} y_{i}^{2} + \sum_{1\leq i\neq j\leq n} H_{i, j}^{2} y_{j}^{2} + \sum_{1\leq i\neq j\neq k\leq n} H_{i,j} H_{i,k} y_{j}y_{k} +  \frac{1}{n} \sum_{i = 1}^{n} (1 + H_{i, i})^{2} y_{i}^{2} \\
           & + \frac{1}{n} \sum_{1\leq i\neq j\leq n} (1 + H_{i, i}) (1+H_{j,j}) y_{i}y_{j} - 2 \sum_{1\leq i\neq j\leq n} H_{i,j} y_{i}y_{j} - 2 \frac{1}{n} \sum_{i=1}^{n} (1+H_{i,i}) y_{i}^{2}\\
           & - 2 \frac{1}{n} \sum_{1\leq i\neq j\leq n} (1+H_{j,j}) y_{i}y_{j} + 2 \frac{1}{n} \sum_{1\leq i\neq j\leq n} H_{i,j} (1+H_{j,j}) y_{j}^{2} + 2 \frac{1}{n} \sum_{1\leq i\neq j\leq n} H_{i,j} (1+H_{j,j}) y_{i}y_{j} \\
           & +  2 \frac{1}{n} \sum_{1\leq i\neq j\neq k\leq n} H_{i,j} (1+H_{k,k}) y_{j}y_{k}
       \end{aligned} \right\} \\
       = & \ \frac{\pi_{0}}{\pi_{1}} \frac{1}{n} \frac{1}{n} \left\{ \begin{aligned}
            & \sum_{i = 1}^{n} y_{i}^{2} + \sum_{1\leq i\neq j\leq n} H_{i, j}^{2} y_{j}^{2} + \sum_{1\leq i\neq j\neq k\leq n} H_{i,j} H_{i,k} y_{j}y_{k} + \frac{1}{n} \sum_{1\leq i\neq j\leq n} (1 + H_{i, i}) (1+H_{j,j}) y_{i}y_{j}\\
            & - 2 \sum_{1\leq i\neq j\leq n} H_{i,j} y_{i}y_{j} - 2 \frac{1}{n} \sum_{1\leq i\neq j\leq n} (1+H_{j,j}) y_{i}y_{j} + 2 \frac{1}{n} \sum_{1\leq i\neq j\neq k\leq n} H_{i,j} (1+H_{k,k}) y_{j}y_{k}
       \end{aligned} \right\} \\
       & + O\left( n^{-2} \right),
       \end{align*}
    the expectation of $\hat{\nu}^{'\sff}_{1}$ is as follows,
    \begin{align*}
        & \bbE \left[ \hat{\nu}^{\sff}_{1} \right] = \frac{\pi_0}{\pi_1} \frac{1}{n} \frac{1}{n} \bbE \left[ \sum_{i=1}^{n} \frac{t_i}{\pi_1} \left( y_{i} + H_{i,i} \frac{t_{i}y_{i}}{\pi_{1}} - \sum_{j=1}^{n} H_{j,i} \frac{t_{j}y_{j}}{\pi_{1}} - \frac{1}{n} \sum_{k=1}^{n} (1+H_{k,k}) \frac{t_{k}y_{k}}{\pi_{1}} \right)^{2} \right] \\
        = & \ \frac{\pi_{0}}{\pi_{1}} \frac{1}{n} \frac{1}{n} \bbE \left[ \begin{aligned}
        & \sum_{i=1}^{n} y_{i}^{2} \frac{t_i}{\pi_1} + \sum_{1\leq i\neq j\leq n} H_{i,j}^{2} y_{j}^{2} \frac{t_{i}t_{j}}{\pi_{1}^{3}} + \sum_{1\leq i\neq j \neq k\leq n} H_{i,j} H_{i,k} y_{j}y_{k} \frac{t_{i}t_{j}t_{k}}{\pi_{1}^{3}} + \frac{1}{n^{2}} \sum_{i=1}^{n} (1+H_{i,i})^{2} y_{i}^{2} \frac{t_{i}}{\pi_{1}^{3}} \\
        & + \frac{1}{n^{2}} \sum_{1\leq i\neq j\leq n} (1+H_{j,j})^{2} y_{j}^{2} \frac{t_{i}t_{j}}{\pi_{1}^{3}} + 2 \frac{1}{n^{2}} \sum_{1\leq i\neq j\leq n} (1+H_{i,i}) (1+H_{j,j}) y_{i}y_{j} \frac{t_{i}t_{j}}{\pi_{1}^{3}} \\ 
        & + \frac{1}{n^{2}} \sum_{1\leq i\neq j \neq k\leq n} (1+H_{j,j}) (1+H_{k,k}) y_{j}y_{k} \frac{t_{i}t_{j}t_{k}}{\pi_{1}^{3}} - 2 \sum_{1\leq i\neq j\leq n} H_{i,j} y_{i}y_{j} \frac{t_{i}t_{j}}{\pi_{1}^{2}} \\ 
        & - 2 \frac{1}{n} \sum_{i=1}^{n} (1+H_{i,i}) y_{i}^{2} \frac{t_{i}}{\pi_{1}^{2}} - 2 \frac{1}{n} \sum_{1\leq i\neq j\leq n} (1+H_{j,j}) y_{i}y_{j} \frac{t_{i}t_{j}}{\pi_{1}^{2}} + 2 \frac{1}{n} \sum_{1\leq i\neq j\leq n} H_{i,j} (1+H_{j,j}) y_{j}^{2} \frac{t_{i}t_{j}}{\pi_{1}^{3}}\\
        & + 2 \frac{1}{n} \sum_{1\leq i\neq j\leq n} H_{i,j} (1+H_{j,j}) y_{i}y_{j} \frac{t_{i}t_{j}}{\pi_{1}^{3}} + 2 \frac{1}{n} \sum_{1\leq i\neq j\neq k\leq n} H_{i,j} (1+H_{k,k}) y_{j}y_{k} \frac{t_{i}t_{j}t_{k}}{\pi_{1}^{3}}\end{aligned} \right]\\
        = & \ \frac{\pi_{0}}{\pi_{1}} \frac{1}{n} \frac{1}{n} \left\{ \begin{aligned}
        & \sum_{i=1}^{n} y_{i}^{2} + \sum_{1\leq i\neq j\leq n} H_{i,j}^{2} y_{j}^{2} \frac{1}{\pi_{1}} + \sum_{1\leq i\neq j \neq k\leq n} H_{i,j} H_{i,k} y_{j}y_{k} + \frac{1}{n^{2}} \sum_{i=1}^{n} (1+H_{i,i})^{2} y_{i}^{2} \frac{1}{\pi_{1}^{2}} \\
        & + \frac{1}{n^{2}} \sum_{1\leq i\neq j\leq n} (1+H_{j,j})^{2} y_{j}^{2} \frac{1}{\pi_{1}} +  2 \frac{1}{n^{2}} \sum_{1\leq i\neq j\leq n} (1+H_{i,i}) (1+H_{j,j}) y_{i}y_{j} \frac{1}{\pi_{1}} \\
        & + \frac{1}{n^{2}} \sum_{1\leq i\neq j \neq k\leq n} (1+H_{j,j}) (1+H_{k,k}) y_{j}y_{k} - 2 \sum_{1\leq i\neq j\leq n} H_{i,j} y_{i}y_{j} - 2 \frac{1}{n} \sum_{i=1}^{n} (1+H_{i,i}) y_{i}^{2} \frac{1}{\pi_{1}} \\
        & - 2 \frac{1}{n} \sum_{1\leq i\neq j\leq n} (1+H_{j,j}) y_{i}y_{j} +  2 \frac{1}{n} \sum_{1\leq i\neq j\leq n} H_{i,j} (1+H_{j,j}) y_{j}^{2} \frac{1}{\pi_{1}} \\
        & + 2 \frac{1}{n} \sum_{1\leq i\neq j\leq n} H_{i,j} (1+H_{j,j}) y_{i}y_{j} \frac{1}{\pi_{1}} + 2 \frac{1}{n} \sum_{1\leq i\neq j\neq k\leq n} H_{i,j} (1+H_{k,k}) y_{j}y_{k} \end{aligned}  \right\}\\
        & \ + O \left( n^{-2} \right) \\
        = & \ \frac{\pi_{0}}{\pi_{1}} \frac{1}{n} \frac{1}{n} \left\{ \begin{aligned}
            & \sum_{i=1}^{n} y_{i}^{2} + \frac{1}{\pi_{1}} \sum_{1\leq i\neq j\leq n} H_{i,j}^{2} y_{j}^{2} +  \sum_{1\leq i\neq j \neq k\leq n} H_{i,j} H_{i,k} y_{j}y_{k} + \frac{1}{\pi_{1}} \frac{1}{n} \sum_{i=1}^{n} (1+H_{i,i})^{2} y_{i}^{2} \\
            & + \frac{1}{n} \sum_{1\leq i\neq j\leq n} (1+H_{i,i}) (1+H_{j,j}) y_{i}y_{j} - 2 \sum_{1\leq i\neq j\leq n} H_{i,j} y_{i}y_{j} - 2 \frac{1}{\pi_{1}} \frac{1}{n} \sum_{i=1}^{n} (1+H_{i,i}) y_{i}^{2} \\
            & - 2 \frac{1}{n} \sum_{1\leq i\neq j\leq n} (1+H_{j,j}) y_{i}y_{j} + 2 \frac{1}{\pi_{1}} \frac{1}{n} \sum_{1\leq i\neq j\leq n} H_{i,j} (1+H_{j,j}) y_{j}^{2} \\
            & + 2 \frac{1}{\pi_{1}} \frac{1}{n} \sum_{1\leq i\neq j\leq n} H_{i,j} (1+H_{j,j}) y_{i}y_{j} + 2 \frac{1}{n} \sum_{1\leq i\neq j\neq k\leq n} H_{i,j} (1+H_{k,k}) y_{j}y_{k} \\ 
            & + \frac{1}{\pi_{1}} \frac{\pi_{0}}{\pi_{1}} \frac{1}{n^{2}} \sum_{i=1}^{n} (1+H_{i,i})^{2} y_{i}^{2} + 2 \frac{\pi_{0}}{\pi_{1}} \frac{1}{n^{2}} \sum_{1\leq i\neq j\leq n} (1+H_{i,i}) (1+H_{j,j}) y_{i}y_{j}
        \end{aligned}  \right\} \\
        & + O \left( n^{-2} \right) \\
        = & \ \frac{\pi_{0}}{\pi_{1}} \frac{1}{n} \frac{1}{n} \left\{ \begin{aligned}
            & \sum_{i = 1}^{n} y_{i}^{2} + \frac{1}{\pi_{1}} \sum_{1\leq i\neq j\leq n} H_{i, j}^{2} y_{j}^{2} + \sum_{1\leq i\neq j\neq k\leq n} H_{i,j} H_{i,k} y_{j}y_{k} + \frac{1}{n} \sum_{1\leq i\neq j\leq n} (1 + H_{i, i}) (1+H_{j,j}) y_{i}y_{j}\\
            & - 2 \sum_{1\leq i\neq j\leq n} H_{i,j} y_{i}y_{j} - 2 \frac{1}{n} \sum_{1\leq i\neq j\leq n} (1+H_{j,j}) y_{i}y_{j} + 2 \frac{1}{n} \sum_{1\leq i\neq j\neq k\leq n} H_{i,j} (1+H_{k,k}) y_{j}y_{k}
          \end{aligned} \right\}\\
          & + O\left( n^{-2} \right)\\
        = & \ \frac{\pi_{0}}{\pi_{1}} \frac{1}{n} \frac{1}{n} \left\{ \begin{aligned}
            & \sum_{i = 1}^{n} y_{i}^{2} + \sum_{1\leq i\neq j\leq n} H_{i, j}^{2} y_{j}^{2} + \sum_{1\leq i\neq j\neq k\leq n} H_{i,j} H_{i,k} y_{j}y_{k} + \frac{1}{n} \sum_{1\leq i\neq j\leq n} (1 + H_{i, i}) (1+H_{j,j}) y_{i}y_{j}\\
            & - 2 \sum_{1\leq i\neq j\leq n} H_{i,j} y_{i}y_{j} - 2 \frac{1}{n} \sum_{1\leq i\neq j\leq n} (1+H_{j,j}) y_{i}y_{j} + 2 \frac{1}{n} \sum_{1\leq i\neq j\neq k\leq n} H_{i,j} (1+H_{k,k}) y_{j}y_{k}\\
            & - \sum_{1\leq i\neq j\leq n} H_{i, j}^{2} y_{j}^{2} +  \frac{1}{\pi_{1}} \sum_{1\leq i\neq j\leq n} H_{i, j}^{2} y_{j}^{2}
          \end{aligned} \right\} \\
          & + O\left( n^{-2} \right)\\
        = & \ \nu^{\sff}_{1} + \left( \frac{\pi_{0}}{\pi_{1}} \right)^{2} \frac{1}{n^{2}} \sum_{1\leq i\neq j\leq n} H_{i,j}^{2} y_{j}^{2} + O\left( n^{-2} \right).
    \end{align*}
    With $\hat{\nu}^{\sff}_{2} = \hat{\nu}^{'\sff}_{2} $, we have
    \begin{align*}
        \bbE[\hat{\nu}^{'\sff}] & = \nu^{\sff} + \left( \frac{\pi_{0}}{\pi_{1}} \right)^{2} \frac{1}{n^{2}} \sum_{1\leq i\neq j\leq n} H_{i,j}^{2} y_{j}^{2} + O\left( n^{-2} \right).
    \end{align*}

    Note that the extra term $\left( \frac{\pi_{0}}{\pi_{1}} \right)^{2} \frac{1}{n^{2}} \sum_{1\leq i\neq j\leq n} H_{i,j}^{2} y_{j}^{2}$ is non-negative with probability 1, so $\hat{\nu}^{'\sff}$ always has larger expectation than $\hat{\nu}^{\sff}$, hence more conservative. The order of $$\left( \frac{\pi_{0}}{\pi_{1}} \right)^{2} \frac{1}{n^{2}} \sum_{1\leq i\neq j\leq n} H_{i,j}^{2} y_{j}^{2}$$ is $O \left( \frac{p}{n^{2}} \right)$, so $\hat{\nu}^{'\sff}$ is still a consistent variance estimator when $p = o (n)$.

\section{A More Detailed Numerical Comparison of the Asymptotic Variances}
\label{app:var comparison}

In this section, as alluded to in Remark \ref{rem:var comparison}, we provide a more detailed comparison between $\nu^{\sff}$ and $\nu^{\sff\dag}$, the dominating part of the variances of $\hat{\tau}_{\adj, 2}$ and $\hat{\tau}_{\adj, 2}^{\dag}$ (and equivalently $\hat{\tau}_{\db}$), respectively. It is quite obvious that when the potential outcomes are more or less homogeneous, $\nu^{\sff\dag}$ should be smaller. To make a more detailed comparison, we first define $\nu^{\sff}_{c}$, which is the same as $\nu^{\sff\dag}$, except that we replace $\bar{\tau}$ by an arbitrary constant $c$. Thus $\nu^{\sff}_{c = \bar{\tau}} \equiv \nu^{\sff\dag}$ whereas $\nu^{\sff}_{c = 0} \equiv \nu^{\sff}$. $\nu^{\sff}_{c}$ can be similarly decomposed into $\nu^{\sff}_{c, 1}$ and $\nu^{\sff}_{c, 2}$:
\begin{align*}
\nu^{\sff}_{c, 1} \coloneqq & \ \frac{1}{n} \sum_{i = 1}^{n} \left\{ y_{i} (1) - c - \sum_{j \neq i} H_{i, j} (y_{j} (1) - c) - \frac{1}{n} \sum_{l = 1}^{n} \left( y_{l} (1) - c - \sum_{k \neq l} H_{l, k} (y_{k} (1) - c) \right) \right\}^{2} \\
= & \ \frac{1}{n} \sum_{i = 1}^{n} \left\{ y_{i} (1) - \sum_{j \neq i} H_{i, j} y_{j} (1) - \frac{1}{n} \sum_{l = 1}^{n} (1 + H_{l, l}) y_{l} (1) - c \left( H_{i, i} - \frac{p}{n} \right) \right\}^{2},
\end{align*}
and
\begin{align*}
\nu^{\sff}_{c, 2} \coloneqq & \ \frac{\pi_{0}}{\pi_{1}} \left\{ \frac{1}{n} \sum_{i = 1}^{n} H_{i, i} (1 - H_{i, i}) (y_{i} (1) - c)^{2} + \frac{1}{n} \sum_{1 \leq i \neq j \leq n} H_{i, j}^{2} (y_{i} (1) - c) (y_{j} (1) - c) \right\}.
\end{align*}

We next take derivatives with respect to $c$ in both terms:
\begin{align*}
\frac{\diff \nu^{\sff}_{c, 1}}{\diff c} = & - \frac{2}{n} \sum_{i = 1}^{n} \left\{ y_{i} (1) - \sum_{j \neq i} H_{i, j} y_{j} (1) - \frac{1}{n} \sum_{l = 1}^{n} (1 + H_{l, l}) y_{l} (1) - c \left( H_{i, i} - \frac{p}{n} \right) \right\} \left( H_{i, i} - \frac{p}{n} \right) \\
= & \ 2 c \frac{1}{n} \sum_{i = 1}^{n} \left( H_{i, i} - \frac{p}{n} \right)^{2} - \frac{2}{n} \sum_{i = 1}^{n} \left\{ y_{i} (1) - \sum_{j \neq i} H_{i, j} y_{j} (1) - \frac{1}{n} \sum_{l = 1}^{n} (1 + H_{l, l}) y_{l} (1) \right\} \left( H_{i, i} - \frac{p}{n} \right)
\end{align*}
and
\begin{align*}
\frac{\diff \nu^{\sff}_{c, 2}}{\diff c} = & - \frac{\pi_{0}}{\pi_{1}} \left\{ \frac{2}{n} \sum_{i = 1}^{n} H_{i, i} (1 - H_{i, i}) (y_{i} (1) - c) + \frac{2}{n} \sum_{1 \leq i \neq j \leq n} H_{i, j}^{2} (y_{i} (1) - c) \right\} \\
= & - \frac{\pi_{0}}{\pi_{1}} \left\{ \begin{aligned}
    & \frac{2}{n} \sum_{i = 1}^{n} H_{i, i} (1 - H_{i, i}) y_{i} (1) + \frac{2}{n} \sum_{1 \leq i \neq j \leq n} H_{i, j}^{2} y_{i} (1) \\
    & - c \left( \frac{2}{n} \sum_{i = 1}^{n} H_{i, i} (1 - H_{i, i}) + \frac{2}{n} \sum_{1 \leq i \neq j \leq n} H_{i, j}^{2} \right) 
\end{aligned} \right\} \\
= & - \frac{\pi_{0}}{\pi_{1}} \left\{ \frac{2}{n} \sum_{i = 1}^{n} H_{i, i} (1 - H_{i, i}) y_{i} (1) + \frac{2}{n} \sum_{1 \leq i \neq j \leq n} H_{i, j}^{2} y_{i} (1) - c \left( \frac{4 p}{n} - \frac{4}{n} \sum_{i = 1}^{n} H_{i, i}^{2} \right) \right\} \\
= & \ 2 c \frac{\pi_{0}}{\pi_{1}} \left( \frac{2 p}{n} - \frac{2}{n} \sum_{i = 1}^{n} H_{i, i}^{2} \right) - 2 \frac{\pi_{0}}{\pi_{1}} \left\{ \frac{1}{n} \sum_{i = 1}^{n} H_{i, i} (1 - H_{i, i}) y_{i} (1) + \frac{1}{n} \sum_{1 \leq i \neq j \leq n} H_{i, j}^{2} y_{i} (1) \right\}.
\end{align*}

Since obviously the second derivatives of the above terms with respect to $c$ are both non-negative, by setting the sum of the first derivatives to zero, we find the global minimum of $\nu^{\sff}_{c}$:
\begin{align*}
& \ 0 \equiv \frac{1}{2} \left( \frac{\diff \nu^{\sff}_{c, 1}}{\diff c} + \frac{\diff \nu^{\sff}_{c, 2}}{\diff c} \right) \\
\Rightarrow & \ c^{\ast} = \frac{\left\{ \begin{array}{c}
\frac{\pi_{0}}{\pi_{1}} \left\{ \frac{1}{n} \sum\limits_{i = 1}^{n} H_{i, i} (1 - H_{i, i}) y_{i} (1) + \frac{1}{n} \sum\limits_{1 \leq i \neq j \leq n} H_{i, j}^{2} y_{i} (1) \right\} \\
+ \, \frac{1}{n} \sum\limits_{i = 1}^{n} \left\{ y_{i} (1) - \sum\limits_{j \neq i} H_{i, j} y_{j} (1) - \frac{1}{n} \sum\limits_{l = 1}^{n} (1 + H_{l, l}) y_{l} (1) \right\} \left( H_{i, i} - \frac{p}{n} \right)
\end{array} \right\}}{\left\{ \frac{\pi_{0}}{\pi_{1}} \left( \frac{2 p}{n} - \frac{2}{n} \sum\limits_{i = 1}^{n} H_{i, i}^{2} \right) + \frac{1}{n} \sum\limits_{i = 1}^{n} \left( H_{i, i} - \frac{p}{n} \right)^{2} \right\}}.
\end{align*}

When $c^{\ast} \approx 0$ but $\bar{\tau} \neq 0$, $\nu^{\sff}$ should be smaller than $\nu^{\sff\dag}$. $c^{\ast} \approx 0$ amounts to the following criterion:
\begin{align}
& \ \frac{\pi_{0}}{\pi_{1}} \left\{ \frac{1}{n} \sum\limits_{i = 1}^{n} H_{i, i} (1 - H_{i, i}) y_{i} (1) + \frac{1}{n} \sum\limits_{1 \leq i \neq j \leq n} H_{i, j}^{2} y_{i} (1) \right\} \nonumber \\
\approx & - \frac{1}{n} \sum\limits_{i = 1}^{n} \left\{ H_{i, i} y_{i} (1) - \sum\limits_{j \neq i} H_{i, i} H_{i, j} y_{j} (1) - H_{i, i} \bar{\tau} - H_{i, i} \frac{1}{n} \sum\limits_{l = 1}^{n} H_{l, l} y_{l} (1) \right\} \nonumber \\
= & \ \left( \frac{p}{n} - 1 \right) \frac{1}{n} \sum_{i = 1}^{n} H_{i, i} y_{i} (1) + \frac{p}{n} \bar{\tau} - \frac{1}{n} \sum_{i = 1}^{n} H_{i, i}^{2} y_{i} (1) + \frac{1}{n} \sum\limits_{i = 1}^{n} H_{i, i} \sum_{j = 1}^{n} H_{i, j} y_{j} (1) \nonumber \\
\Rightarrow & \ \bar{\tau} \approx \frac{1}{p} \sum_{i = 1}^{n} \left\{ \frac{\pi_{0}}{\pi_{1}} H_{i, i} (1 - H_{i, i}) + \frac{n - p}{n} H_{i, i} + \left( 1 - \frac{\pi_{0}}{\pi_{1}} \right) H_{i, i}^{2} \right\} y_{i} (1) \\
& + \frac{1}{p} \sum_{i = 1}^{n} \sum_{j = 1}^{n} \left( \frac{\pi_{0}}{\pi_{1}} H_{i, j} - H_{i, i} \right) H_{i, j} y_{j} (1) \nonumber \\
\Rightarrow & \ \frac{1}{n} \sum_{i = 1}^{n} y_{i} (1) \approx \frac{1}{n} \sum_{i = 1}^{n} \underbrace{\frac{n}{p} \left\{ \left( \frac{\pi_{0}}{\pi_{1}} + \frac{n - p}{n} \right) H_{i, i} + \left( 1 - \frac{2 \pi_{0}}{\pi_{1}} \right) H_{i, i}^{2} + \sum_{j = 1}^{n} \left( \frac{\pi_{0}}{\pi_{1}} H_{j, i} - H_{j, j} \right) H_{j, i} \right\}}_{\eqqcolon \, \omega_{i}} y_{i} (1). \label{constraint}
\end{align}
Thus if $\bH$ is such that $\omega_{i} \approx 1$, $c^{\ast} \approx 0$. Let $\bm{\omega} \coloneqq (\omega_{1}, \cdots, \omega_{n})^{\top}$. Then we have the following

\begin{conjecture}
\label{conj}
The data generating process for which $\nu^{\sff} < \nu^{\sff\dag}$ and $\nu^{\sff} < \var (\hat{\tau}_{\unadj})$ is not vacuous.
\end{conjecture}

We prove the above conjecture using a computer-assisted proof -- as a result, we decide not to upgrade its status to Lemma/Proposition/Theorem. Finding such a data generating process with a computer-assisted strategy is a bit contrived, but we do not yet have a more interpretable condition on when the desired criterion holds.

The computer-assisted proof can be formulated as a simple optimization problem. We simply design an optimization program to minimize $\nu^{\sff}$ while making sure $\nu^{\sff} < \nu^{\sff\dag}$ and $\nu^{\sff} < \var (\hat{\tau}_{\unadj})$. With this desiderata in mind, we eventually solve the mathematical program \eqref{program} using the \texttt{nloptr} \texttt{R} package, which is an \texttt{R} interface to the nonlinear optimization python package \href{https://nlopt.readthedocs.io/en/latest/}{\texttt{NLopt}} \citep{NLopt}.

To introduce the mathematical program \eqref{program}, we need to define some additional notation. We first define $P_c \coloneqq \mathbf{I}_n - \frac{1}{n}\mathbbm{1}\mathbbm{1}^\top, D_h \coloneqq \diag (\bH), D_{h(1-h)} = \diag (\bH \circ (1 - \bH)), D_{h^2} \coloneqq \diag (\bH \circ \bH)$. Then $\var (\hat{\tau}_{\unadj}) = \frac{\pi_0}{\pi_1} \frac{1}{n^2} \bv^\top P_c \bv\coloneqq \bv^\top Q_0 \bv$, as now $\bv$ is the vector of potential outcomes in the treatment group. We similarly write $\var (\hat{\tau}_{\db})$ in the matrix form
\begin{align*}
    \nu_1^\sff = &\ \left(\frac{\pi_0}{\pi_1}\right)\frac{1}{n} \bv^\top P_c \left(\bI_n - \bH + D_h\right)P_c \left(\bI_n - \bH +D_h\right)P_c \bv = \bv^\top Q_{11} \bv, \\
    \nu_2^\sff = &\ \left(\frac{\pi_0}{\pi_1}\right)^2 \left(
    \frac{1}{n} \bv^\top P_c D_{h(1-h)}P_c \bv + \frac{1}{n} \bv^\top P_c \bH^{\circ 2} P_c \bv - \frac{1}{n} \bv^\top P_c D_{h^2}P_c \bv    \right) = \bv^\top Q_{12} \bv, \\
    \nu^{\sff\dag} = &\ \bv^\top Q_1 \bv.
\end{align*}
where 
\begin{align*}
Q_1 & = Q_{11} + Q_{12} \\
Q_{11} & = \left(\frac{\pi_0}{\pi_1}\right)\frac{1}{n} P_c \left(\bI_n - \bH + D_h\right)P_c \left(\bI_n - \bH +D_h\right)P_c, \\
Q_{12} & = \left(\frac{\pi_0}{\pi_1}\right)^2 \frac{1}{n} \left(
P_c D_{h(1-h)}P_c + P_c\bH^{\circ 2} P_c - P_c D_{h^2}P_c \right).
\end{align*}
Similarly, we can express $\nu^{\sff}$ in the matrix form
\begin{align*}
    \nu^{\sff} = &\ \bv^\top Q_2 \bv  \\
    Q_2  = &\ Q_{21} + Q_{22}, \\
    Q_{21}  =&\ \frac{\pi_0}{\pi_1}\frac{1}{n^2}\left(\bI_n - \bH + D_h\right)P_c \left(\bI_n - \bH + D_h\right), \\
    Q_{22} = &\ \left(\frac{\pi_0}{\pi_1}\right)^2\frac{1}{n^2}\left(D_h - 2D_{h^2}+\bH^{\circ 2} \right).
\end{align*}
Then we solve the following
\begin{equation}
\label{program}
\begin{split}
    \min_{\bv \in \bbR^{n}} &\ \bv^\top Q_2 \bv  \\
    \textrm{s.t.} &\ \bv^\top \left(Q_2 - Q_1\right)\bv + 0.1 \leq 0, \\
    &\ \bv^\top \left(Q_2 - Q_0\right)\bv + 0.05 \leq 0, \\
    &\ -\bv^\top Q_0 \bv + 0.01 \leq 0.   
\end{split}
\end{equation}

We fix $\pi_{1} = 0.5$ and $\bx \sim t_3(\bm{0}, \bSigma)$, the centered multivariate $t$-distribution with shape $\bSigma$, with $\bSigma_{i,j} = 0.1^{|i-j|}$, for $i,j=1,\cdots,p$, and vary the following in the simulation:
\begin{itemize}
    \item sample size: $n\in\{50,100,200,500,1000,2000\}$
    \item covariate dimension: $\alpha \coloneqq \frac{p}{n} \in \{0.05,0.1,0.2,\cdots,0.7\}$, and $p=\lceil n \alpha \rceil$.
\end{itemize}
The potential outcomes are generated according to $y_i (1) = \bv + \epsilon_i$ for $i = 1, \cdots, n$, where the exogenous error distribution $\epsilon$ follows univariate $t_{3}$ distribution. The relative efficiencies based on the exact and approximate variances of different estimators, benchmarked by $\var^{\sff} (\hat{\tau}_{\unadj})$, are displayed in Figure \ref{fig:oracle_var_opt}. It is quite obvious that $\var^{\sff} (\hat{\tau}_{\adj, 2})$ is generally lower than $\var^{\sff} (\hat{\tau}_{\adj, 2}^{\dag})$, despite not centering the potential outcomes. 
\begin{figure}[H]
    \centering
    \includegraphics[width=\linewidth,page=1]{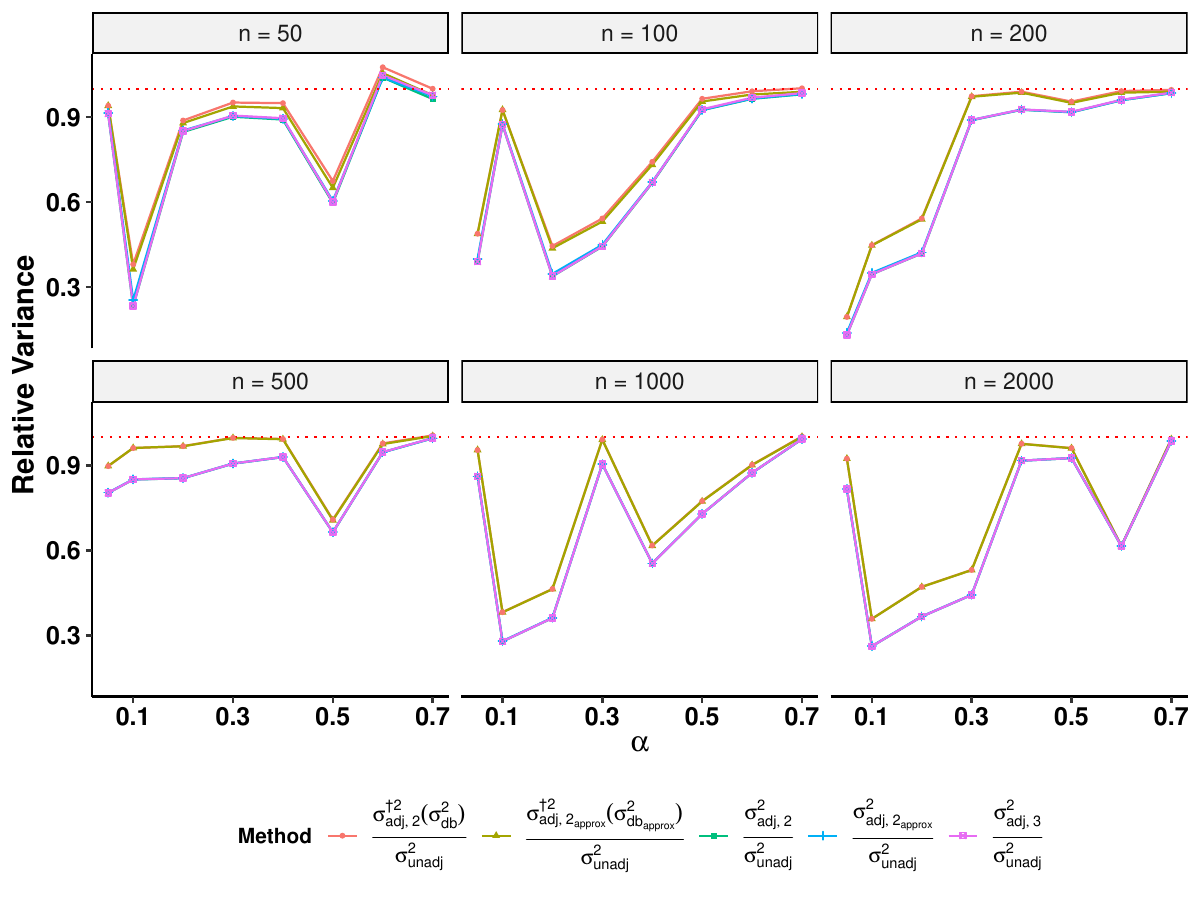}
    \caption{Relative efficiency of $\hat{\tau}_{\adj, 2}^{\dag}$ (or equivalently $\hat{\tau}_{\db}$), $\hat{\tau}_{\adj, 2}$, $\hat{\tau}_{\adj, 3}$ based on exact and approximate formula.}
    \label{fig:oracle_var_opt}
\end{figure}

Finally, it is worth pointing out that another possible approach to finding such data generating processes, in particular in actual RCT data analysis, is via the ``conditional generative learning'' paradigm \citep{athey2024using}, which we leave to future work.

\section{Preparatory Technical Results}
\label{app:technical lemmas}

In this section, we state several useful but well known technical results useful in the above proofs. We first state the following useful results regarding the projection ``hat'' matrix $\bH$.

\begin{lemma}
\label{lem:hat}
The following assertions hold for $\bH$:
\begin{align*}
\bH^{2} = \bH, \trace (\bH) \equiv p \wedge n, \text{ and } \mathbbm{1}^{\top} \bH \bv \equiv \sum_{i = 1}^{n} \sum_{j = 1}^{n} H_{i, j} v_{j} \equiv 0 \text{ for any $\bv \coloneqq (v_{1}, \cdots, v_{n})^{\top}$}.
\end{align*}
\end{lemma}

\begin{proof}
The first two assertions are simple facts of projection ``hat'' matrices. The third claim follows by observing that $\sum_{i = 1}^{n} \left( \bx_{i} - \bar{\bx} \right) \equiv 0$.
\end{proof}

The next lemma characterizes low-order moments of the treatment assignment indicators under CRE.
\begin{lemma}
\label{lem:CRE}
Under CRE, we have
\begin{equation}
\label{cre cov}
\bbE \left[ \prod_{i = 1}^{j} t_{i} \right] = \pi_{1} \prod_{i = 1}^{j - 1} \frac{n \pi_{1} - i}{n - i} = \pi_{1} \prod_{i = 1}^{j - 1} \frac{n_{1} - i}{n - i}.
\end{equation}
\end{lemma}

The following result on the Bernoulli sampling is elementary. We state it nonetheless for completeness.
\begin{lemma}
\label{lem:Bernoulli}
Under the Bernoulli sampling, we have, $t_{i} \overset{\rm i.i.d.}{\sim} \mathrm{Bernoulli} (\pi_{1})$, so
\begin{equation}
\bbE \left[ \prod_{i = 1}^{j} t_{i} \right] \equiv \pi_{1}^{j}.
\end{equation}
\end{lemma}

As an almost immediate corollary of Lemma \ref{lem:CRE}, we have the following result useful for deriving the bias and variance of $\hat{\tau}_{\adj, 2}^{\dag}$ (or equivalently $\hat{\tau}_{\db}$).

\begin{lemma}
\label{lem:cre_db}
Under CRE, the following assertions hold:
\begin{align*}
    \bbE \left[ \prod_{i=1}^{m} t_{i} \hat{\tau}_{\unadj} \right] = \prod_{i=1}^{m} \frac{n_{1} - i}{n - i} \cdot \bar{\tau} + \prod_{i=1}^{m-1} \frac{n_{1} - i}{n - i} \frac{\pi_{0}}{n - m} \left( \sum_{j=1}^{m} y_{j} (1) \right),
\end{align*}
and
\begin{align*}
  \bbE \left[ \prod_{i=1}^{m} t_{i} (\hat{\tau}_{\unadj})^{2} \right] = & \ \frac{1}{\pi_{1}} \prod_{i=1}^{m+1} \frac{n\pi_{1}-i}{n-i} \cdot \bar{\tau}^{2} + \frac{\pi_{0}}{\pi_{1}} \prod_{i=1}^{m} \frac{n\pi_{1}-i}{n-i} \frac{1}{n-(m+1)} \cdot \bar{\tau}^{(2)} \\
   & \ + \frac{\pi_{0}}{\pi_{1}} \prod_{i=1}^{m} \frac{n\pi_{1}-i}{n-i} \frac{2}{n-(m+1)} \cdot \bar{\tau} \left( \sum_{i=1}^{n} y_{i}(1) \right)\\
   & \ - \frac{\pi_{0}}{\pi_{1}} \prod_{i=1}^{m} \frac{n\pi_{1}-i}{n-i} \frac{2}{n(n-(m+1))} \left (\sum_{i,j\in\{1,\cdots,m \}} y_{i}(1) y_{j}(1) \right)\\
   & \ + \frac{\pi_{0}}{\pi_{1}} \prod_{i=1}^{m-1} \frac{n\pi_{1}-i}{n-i} \frac{1}{n(n-m)} \left( \sum_{i=1}^{m} y_{i}(1) \right)^{2},
\end{align*}
where $i = 1, \cdots , m$ represents different subscripts.
\end{lemma}

\begin{proof}

\begin{align*}
    & \bbE \left[ \prod_{i=1}^{m} t_{i} \left( \frac{1}{n} \sum_{k=1}^{n} \frac{t_{k}}{\pi_{1}} y_{k} (1) \right) \right] = \frac{1}{\pi_{1}} \frac{1}{n} \sum_{k=1}^{n} y_{k} (1) \bbE \left[ t_{k} \prod_{i=1}^{m} t_{i} \right] \\
    = & \ \frac{1}{\pi_{1}} \frac{1}{n} \left\{ \sum_{k \neq \{1,\cdots,m\}} y_{k} (1) \bbE \left[ t_{k} \prod_{i=1}^{m} t_{i} \right] + \sum_{j=1}^{m} y_{j} (1) \bbE \left[ t_{j} \prod_{i=1}^{m} t_{i} \right] \right\} \\
    = & \ \frac{1}{\pi_{1}} \frac{1}{n} \left\{ \sum_{k \neq \{1,\cdots,m\}} y_{k} (1) \bbE \left[ t_{k} \prod_{i=1}^{m} t_{i} \right] + \sum_{j=1}^{m} y_{j} (1) \bbE \left[\prod_{i=1}^{m} t_{i} \right] \right\} \\
    = & \ \frac{1}{n} \left\{ \prod_{i=1}^{m} \frac{n_{1} - i}{n - i} \left( \sum_{k=1}^{n} y_{k} (1) - \sum_{j=1}^{m} y_{j} (1) \right) + \prod_{i=1}^{m-1} \frac{n_{1} - i}{n - i} \sum_{j=1}^{m} y_{j} (1) \right\} \\
    = & \ \prod_{i=1}^{m} \frac{n_{1} - i}{n - i} \left( \frac{1}{n}\sum_{i=1}^{n} y_{k} (1) \right) + \prod_{i=1}^{m-1} \frac{n_{1} - i}{n - i} \left( 1 - \frac{n_{1} - m}{n - m} \right) \frac{1}{n}\sum_{j=1}^{m} y_{j} (1) \\
    = & \ \prod_{i=1}^{m} \frac{n_{1} - i}{n - i} \cdot \bar{\tau} + \prod_{i=1}^{m-1} \frac{n_{1} - i}{n - i} \frac{\pi_{0}}{n - m} \left( \sum_{j=1}^{m} y_{j} (1) \right).
\end{align*}

\begin{align*}
& \ \bbE [t_{j} \hat{\tau}_{\unadj}^{2}] = \bbE \left[ t_{j} \frac{1}{n^{2}} \sum_{k = 1}^{n} \sum_{l = 1}^{n} \frac{t_{k} t_{l}}{\pi_{1}^{2}} y_{k} (1) y_{l} (1) \right] \\
= & \ \frac{1}{\pi_{1}^{2}} \frac{1}{n^{2}} \sum_{1 \leq k \neq l \leq n} y_{k} (1) y_{l} (1) \bbE [t_{j} t_{k} t_{l}] + \frac{1}{\pi_{1}^{2}} \frac{1}{n^{2}} \sum_{k = 1}^{n} y_{k} (1)^{2} \bbE [t_{j} t_{k}] \\
= & \ \frac{1}{\pi_{1}} \frac{1}{n^{2}} \sum_{k \neq l \neq j} y_{k} (1) y_{l} (1) \left( \pi_{1} - \frac{\pi_{0}}{n - 1} \right) \left( \pi_{1} - \frac{2 \pi_{0}}{n - 2} \right) + \frac{2}{\pi_{1}} \frac{1}{n^{2}} \sum_{k \neq j} y_{k} (1) y_{j} (1) \left( \pi_{1} - \frac{\pi_{0}}{n - 1} \right) \\
& + \frac{1}{\pi_{1}} \frac{1}{n^{2}} \sum_{k \neq j} y_{k} (1)^{2} \left( \pi_{1} - \frac{\pi_{0}}{n - 1} \right) + \frac{1}{\pi_{1}} \frac{1}{n^{2}} y_{j} (1)^{2} \\
= & \ \frac{1}{\pi_{1}} \left( \pi_{1} - \frac{\pi_{0}}{n - 1} \right) \left( \pi_{1} - \frac{2 \pi_{0}}{n - 2} \right) \left\{ \bar{\tau}^{2} - \frac{2}{n^{2}} \sum_{k \neq j} y_{k} (1) y_{j} (1) - \frac{1}{n^{2}} \sum_{k \neq j} y_{k} (1)^{2} - \frac{1}{n^{2}} y_{j} (1)^{2} \right\} \\
& + \frac{2}{\pi_{1}} \left( \pi_{1} - \frac{\pi_{0}}{n - 1} \right) \left( \bar{\tau} \frac{1}{n} y_{j} (1) - \frac{1}{n^{2}} y_{j} (1)^{2} \right) + \frac{1}{\pi_{1}} \left( \pi_{1} - \frac{\pi_{0}}{n - 1} \right) \left( \frac{1}{n^{2}} \sum_{k = 1}^{n} y_{k} (1)^{2} - \frac{1}{n^{2}} y_{j} (1)^{2} \right) \\
& + \frac{1}{\pi_{1}} \frac{1}{n^{2}} y_{j} (1)^{2} \\
= & \ \frac{1}{\pi_{1}} \left( \pi_{1} - \frac{\pi_{0}}{n - 1} \right) \left( \pi_{1} - \frac{2 \pi_{0}}{n - 2} \right) \left\{ \bar{\tau}^{2} - 2 \bar{\tau} \frac{1}{n} y_{j} (1) + \frac{2}{n^{2}} y_{j} (1)^{2} - \frac{1}{n^{2}} \sum_{k = 1}^{n} y_{k} (1)^{2} \right\} \\
& + \frac{2}{\pi_{1}} \left( \pi_{1} - \frac{\pi_{0}}{n - 1} \right) \left( \bar{\tau} \frac{1}{n} y_{j} (1) - \frac{1}{n^{2}} y_{j} (1)^{2} \right) + \frac{1}{\pi_{1}} \left( \pi_{1} - \frac{\pi_{0}}{n - 1} \right) \left( \frac{1}{n^{2}} \sum_{k = 1}^{n} y_{k} (1)^{2} - \frac{1}{n^{2}} y_{j} (1)^{2} \right)\\
& + \frac{1}{\pi_{1}} \frac{1}{n^{2}} y_{j} (1)^{2} \\
= & \ \frac{1}{\pi_{1}} \left( \pi_{1} - \frac{\pi_{0}}{n - 1} \right) \left( \pi_{1} - \frac{2 \pi_{0}}{n - 2} \right) \bar{\tau}^{2} + 2 \frac{\pi_{0}}{\pi_{1}} \left( \pi_{1} - \frac{\pi_{0}}{n - 1} \right) \bar{\tau} \frac{1}{n - 2} y_{j} (1) \\
& + \frac{\pi_{0}}{\pi_{1}} \left( \pi_{1} - \frac{\pi_{0}}{n - 1} \right) \frac{1}{n (n - 2)} \sum_{k = 1}^{n} y_{k} (1)^{2} \\
& + \frac{1}{\pi_{1}} \left\{ 2 \left( \pi_{1} - \frac{\pi_{0}}{n - 1} \right) \left( \pi_{1} - \frac{2 \pi_{0}}{n - 2} \right) - 3 \left( \pi_{1} - \frac{\pi_{0}}{n - 1} \right) + 1 \right\} \frac{1}{n^{2}} y_{j} (1)^{2}\\
= & \ \frac{1}{\pi_{1}} \left( \pi_{1} - \frac{\pi_{0}}{n - 1} \right) \left( \pi_{1} - \frac{2 \pi_{0}}{n - 2} \right) \bar{\tau}^{2} + \frac{\pi_{0}}{\pi_{1}} \left( \pi_{1} - \frac{\pi_{0}}{n - 1} \right) \frac{1}{n - 2} \bar{\tau}^{(2)} + \frac{\pi_{0}}{\pi_{1}} \left( \pi_{1} - \frac{\pi_{0}}{n - 1} \right) \frac{2}{n - 2} \bar{\tau} y_{j} (1) \\
& \ - \frac{\pi_{0}}{\pi_{1}} \left( \pi_{1} - \frac{\pi_{0}}{n - 1} \right) \frac{2}{n(n-2)} y_{j} (1)^{2} + \frac{\pi_{0}}{\pi_{1}} \frac{1}{n(n-1)} y_{j} (1)^{2} 
\end{align*}

\begin{align*}
& \ \bbE [t_{i} t_{j} \hat{\tau}_{\unadj}^{2}] = \bbE \left[ t_{i} t_{j} \frac{1}{n^{2}} \sum_{k = 1}^{n} \sum_{l = 1}^{n} \frac{t_{k} t_{l}}{\pi_{1}^{2}} y_{k} (1) y_{l} (1) \right] \\
= & \ \frac{1}{n^{2}} \frac{1}{\pi_{1}^{2}} \left( \sum_{1 \leq k \neq l \leq n} y_{k} (1) y_{l} (1) \bbE [t_{i} t_{j} t_{k} t_{l}] + \sum_{k = 1}^{n} y_{k} (1)^{2} \bbE [t_{i} t_{j} t_{k}] \right) \\
= & \ \frac{1}{n^{2}} \frac{1}{\pi_{1}^{2}} \left( \sum_{k \neq l \neq i \neq j} y_{k} (1) y_{l} (1) \bbE [t_{i} t_{j} t_{k} t_{l}] + 2 \sum_{k \neq i \neq j} y_{k} (1) (y_{i} (1) + y_{j} (1)) \bbE [t_{i} t_{j} t_{k}] + 2 y_{i} (1) y_{j} (1) \bbE [t_{i} t_{j}] \right) \\
& + \frac{1}{n^{2}} \frac{1}{\pi_{1}^{2}} \left( \sum_{k \neq i \neq j} y_{k} (1)^{2} \bbE [t_{i} t_{j} t_{k}] + (y_{i} (1)^{2} + y_{j} (1)^{2}) \bbE [t_{i} t_{j}] \right) \\
= & \ \frac{1}{n^{2}} \frac{1}{\pi_{1}} \sum_{k \neq l \neq i \neq j} y_{k} (1) y_{l} (1) \left( \pi_{1} - \frac{\pi_{0}}{n - 1} \right) \left( \pi_{1} - \frac{2 \pi_{0}}{n - 2} \right) \left( \pi_{1} - \frac{3 \pi_{0}}{n - 3} \right) \\
& + \frac{1}{n^{2}} \frac{1}{\pi_{1}} \sum_{k \neq i \neq j} \{2 y_{k} (1) (y_{i} (1) + y_{j} (1)) + y_{k} (1)^{2}\} \left( \pi_{1} - \frac{\pi_{0}}{n - 1} \right) \left( \pi_{1} - \frac{2 \pi_{0}}{n - 2} \right) \\
& + \frac{1}{n^{2}} \frac{1}{\pi_{1}} (y_{i} (1) + y_{j} (1))^{2} \left( \pi_{1} - \frac{\pi_{0}}{n - 1} \right) \\
= & \ \frac{1}{\pi_{1}} \left( \bar{\tau}^{2} - 2 \bar{\tau} \frac{y_{i} (1) + y_{j} (1)}{n} + \frac{2 (y_{i} (1)^{2} + y_{i} (1) y_{j} (1) + y_{j} (1)^{2})}{n^{2}} - \frac{1}{n^{2}} \sum_{k = 1}^{n} y_{k} (1)^{2} \right) \\
& \times \left( \pi_{1} - \frac{\pi_{0}}{n - 1} \right) \left( \pi_{1} - \frac{2 \pi_{0}}{n - 2} \right) \left( \pi_{1} - \frac{3 \pi_{0}}{n - 3} \right) \\
& + \frac{1}{\pi_{1}} \left\{ 2 \bar{\tau} \frac{y_{i} (1) + y_{j} (1)}{n} - \frac{2 (y_{i} (1)^{2} + y_{i} (1) y_{j} (1) + y_{j} (1)^{2}) + (y_{i} (1) + y_{j} (1))^{2}}{n^{2}} + \frac{1}{n^{2}} \sum_{k = 1}^{n} y_{k} (1)^{2} \right\} \\
& \times \left( \pi_{1} - \frac{\pi_{0}}{n - 1} \right) \left( \pi_{1} - \frac{2 \pi_{0}}{n - 2} \right) \\
& + \frac{1}{\pi_{1}} \left( \frac{y_{i} (1) + y_{j} (1)}{n} \right)^{2} \left( \pi_{1} - \frac{\pi_{0}}{n - 1} \right) \\
= & \ \frac{1}{\pi_{1}} \left( \pi_{1} - \frac{\pi_{0}}{n - 1} \right) \left( \pi_{1} - \frac{2 \pi_{0}}{n - 2} \right) \left( \pi_{1} - \frac{3 \pi_{0}}{n - 3} \right) \bar{\tau}^{2}  + \frac{\pi_{0}}{\pi_{1}} \left( \pi_{1} - \frac{\pi_{0}}{n - 1} \right) \left( \pi_{1} - \frac{2 \pi_{0}}{n - 2} \right) \frac{1}{n - 3} \bar{\tau}^{(2)}\\
& + \frac{\pi_{0}}{\pi_{1}} \left( \pi_{1} - \frac{\pi_{0}}{n - 1} \right) \left( \pi_{1} - \frac{2 \pi_{0}}{n - 2} \right) \frac{2}{n - 3} \bar{\tau} \left( y_{i} (1) + y_{j} (1) \right)\\
& - \frac{\pi_{0}}{\pi_{1}} \left( \pi_{1} - \frac{\pi_{0}}{n - 1} \right) \left( \pi_{1} - \frac{2 \pi_{0}}{n - 2} \right) \frac{2}{n (n - 3)} (y_{i} (1)^{2} + y_{i} (1) y_{j} (1) + y_{j} (1)^{2}) \\
& + \frac{\pi_{0}}{\pi_{1}} \left( \pi_{1} - \frac{\pi_{0}}{n - 1} \right) \frac{(y_{i} (1) + y_{j} (1))^{2}}{n (n - 2)}.
\end{align*}

\begin{align*}
& \ \bbE [t_{i} t_{j} t_{k} \hat{\tau}_{\unadj}^{2}] = \frac{1}{n^{2} \pi_{1}^{2}} \sum_{l = 1}^{n} \sum_{m = 1}^{n} y_{l} (1) y_{m} (1) \bbE \left[ t_{i} t_{j} t_{k} t_{l} t_{m} \right] \\
= & \ \frac{1}{n^{2} \pi_{1}^{2}} \sum_{1 \leq l \neq m \leq n} y_{l} (1) y_{m} (1) \bbE [t_{i} t_{j} t_{k} t_{l} t_{m}] + \frac{1}{n^{2} \pi_{1}^{2}} \sum_{l = 1}^{n} y_{l} (1)^{2} \bbE [t_{i} t_{j} t_{k} t_{l}] \\
= & \ \frac{1}{n^{2} \pi_{1}} \sum_{l \neq m \neq i \neq j \neq k} y_{l} (1) y_{m} (1) \left( \pi_{1} - \frac{\pi_{0}}{n - 1} \right) \left( \pi_{1} - \frac{2 \pi_{0}}{n - 2} \right) \left( \pi_{1} - \frac{3 \pi_{0}}{n - 3} \right) \left( \pi_{1} - \frac{4 \pi_{0}}{n - 4} \right) \\
& + \frac{2}{n^{2} \pi_{1}} \sum_{l \neq i \neq j \neq k} y_{l} (1) (y_{i} (1) + y_{j} (1) + y_{k} (1)) \left( \pi_{1} - \frac{\pi_{0}}{n - 1} \right) \left( \pi_{1} - \frac{2 \pi_{0}}{n - 2} \right) \left( \pi_{1} - \frac{3 \pi_{0}}{n - 3} \right) \\
& + \frac{2}{n^{2} \pi_{1}} (y_{i} (1) y_{j} (1) + y_{i} (1) y_{k} (1) + y_{j} (1) y_{k} (1)) \left( \pi_{1} - \frac{\pi_{0}}{n - 1} \right) \left( \pi_{1} - \frac{2 \pi_{0}}{n - 2} \right) \\
& + \frac{1}{n^{2} \pi_{1}} \sum_{l \neq i \neq j \neq k} y_{l} (1)^{2} \left( \pi_{1} - \frac{\pi_{0}}{n - 1} \right) \left( \pi_{1} - \frac{2 \pi_{0}}{n - 2} \right) \left( \pi_{1} - \frac{3 \pi_{0}}{n - 3} \right) \\
& + \frac{1}{n^{2} \pi_{1}} (y_{i} (1)^{2} + y_{j} (1)^{2} + y_{k} (1)^{2}) \left( \pi_{1} - \frac{\pi_{0}}{n - 1} \right) \left( \pi_{1} - \frac{2 \pi_{0}}{n - 2} \right) \\
= & \ \frac{1}{\pi_{1}} \left( \pi_{1} - \frac{\pi_{0}}{n - 1} \right) \left( \pi_{1} - \frac{2 \pi_{0}}{n - 2} \right) \left( \pi_{1} - \frac{3 \pi_{0}}{n - 3} \right) \left( \pi_{1} - \frac{4 \pi_{0}}{n - 4} \right) \left\{ \begin{array}{c}
\bar{\tau}^{2} - 2 \bar{\tau} \frac{y_{i} (1) + y_{j} (1) + y_{k} (1)}{n} \\
+ \, \frac{2}{n^{2}} \sum\limits_{r, s \in \{i, j, k\}} y_{r} (1) y_{s} (1)\\
- \, \frac{1}{n^{2}} \sum\limits_{l = 1}^{n} y_{l} (1)^{2}
\end{array} \right\} \\
& + \frac{1}{\pi_{1}} \left( \pi_{1} - \frac{\pi_{0}}{n - 1} \right) \left( \pi_{1} - \frac{2 \pi_{0}}{n - 2} \right) \left( \pi_{1} - \frac{3 \pi_{0}}{n - 3} \right) \left\{ \begin{array}{c} 
2 \bar{\tau} \cdot \frac{y_{i} (1) + y_{j} (1) + y_{k} (1)}{n} - 2 \frac{(y_{i} (1) + y_{j} (1) + y_{k} (1))^{2}}{n^{2}} \\
- \, \frac{y_{i} (1)^{2} + y_{j} (1)^{2} + y_{k} (1)^{2}}{n^{2}} + \frac{1}{n^{2}} \sum\limits_{l = 1}^{n} y_{l} (1)^{2}
\end{array} \right\} \\
& + \frac{1}{\pi_{1}} \left( \pi_{1} - \frac{\pi_{0}}{n - 1} \right) \left( \pi_{1} - \frac{2 \pi_{0}}{n - 2} \right) \left( \frac{y_{i} (1) + y_{j} (1) + y_{k} (1)}{n} \right)^{2} \\
= & \ \frac{1}{\pi_{1}} \left( \pi_{1} - \frac{\pi_{0}}{n - 1} \right) \left( \pi_{1} - \frac{2 \pi_{0}}{n - 2} \right) \left( \pi_{1} - \frac{3 \pi_{0}}{n - 3} \right) \left( \pi_{1} - \frac{4 \pi_{0}}{n - 4} \right) \bar{\tau}^{2} \\
& + \frac{\pi_{0}}{\pi_{1}} \left( \pi_{1} - \frac{\pi_{0}}{n - 1} \right) \left( \pi_{1} - \frac{2 \pi_{0}}{n - 2} \right) \left( \pi_{1} - \frac{3 \pi_{0}}{n - 3} \right) \frac{1}{n - 4} \bar{\tau}^{(2)} \\
& + \frac{\pi_{0}}{\pi_{1}} \left( \pi_{1} - \frac{\pi_{0}}{n - 1} \right) \left( \pi_{1} - \frac{2 \pi_{0}}{n - 2} \right) \left( \pi_{1} - \frac{3 \pi_{0}}{n - 3} \right) \frac{2}{n - 4}  \bar{\tau} ( y_{i} (1) + y_{j} (1) + y_{k} (1) ) \\
& - \frac{\pi_{0}}{\pi_{1}} \left( \pi_{1} - \frac{\pi_{0}}{n - 1} \right) \left( \pi_{1} - \frac{2 \pi_{0}}{n - 2} \right) \left( \pi_{1} - \frac{3 \pi_{0}}{n - 3} \right) \frac{2}{n (n - 4)} \sum\limits_{r, s \in \{i, j, k\}} y_{r} (1) y_{s} (1) \\
& + \frac{\pi_{0}}{\pi_{1}} \left( \pi_{1} - \frac{\pi_{0}}{n - 1} \right) \left( \pi_{1} - \frac{2 \pi_{0}}{n - 2} \right) \frac{1}{n (n - 3)} (y_{i} (1) + y_{j} (1) + y_{k} (1))^{2}.
\end{align*}

\end{proof}

\section{Realistic Simulations}
\label{app:sim}
In this section, parallel to the data generating process in Section \ref{sec:simulations}, we conduct realistic simulations by actually drawing treatment assignments based on CRE. Once $\{\bx_i, y_i(1)\}_{i = 1}^{n}$ are generated, they are fixed and random treatment assignments are drawn repeatedly from CRE, with the Monte Carlo repetition size $K = 20000$. 
Figures \ref{fig:main abs rela bias} to \ref{fig:main sd_infl_ratio} respectively report the empirical relative absolute bias, root mean squared error (RMSE), coverage probabilities by the associated large-sample nominal 95\% Wald confidence interval, average length of the confidence intervals and the standard deviation inflation ratio. Overall, the main message is very similar to the oracle simulations, but the differences between $\hat{\tau}_{\adj, 2}$ and $\hat{\tau}_{\adj, 2}^{\dag}$ in terms of the evaluation metrics are not as pronounced as in the theoretical relative efficiency comparison shown in Figure \ref{fig:oracle_var_main}.

    

The mentioned evaluation metrics are defined as follows.
\begin{itemize}
    \item empirical relative absolute bias: $|\frac{1}{K}\sum_{k=1}^{K}\hat{\tau}_{k} - \tau|/\sqrt{\sigma_{n}^2/n}$, where $\sigma_n^2$ is the theoretical asymptotic variance of $\hat{\tau}$.
    
    \item empirical root mean squared error (RMSE): $\sqrt{\frac{1}{K}\sum_{k=1}^{K}(\hat{\tau}_k - \tau)^2}$.
    
    \item empirical coverage probabilities:$\frac{1}{K}\sum_{k=1}^{K}\mathbb{I}\left(\hat{\tau}_k - 1.96\sqrt{\hat{\sigma}_{n,k}^2/n}\leq \tau \leq \hat{\tau}_k + 1.96\sqrt{\hat{\sigma}_{n,k}^2/n}\right)$, where $\hat{\sigma}_{n,k}^{2}$ is the  estimate of the  asymptotic variance of $\hat{\tau}$. For $\hat{\tau}_{\adj, 2},\hat{\tau}_{\adj, 3}$ and $\hat{\tau}_{\adj, 2}^{\dag}(\hat{\tau}_{\db})$, we considered both conservative ($\hat{\nu}^{'\sff}$, in dashed lines) and non-conservative but consistent when $p = O (n)$ ($\hat{\nu}^{\sff}$, in solid lines)  variance estimators defined in Appendix \ref{app:variance estimators}. From Figures \ref{fig:main coverage} -- \ref{fig:main length}, we observe that the large-sample nominal 95\% confidence intervals based on the conservative variance estimators in general have higher coverage probabilities (Figure \ref{fig:main coverage}) than the non-conservative ones with similar coverage lengths (Figure \ref{fig:main length}).
    
    \item empirical average length of the confidence intervals: $2\times 1.96 \sqrt{\hat{\sigma}_{n,k}^2/n}$. It is evident that the average lengths increase as $p$ increases from Figure \ref{fig:main length}. Also, we can observe that, in the upper left panels of Figure \ref{fig:main length}, the average lengths of conservative confidence intervals can still be smaller than those of the confidence intervals associated with the unadjusted estimators. We encourage interested readers to refer to Appendix \ref{app:sim small p} for a more detailed comparison in this regard when $p$ is small compared to $n$.
    
    \item standard deviation inflation ratio: $\frac{1}{K}\sum_{k=1}^{K}\hat{\sigma}_{n,k} / \hat{\sigma}_{n,*}$, where $\hat{\sigma}_{n,*}^2\coloneqq \frac{n}{K-1}\sum_{k=1}^{K}\left(\hat{\tau}_k - \bar{\hat{\tau}}\right)^2$ is the unbiased estimate of the true sampling variance of $\sqrt{n}\hat{\tau}$. Again, it is interesting to note that the estimated variances are generally closer to the Monte Carlo variances as $p$ grows, from Figure \ref{fig:main sd_infl_ratio}.
\end{itemize}

\begin{figure}[H]
    \centering
    \includegraphics[width=\linewidth,page=1]{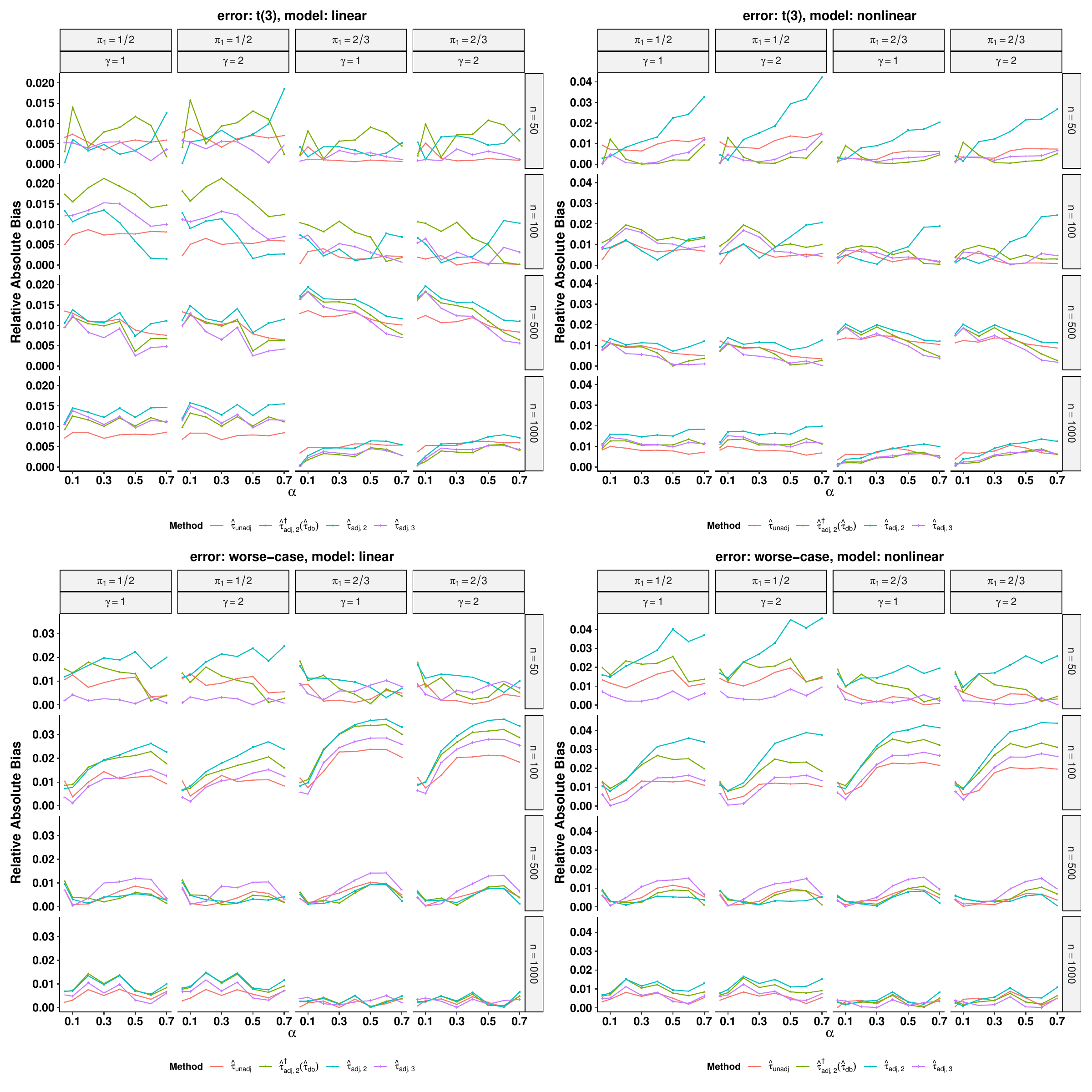}
    \caption{Empirical relative absolute bias for different $\pi_1$, $\gamma$, $\alpha$ and $n$ under the independent t error and the worst-case error. }
    \label{fig:main abs rela bias}
\end{figure}

\begin{figure}[H]
    \centering
    \includegraphics[width=\linewidth,page=2]{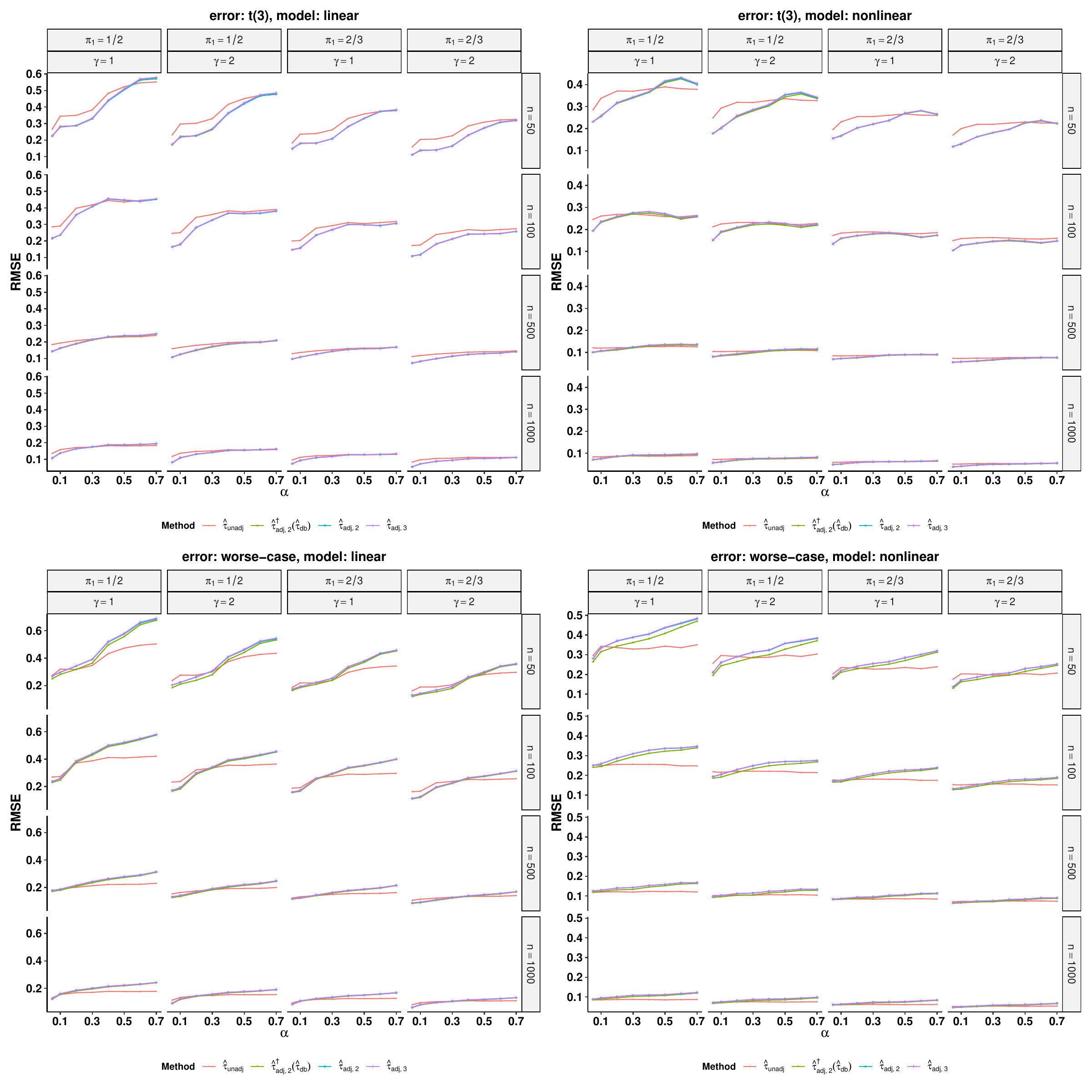}
    \caption{RMSE for different $\pi_1$, $\gamma$, $\alpha$ and $n$ under the independent t error and the worst-case error. }
    \label{fig:main rmse}
\end{figure}

\begin{figure}[H]
    \centering
    \includegraphics[width=\linewidth,page=4]{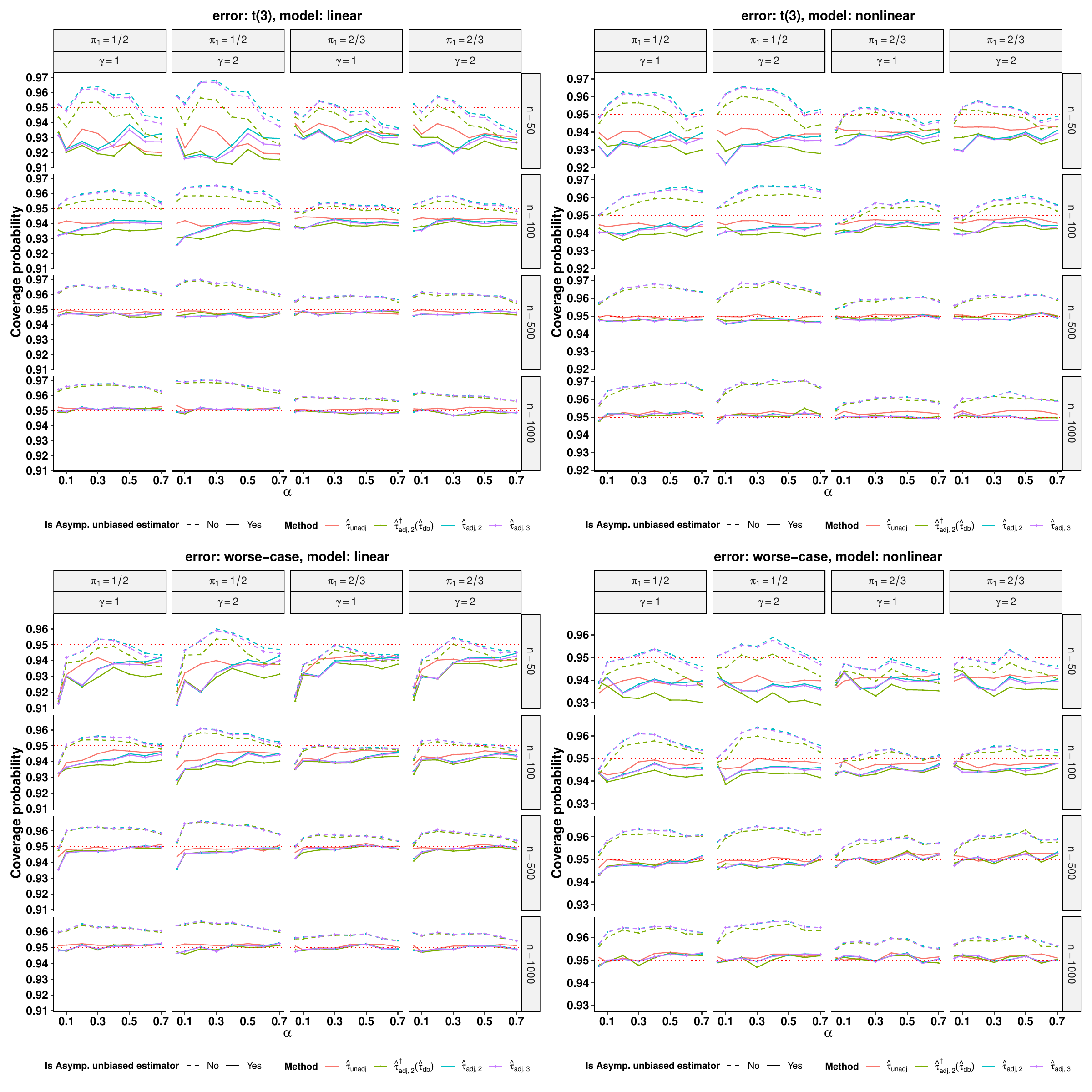}
    \caption{Empirical coverage probabilities for different  $\pi_1$, $\gamma$, $\alpha$ and $n$ under the independent t error and the worst-case error. }
    \label{fig:main coverage}
\end{figure}

\begin{figure}[H]
    \centering
    \includegraphics[width=\linewidth,page=5]{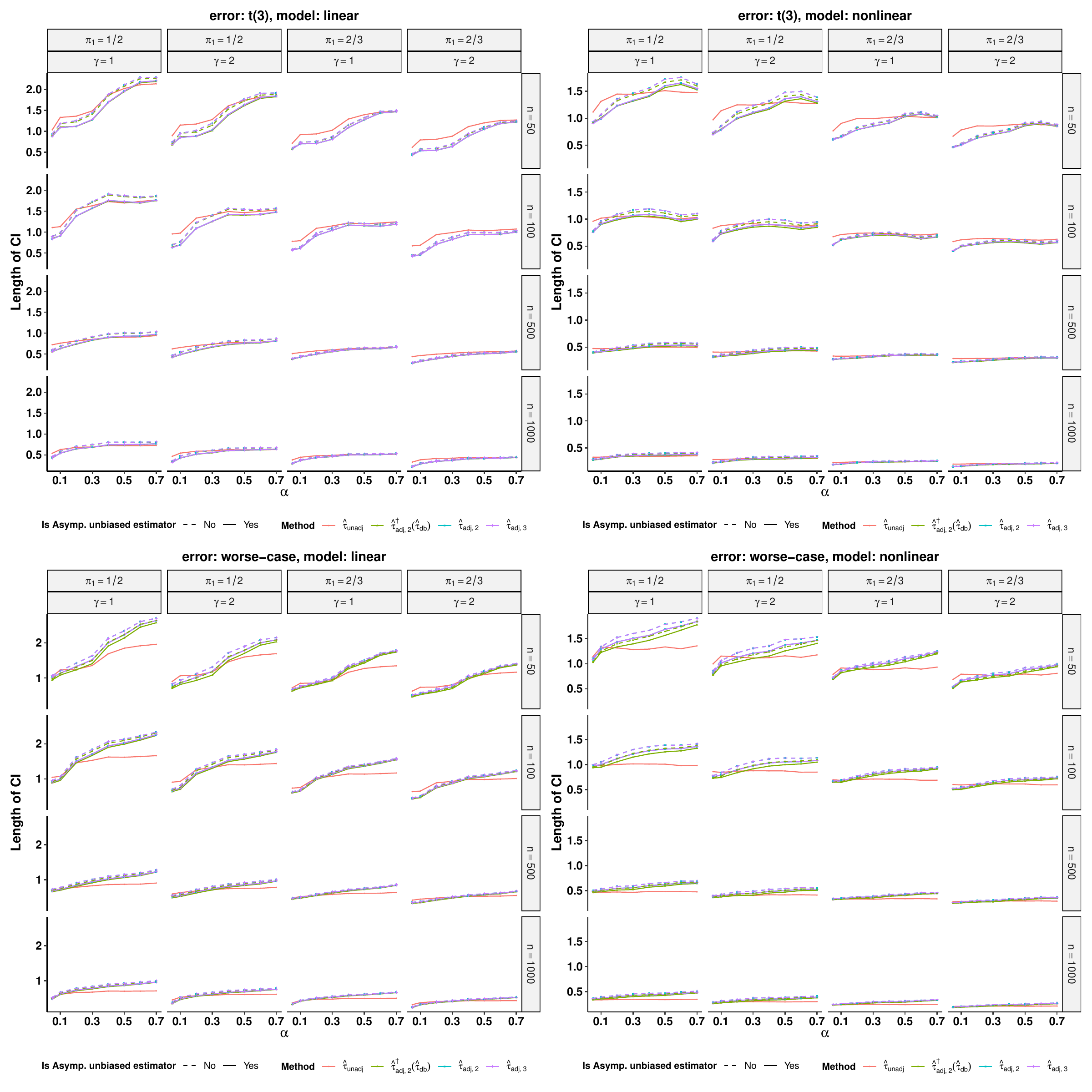}
    \caption{Empirical average length of the confidence intervals  for different $\pi_1$, $\gamma$, $\alpha$ and $n$ under the independent t error and the worst-case error. }
    \label{fig:main length}
\end{figure}

\begin{figure}[H]
    \centering
    \includegraphics[width=\linewidth,page=6]{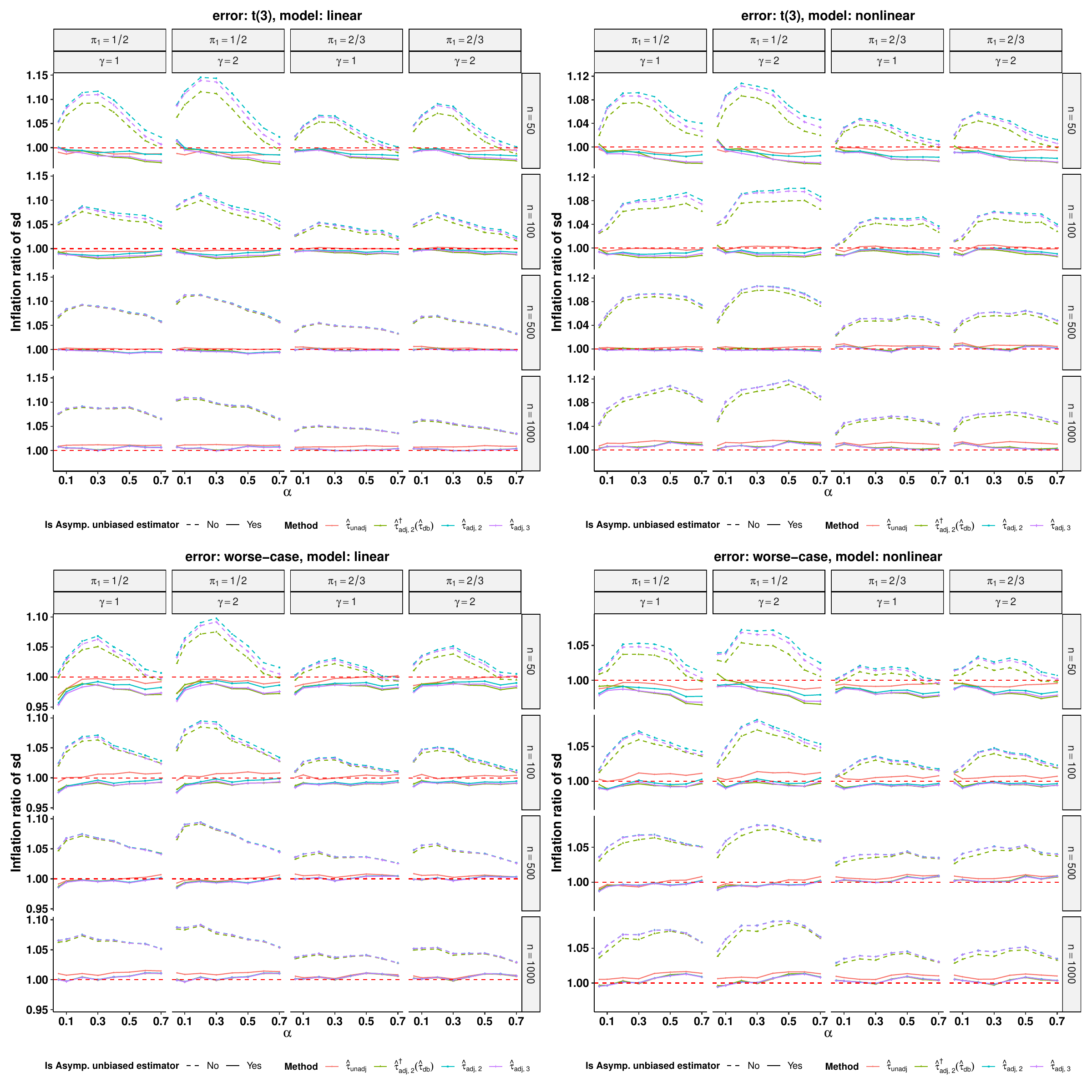}
    \caption{Empirical standard deviation inflation ratio for different $\pi_1$, $\gamma$, $\alpha$ and $n$ under the independent t error and the worst-case error. }
    \label{fig:main sd_infl_ratio}
\end{figure}

\section{More Simulation Experiments When \texorpdfstring{$p$}{p} is Small Compared to \texorpdfstring{$n$}{n}}
\label{app:sim small p}

In this section, we focus on the case when $p/n$ is small. We use the same covariates data $\mathcal{X}\in \bbR^{N\times N}$ and exogenous error $\epsilon_0$ as in Section \ref{sec:simulations}.
Parallel to the data generating process in Section \ref{sec:simulations}, we vary $p \in \{4,6,8,10,\cdots,20\}$, and $ n\in\{50,100,200,500\}$, while other scenarios are the same. Once $\{\bx_i, y_i(1)\}_{i = 1}^{n}$ are generated, they are fixed and random treatment assignments are drawn repeatedly from CRE, with the Monte Carlo repetition size $K = 20000$.

Figures \ref{fig:main abs rela bias small p} to \ref{fig:main sd_infl_ratio small p} respectively report the empirical relative absolute bias, RMSE, coverage probabilities by the associated large-sample nominal 95\% Wald confidence interval, average length of the confidence intervals, and the standard deviation inflation ratio. 
Overall, the main message is that even conservative variance estimators outperform $\hat{\tau}_{\unadj}$ in terms of these  evaluation metrics, particularly regarding the coverage probabilities in Figure \ref{fig:main coverage small p} and the average length of the confidence intervals in Figure \ref{fig:main length small p}. Notably, we do not use any conservative variance estimator for $\hat{\tau}_{\unadj}$; this explains why the nominal large-sample 95\% confidence interval centered around $\hat{\tau}_{\unadj}$ may  under-cover when the sample size is small. Since the conservative confidence intervals centered around the HOIF-motivated adjusted estimators are narrower on average than those centered around $\hat{\tau}_{\unadj}$ based on non-conservative variance estimators, the intervals centered around  $\hat{\tau}_{\unadj}$  using any conservative variance estimators should only be even wider.

\begin{figure}[H]
    \centering
    \includegraphics[width=\linewidth,page=1]{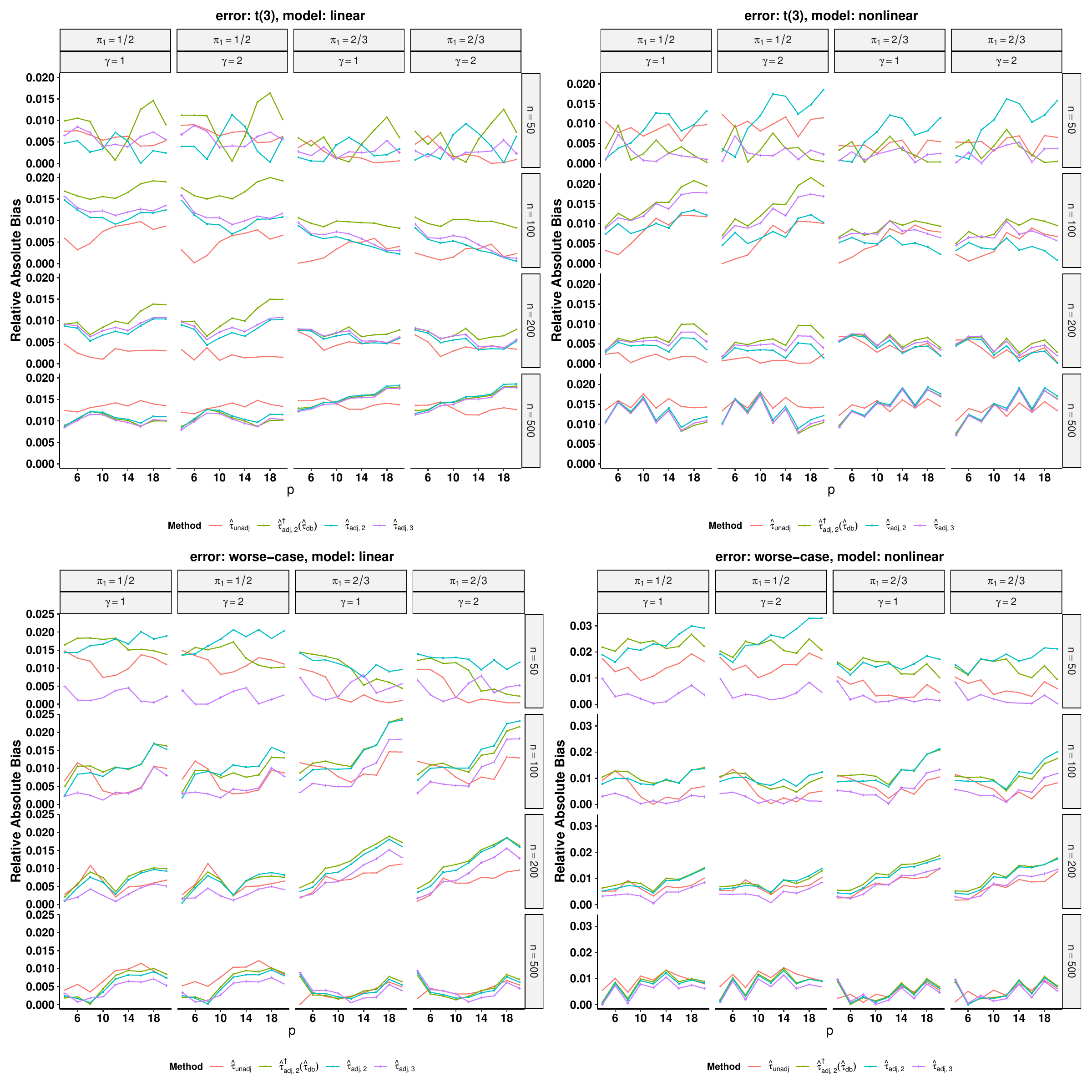}
    \caption{Empirical relative absolute bias for different $\pi_1$, $\gamma$, $\alpha$ and $n$ under the independent t error and the worst-case error. }
    \label{fig:main abs rela bias small p}
\end{figure}

\begin{figure}[H]
    \centering
    \includegraphics[width=\linewidth,page=2]{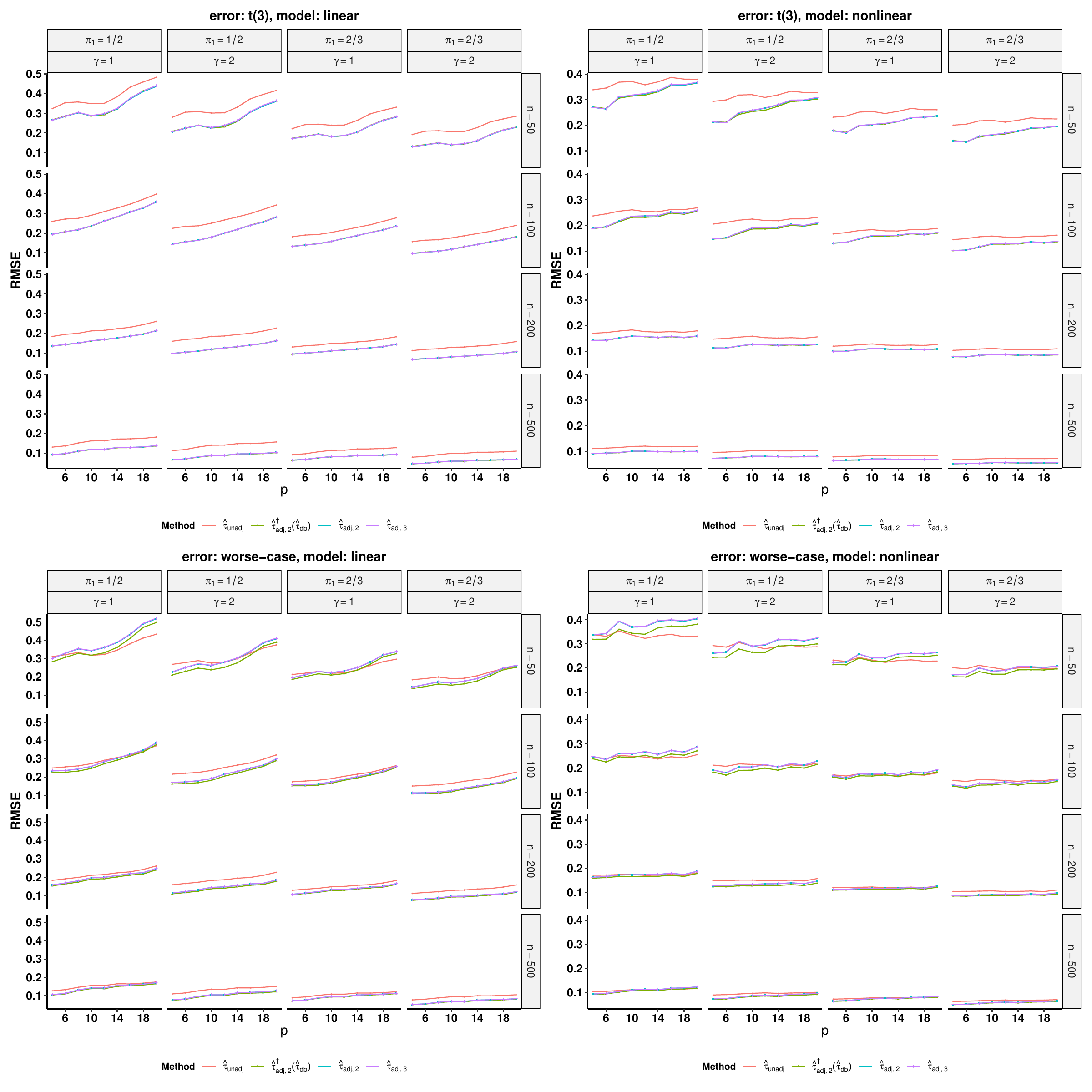}
    \caption{RMSE for different $\pi_1$, $\gamma$, $\alpha$ and $n$ under the independent t error and the worst-case error. }
    \label{fig:main rmse small p}
\end{figure}

\begin{figure}[H]
    \centering
    \includegraphics[width=\linewidth,page=4]{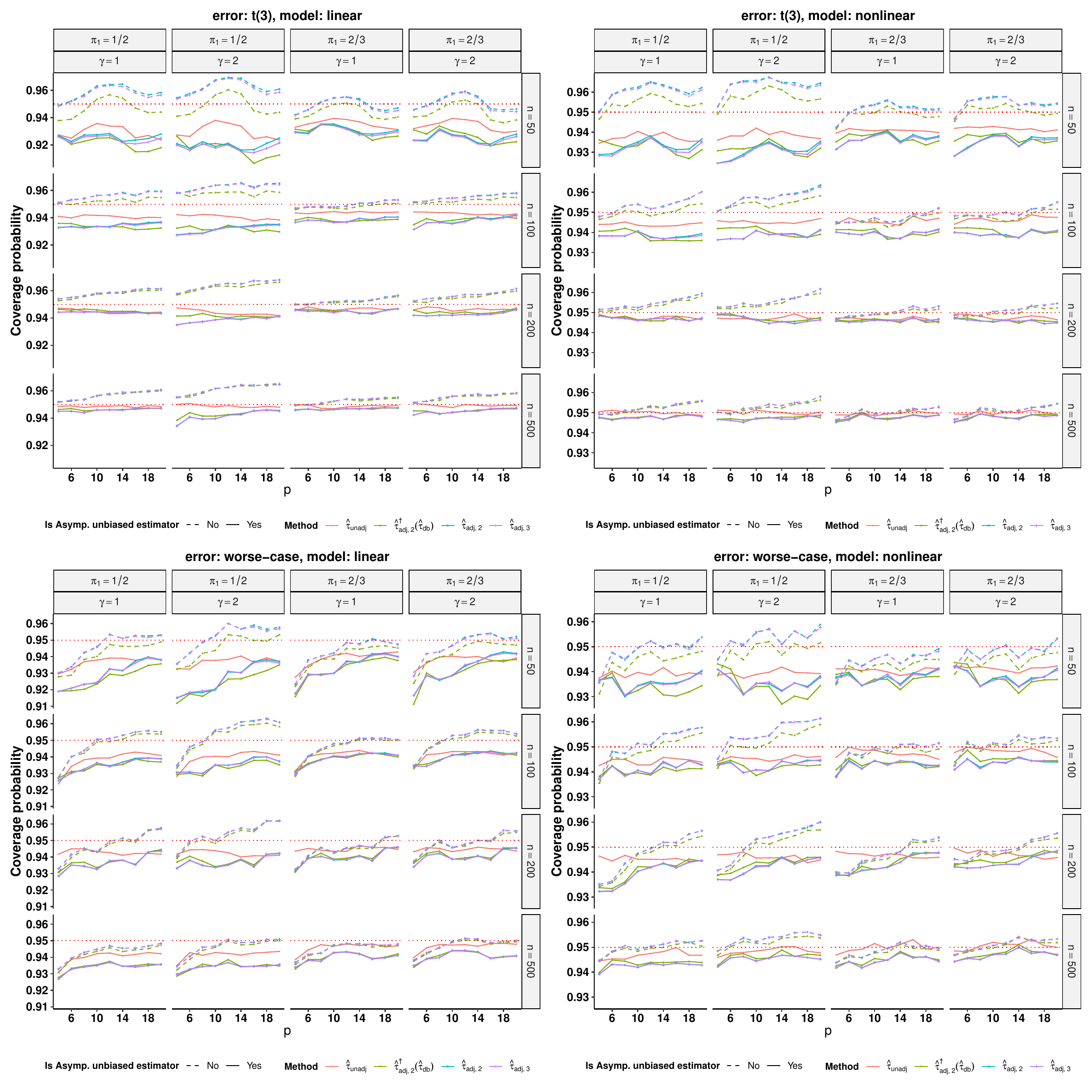}
    \caption{Empirical coverage probabilities for different  $\pi_1$, $\gamma$, $\alpha$ and $n$ under the independent t error and the worst-case error. }
    \label{fig:main coverage small p}
\end{figure}

\begin{figure}[H]
    \centering
    \includegraphics[width=\linewidth,page=5]{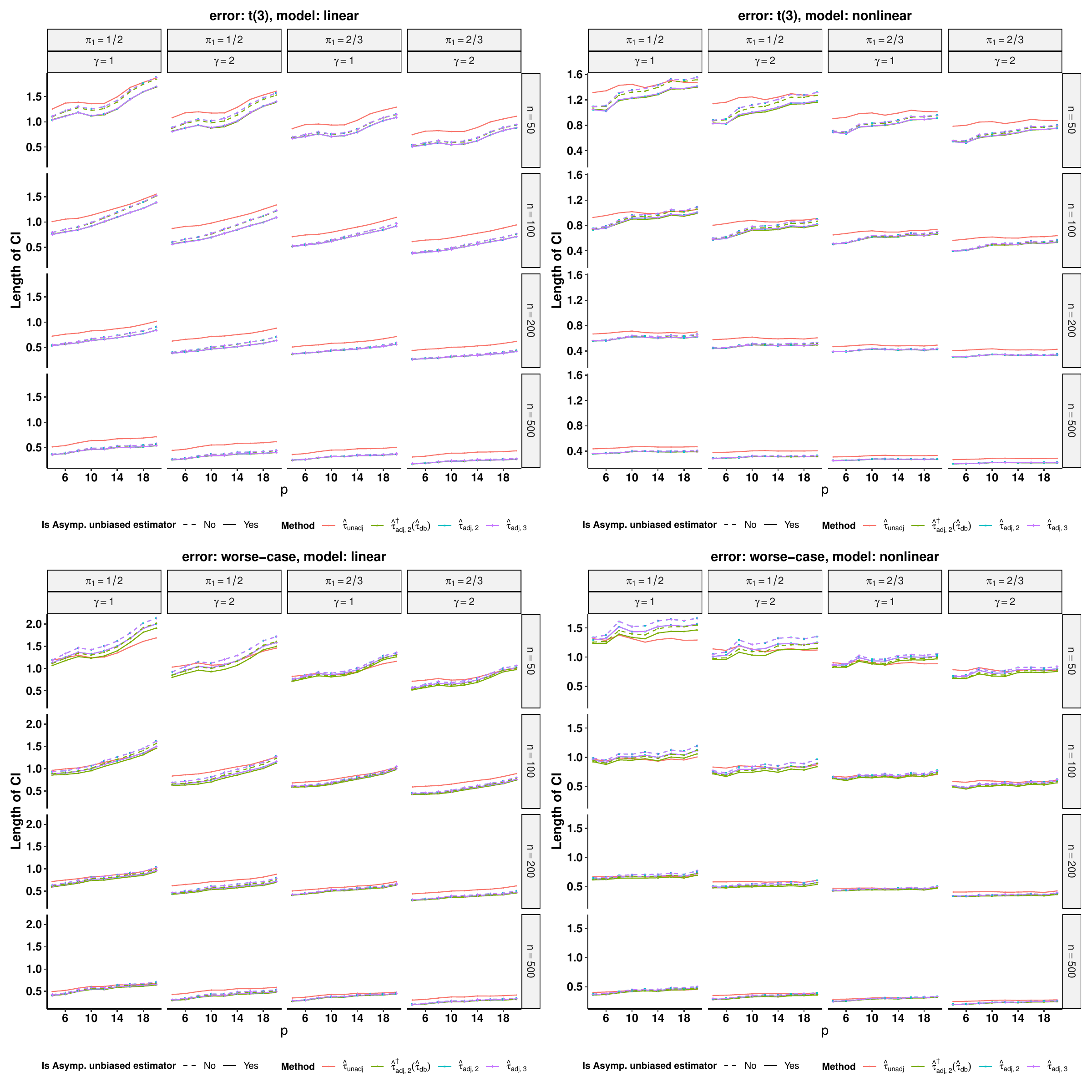}
    \caption{Empirical average length of the confidence intervals for different $\pi_1$, $\gamma$, $\alpha$ and $n$ under the independent t error and the worst-case error. }
    \label{fig:main length small p}
\end{figure}

\begin{figure}[H]
    \centering
    \includegraphics[width=\linewidth,page=6]{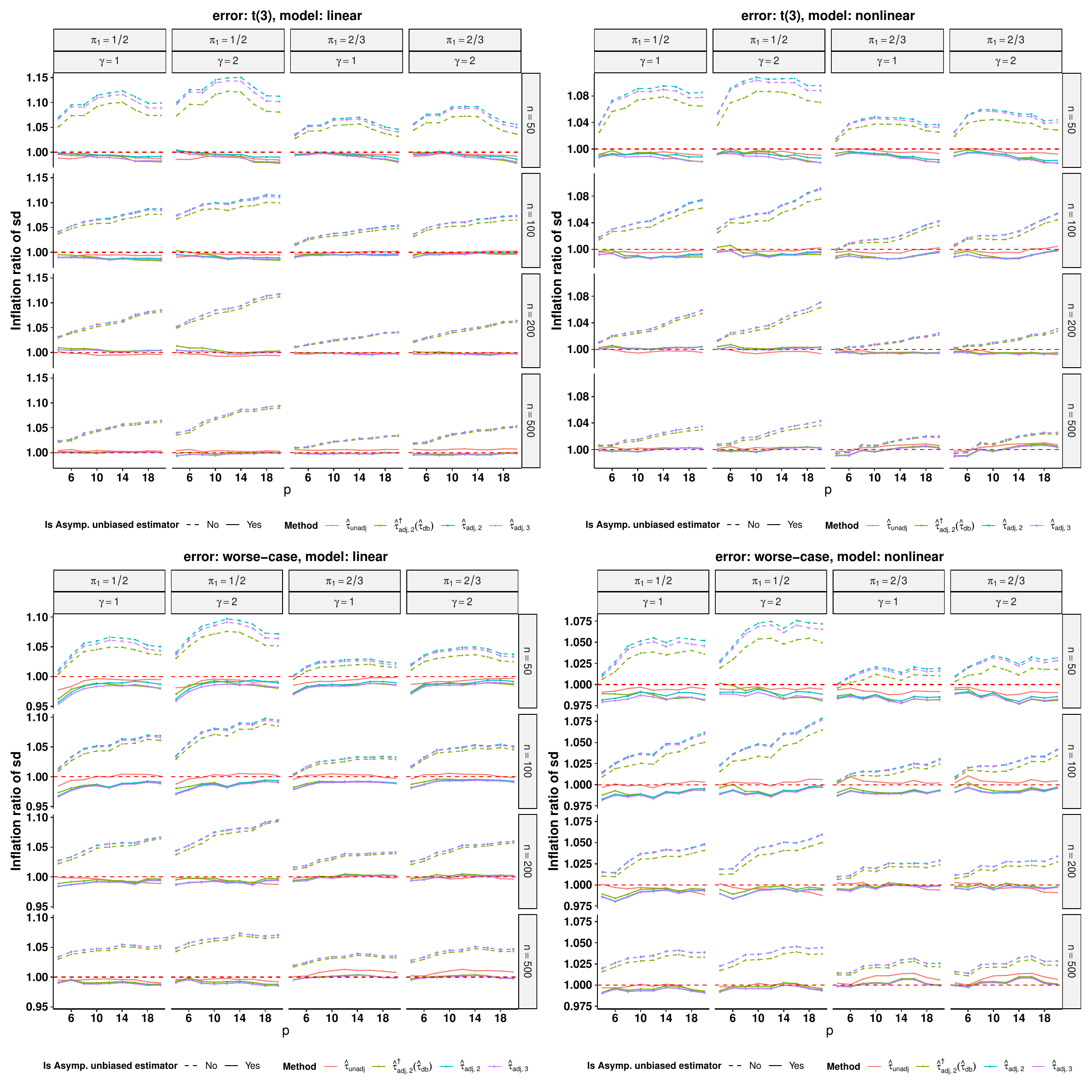}
    \caption{Empirical standard deviation inflation ratio for different $\pi_1$, $\gamma$, $\alpha$ and $n$ under the independent t error and the worst-case error. }
    \label{fig:main sd_infl_ratio small p}
\end{figure}

\putbib[Master_appendix]

\end{bibunit}
\end{appendices}


\end{document}